\documentclass{ipi}
  \usepackage{amsmath,amssymb,graphicx,amsthm,mathrsfs}
  \usepackage{paralist}
  \usepackage{etex}
  \usepackage{graphics} 
  \usepackage{epsfig} 

  \usepackage{hyperref} 

\def\currentvolume{X}
 \def\currentissue{X}
  \def\currentyear{200X}
   
    \def\ppages{X--XX}

\newtheorem{theorem}{Theorem}[section]

\newtheorem{lemma}[theorem]{Lemma}

\theoremstyle{definition}

\newtheorem{remark}{Remark}

\title[DSM for Inverse Scattering Using Far-Field Data]
      {A Direct Sampling Method \\
 for Inverse Scattering Using Far-Field Data}

\author[Jingzhi Li and Jun Zou]{}



\thanks{The first author was supported by the NSF of China. (No.
11201453 and 91130022). The second author was substantially
supported by Hong Kong RGC grants (Projects 405110 and 404611).}

\begin{document}
\maketitle

\centerline{\scshape Jingzhi Li}
\medskip
{\footnotesize
 \centerline{Faculty of Science,
South University of Science and Technology of China}
   \centerline{ Shenzhen, 518055, P. R.
China}
} 

\medskip

\centerline{\scshape Jun Zou}
\medskip
{\footnotesize
 \centerline{Department of Mathematics, The
Chinese University of Hong Kong}
   \centerline{Shatin, Hong Kong}
} 

\bigskip

 \centerline{(Communicated by the associate editor name)}

\begin{abstract}
This work is concerned with a direct sampling method (DSM)
for inverse acoustic scattering problems using far-field data.
Using one or few incident waves, the DSM provides quite
reasonable profiles of scatterers in time-harmonic inverse acoustic
scattering without a priori knowledge of either the physical properties
or the number of disconnected components of the scatterer. We shall
first present a novel derivation of the DSM using far-field data,
then carry out a systematic evaluation of the performances and distinctions
of the DSM using both near-field and far-field data. A new
interpretation from the physical perspective is provided based on some
numerical observations.
It is shown from a variety of numerical experiments that the method
has several interesting and promising potentials:
a) ability to identify not only medium scatterers, but also obstacles,
and even cracks, using measurement data from one or few incident directions,
b) robustness with respect to large noise, and c) computational efficiency
with only inner products involved.

\smallskip
\noindent {\bf MSC (AMS 2010)}: 35R30, 41A27, 78A46
\end{abstract}


\section{Introduction}

The capability of effectively retrieving the location and/or
geometrical features of unknown scatterers from the knowledge of the
scattered wave field (near-field or far-field data) is of paramount
importance in many practical applications such as underground mine
detection, geophysical exploration in oil industry, detection of
defects or cracks in nondestructive testing, target detection using
radar or sonar systems and ultrasound imaging in biomedical
equipments \cite{xu_siam,CoK98}. Both locations and geometrical
shapes of scatterers are often the final target in the applications.
Many quantitative methods have been developed for such purposes
\cite{AmK04,BaL05,CoK98,Hoh01,IJZ12r,Som06,BBA99}. But there is a
crucial step for these imaging processes, to find some reasonable
estimated locations and shapes as their initial computational sampling 
domains, otherwise the imaging processes may not work or work with huge computational
efforts.

A promising direct sampling method (DSM) was proposed 
recently in \cite{IJZ12} to recover the inhomogeneous medium scatterers
from near-field measurement data for inverse acoustic medium scattering problems.
Compared with other sampling-type schemes (see, e.g., \cite{Che01,KiG10,Pot06}
for the detailed surveys) such as multiple signal classification
(MUSIC) \cite{ChZ09,Dev99,GrD04,HSZ06,Sch86} and linear sampling
method (LSM) \cite{CCM11,CoK96,Kir98}, a distinct feature of the
DSM is its capability of depicting the profile of medium scatterers
using one or very few incident wave fields, through the computationally
very cheap inner products of the measured data and the fundamental
solution in the homogeneous background medium.

In this work, we will further explore in this direction by
generalizing the DSM from the knowledge of near-field data to
far-field data of the scattered waves for inverse medium and obstacle 
scattering problems.
Extensive numerical experiments are done to investigate
systematically the behavior of the DSM. An important new observation
is that the DSM is able to recover not only medium scatterers, but
also obstacles, and even cracks using one or two incident waves.

The full inverse scattering problem (ISP) is to determine
both the shapes and the physical profiles of the scatterers, hence
may require the measured data at many frequencies, not the one from
a single frequency as considered in this work.
It is well known that the ISP is highly nonlinear, and the first key step
for the numerical treatment of a nonlinear problem is to work out
a good initial guess for its approximate solution.
As DSM is computationally very cheap and works with the
data from only one or few incident fields, it can naturally serve as a fast, simple
and effective alternative to existing numerical tools for locating
a reliable approximate position of the unknown scatterers, which can
then serve as a good initial guess in any existing method (see, e.g., \cite{BaL05,CoK98,Hoh01,IJZ12r})
for achieving a more accurate estimate of the scatterer support and the inhomogeneity
distribution. 
As we shall see numerically, the DSM can provide a very
reliable location of each individual scatterer component so that one
may start with a much smaller sampling region in a more accurate but
computationally more demanding method.
The reduction of the sizes of the initial sampling regions for unknown
scatterers may save us an essential fraction of the computational
efforts in the entire reconstruction process by most existing methods.
Other approaches are available in the literature, 
which do not require an a priori good initial guess of the scatterers, e.g., 
the globally convergent numerical method. 
We refer to the monograph
\cite{BeK12} and many references therein for the detailed description of the method and 
its theory for ISPs 
corresponding to some hyperbolic equations. 
Because of the time dependence, the ISPs of \cite{BeK12} can be
considered as the ones with the data given at many frequencies.

Inverse scattering reconstruction methods
using one single incident wave, or simply \emph{one-shot} methods, date back to a long-standing open
problem in the inverse scattering  community, namely whether one can determine
the unknown scatterer by using the near-field/far-field data from only one or
several incident waves. Theoretically, the results on
unique identifiability have been understood only partially, e.g.,  for some
special classes of scatterers \cite{CoK98,CoS83}, or for the class of
polygonal or polyhedral scatterers
\cite{AlR05,ElY08,LYZ07,LZZ09,LiZ06}.
Armed with the aforementioned uniqueness results,
{one-shot} methods have been studied widely in the past few years,
e.g., in \cite{Pot10,Gri11,IJZ12,Han12}.
This work extends the one in \cite{IJZ12}  to make a more systematical
investigation of such a
promising one-shot method applied to inverse scattering problems using near-field or far-field data
for many possible cases of scatterers: obstacles, inhomogeneous media and cracks  or their
combinations. This study reveals  more potentials of the one-shot method
and provides some physical hints to answer the pending open problem from the numerical perspective.
We emphasize that the direct sampling method in the current work is based on
the same indicator function as that in the orthogonality sampling method developed and studied
in \cite{Pot10,Gri11}, but we shall present a different approach and motivation
to derive this method and justify its effectiveness, in particular
removing the ``smallness'' assumption that is crucial
to the derivation of the method in \cite{Pot10,Gri11}.


The paper is organized as follows. In Section\ \ref{sec:review:isp}
a brief review of inverse scattering problems is presented, along
with some useful notations and identities. Section\ \ref{sec:dsm}
describes the mathematical motivation of the DSM using far-field data
and proposes a new indicator function.  Section\ \ref{sec:numerics}
provides extensive numerical simulations to evaluate the performance
of the DSM using near-field or far-field data from obstacles, media
and cracks. In addition, we shall provide
a physical interpretation of the DSM, and compare the major features
of the DSMs using near-field and far-field data.
Some concluding remarks will be given in Section\ \ref{sec:Conclusion}.

\section{A brief review of inverse scattering problems \label{sec:review:isp}}

In this section we shall briefly describe the time harmonic inverse
obstacle and medium scattering problem using near-field or far-field
measurements \cite{CoK83,CoK98}. Consider a homogeneous background
space $\mathbb{R}^{N}$ ($N=2,3$) that contains some scatterer components
such as obstacles or inhomogeneous media, or both, occupying a bounded
domain $D$. Let $u^{\mathrm{inc}}=\exp(\mathrm{i}kx\cdot d)$ be
an incident plane wave, with the incident direction $d\in\mathbb{S}^{N-1}$
and the wave number $k$, and $u=u^{\mathrm{inc}}+u^{s}$ be the total
field formed by the incident and scattered fields. Then the total
field $u$ induced by the obstacles satisfies the Helmholtz equation
\begin{equation}
\Delta u+k^{2}u=0\quad\mathrm{in}\quad\mathbb{R}^{N}\setminus D,\label{eq:iosp}
\end{equation}
 or induced by the inhomogeneous medium scatterers satisfies
\begin{equation}
\Delta u+k^{2}n^{2}(x)u=0\quad\mathrm{in}\quad\mathbb{R}^{N},\label{eq:imsp}
\end{equation}
 where $n(x)$ is the refractive index. To account for the absorbing medium, the
refractive index can be modeled by the complex form
\begin{equation}
n^{2}(x)=n_{1}(x)+\mathrm{i}n_{2}(x).\label{eq:absorbing:n}
\end{equation}
The models above describe not only time-harmonic acoustic wave propagation,
but also electromagnetic wave propagation in either the transverse
magnetic or transverse electric models \cite{CoK98,IJZ12}.

\begin{remark}
Note that the obstacle scatterer can be viewed as the limiting case
of the medium scatterer with vanishing or singular material
properties. For example, an acoustic sound-soft obstacle is a
limiting case of \eqref{eq:imsp} as $n_{2}\to+\infty$; see, e.g.,
\cite[Eq. (4.4)]{KOVW10} or \cite[Sect. 4]{LLS12}. One may also
refer to \cite{LS12} for the sound-hard limiting case for general
dimensions. Therefore, the indicator functions derived in the sequel
are applicable to both inverse obstacle and medium scattering
problems, although our derivations are based only on the medium
scattering case.
\end{remark}

In the rest of this section we introduce some basic
notation and fundamental functions that will be needed in the
subsequent discussions.
First, we define a coefficient function,
$\eta=\left(n^{2}-1\right)k^{2}$, which characterizes the inhomogeneity
of the concerned media and is supported in the scatterer
$D\subset\mathbb{R}^{d}$. Then we define function $I=\eta u$, which is called
the induced current by the inhomogeneous media. Let $G(x,y)$ be the fundamental solution
to the Helmholtz equation in the homogeneous background, that can be
represented by (cf.~\cite{CoK83,CoK98})
\begin{equation}
G(x,y)=\begin{cases}
{\displaystyle \text{\ensuremath{\frac{\mathrm{i}}{4}}}H_{0}^{(1)}(k|x-y|)} & \mbox{for} ~~N=2;\\
{\displaystyle \frac{1}{4\pi}\frac{\exp(\mathrm{i}k|x-y|)}{k|x-y|}} & \mbox{for} ~~N=3
\end{cases}\label{eq:G}
\end{equation}
where $H_{0}^{(1)}$ refers to the Hankel function
of the first kind and zeroth-order. With the help of the asymptotic
properties of the fundamental solution $G$ \cite[Eqs.~(2.14) and (3.63)]{CoK98},
we have for $N=2,3$ that
\begin{equation}
u^{s}(x)=\frac{\exp(\mathrm{i}k|x|)}{|x|^{(N-1)/2}}\left\{ u^{\infty}(\hat{x})+O(\frac{1}{|x|})\right\} \qquad\mbox{as}\quad\,|x|\to\infty\,,\label{eq:us:uinf:relation}
\end{equation}
with $\hat{x}=x/|x|\in\mathbb{S}^{N-1}$. The scattered
near-field $u^{s}$ and the far-field $u^{\infty}$ in (\ref{eq:us:uinf:relation})
has the following very convenient representations (cf.~\cite[Eqs.,(8.12)  and (8.27)]{CoK98}):
\begin{eqnarray}
u^{s}(x) & = & \int_{D}G(x,y)I(y)\,\mathrm{d}y\,,\label{eq:us}\\
u^{\infty}(\hat{x}) & = & \int_{D}G^{\infty}(\hat{x},y)I(y)\,\mathrm{d}y\,,\label{eq:uinf}
\end{eqnarray}
 where the far-field pattern associated with the fundamental solution
$G$ is given by
\begin{equation}
G^{\infty}(\hat{x},y)=\begin{cases}
{\displaystyle \frac{\exp(\mathrm{i}\pi/4)}{\sqrt{8k\pi}}\exp(-\mathrm{i}k\hat{x}\cdot y)}, & N=2;\\
\\
{\displaystyle \frac{1}{4\pi}\exp(-\mathrm{i}k\hat{x}\cdot y)}, & N=3.
\end{cases}\label{eq:Ginf}
\end{equation}

\section{Direct sampling method using far-field data\label{sec:dsm}}

In this section we shall derive the DSM using far-field data.
Like many other sampling-type methods such as LSM, MUSIC and factorization
methods \cite{CoK96,KiG10,Sch86}, the essence of the DSM is to construct
an indicator function which has significantly different behaviors
inside and outside the scatterers, e.g., the indicator function of
the LSM blows up outside the scatterer but remains finite within the
scatterer. Let $\Gamma$ be the surface where the near-field data
is measured, $\Omega$ be a domain contained inside $\Gamma$ such
that the scatterer $D$ lies in $\Omega$. For any sampling point
$x_{p}\in\Omega$, the indicator function of the DSM using near-field
data is given by (cf.\,\cite{IJZ12})
\begin{equation}
\Phi(x_{p})=\frac{\left|\left\langle u^{s},\, G(\cdot,x_{p})\right\rangle _{L^{2}(\Gamma)}\right|}{\left\Vert u^{s}\right\Vert _{L^{2}(\Gamma)}\left\Vert G(\cdot,x_{p})\right\Vert _{L^{2}(\Gamma)}}.\label{eq:indicator:near}
\end{equation}
 For the motivation and derivations of \eqref{eq:indicator:near}
 using near-field data we refer the reader to \cite{IJZ12}.

We are now going to derive an indicator function for the far-field data,
a counterpart of \eqref{eq:indicator:near} for the near-field data. To do so, we first derive
 the following key lemma about the $L^{2}$-correlation
 measure of the far-field data of two monopoles on the unit sphere.
\begin{lemma}
\label{lem:Ginf} For the far-field patterns $G^{\infty}(\hat{x},y)$
associated with the fundamental solutions $G(x,y)$ (see \eqref{eq:Ginf}
and \eqref{eq:G}), the following correlation holds for any two points
$x_{j}$ and $x_{p}$ in the sampling domain $\Omega$:
\begin{equation}
\int_{\mathbb{S}^{N-1}}G^{\infty}(\hat{x},x_{j})\overline{G^{\infty}}(\hat{x},x_{p})\,\mathrm{d}s(\hat{x})=C\,\Im\left(G(x_{p},x_{j})\right)\label{eq:Ginf:identity}
\end{equation}
 where $C$ is a constant depending only on the wave number $k$ and
the dimension $N$. \end{lemma}
\begin{proof}
By the definition of $G^{\infty}$, it is easy to verify that
\begin{equation}
\int_{\mathbb{S}^{N-1}}G^{\infty}(\hat{x},x_{j})\overline{G^{\infty}}(\hat{x},x_{p})\,\mathrm{d}s(\hat{x})=C_{1}\int_{\mathbb{S}^{N-1}}\exp(\mathrm{i}k\hat{x}\cdot(x_{p}-x_{j}))\,\mathrm{d}s(\hat{x}),\label{eq:tmp1}
\end{equation}
 where $C_{1}$ is a constant depending only on the wave number $k$
and the dimension $N$.

We carry out the proof in two steps, one for $N=2$, and the other
for $N=3$.

\textbf{Step 1}.  For $N=2$, \eqref{eq:Ginf:identity} is a special
case of Graf's addition theorem \cite[Eq.~(9.1.79)]{AS65} by showing
an integral representation of $J_{n}(kr)\exp(\mathrm{i}n\theta)$ for
an integer $n$. To see this, we can write the right-hand side of
\eqref{eq:tmp1} in polar coordinates
\begin{equation}
\int_{\mathbb{S}^{1}}\exp(\mathrm{i}k\hat{x}\cdot(x_{p}-x_{j}))\,\mathrm{d}s(\hat{x})=\int_{-\pi}^{\pi}\exp\left(\mathrm{i}k(x_{p}-x_{j})\cdot\left(\begin{array}{c}
\cos\theta\\
\sin\theta
\end{array}\right)\right)\,\mathrm{d}\theta.\label{eq:FH2d}
\end{equation}
 It suffices to show that 
\begin{equation}
J_{n}(kr)\exp(\mathrm{i}n\theta)=\frac{(-\mathrm{i})^{n}}{2\pi}\int_{-\pi}^{\pi}\exp\left(\mathrm{i}kx\cdot\left(\begin{array}{c}
\cos\phi\\
\sin\phi
\end{array}\right)\right)\exp(\mathrm{i}n\phi)\,\mathrm{d}\phi,\label{eq:Jnkr}
\end{equation}
with $r=|x|$, and $x=(r\cos\theta,\, r\sin\theta)$. Therefore we can arrive at \eqref{eq:Ginf:identity}
in two dimensions, by setting $n=0$,
 replacing $x$ by $x_p - x_j$  in
 \eqref{eq:Jnkr} and using the fact that 
 $$
 C \Im \left( G(x_{p},x_{j}) \right) =\Im \left({\mathrm{i}} H_{0}^{(1)}(k|x_{p}-x_{j}|) \right) =J_0(k|x_{p}-x_{j}|)\,.
 $$
The formula \eqref{eq:Jnkr} expresses regular cylindrical wave
functions as a superposition of plane waves.
To see \eqref{eq:Jnkr}, we first recall the Jacobi expansion \cite[Eq.~(22)]{Wat44}
\begin{equation}
\exp(\mathrm{i}kr\cos\psi)=\sum_{m=-\infty}^{\infty}\mathrm{i}^{m}J_{m}(kr)\exp(\mathrm{i}m\psi).\label{eq:Jacobi:expansion}
\end{equation}
 By replacing $\psi$ by $\theta-\phi$ in \eqref{eq:Jacobi:expansion},
then multiplying it by $\exp\left(-\mathrm{i}n(\theta-\phi)\right)$
and integrating over $\left(-\pi,\,\pi\right)$ with respect to
$\phi$ we obtain
\begin{align}
\mathrm{i}^{n}J_{n}(kr) &
=\frac{1}{2\pi}\int_{-\pi}^{\pi}\exp\left(\mathrm{i}kr\cos(\theta-\phi)\right)\exp\left(-\mathrm{i}n(\theta-\phi)\right)\,\mathrm{d}\phi,\label{eq:imJm}
\end{align}
 where we have employed the simple fact that
 $\int_{-\pi}^{\pi}\exp(\mathrm{i}n\psi)\,\mathrm{d}\psi=2\pi$ for $n=0$ and $0$
 for $n\ne 0$.

Now we can rewrite \eqref{eq:imJm} as
\begin{equation}
\mathrm{i}^{n}J_{n}(kr)=\frac{\exp\left(-\mathrm{i}n\theta\right)}{2\pi}\int_{-\pi}^{\pi}\exp\left(\mathrm{i}kr\cos\phi\cos\theta+r\sin\phi\sin\theta)\right)\exp(\mathrm{i}n\phi)\,\mathrm{d}\phi,
\end{equation}
 which reproduces \eqref{eq:Jnkr} by a simple rearrangement.

\textbf{Step 2}. For $N=3$, we use the following special
variant of the Funk-Hecke formula (cf.\ \cite[Eq. 2.44]{CoK98} and
\cite[p. 29]{Mul97})
\begin{equation}
\int_{\mathbb{S}^{2}}\exp(-\mathrm{i}kx\cdot\hat{z})Y_{n}(\hat{z})\,\mathrm{d}s(\hat{z})=\frac{4\pi}{\mathrm{i}^{n}}j_{n}(k|x|)Y_{n}(\hat{x}),\quad x\in\mathbb{R}^{3},\, r>0\label{eq:FH:3d}
\end{equation}
for spherical harmonics $Y_{n}(\cdot)$ of order $n$ and spherical
Bessel functions $j_{n}(\cdot)$ of order $n$.
Now we set $n=0$  and replace $x$ by
$x_{p}-x_{j}$ in \eqref{eq:FH:3d}.  Being aware of the lowest order
spherical harmonics $Y_{0}(\hat{x})\equiv1/\sqrt{4\pi}$, which is a
constant and thus can be dropped from both sides of the equation.
Since $j_{0}(t)=\sin(t)/ t$, we obtain that
$$
C \Im \left( G(x_{p},x_{j}) \right) =  \Im \left(
\frac{\exp(\mathrm{i}k|x_{p}-x_{j}|)}{k|x_{p}-x_{j}|} \right) =
j_0(k|x_{p}-x_{j}|)
$$
in three dimensions.  This completes the proof.
\end{proof}

Next, we divide the sampling
domain $\Omega$ enclosing the scatterer $D$ into a set of small
elements $\left\{ \tau_{j}\right\} $, i.e., squares in 2D or cubes
in 3D. Then we approximate the integral relation \eqref{eq:uinf}
by the rectangular quadrature rule on each element:
\begin{equation}
u^{\infty}(\hat{x})=\int_{\Omega}G^{\infty}(\hat{x},y)I(y)\,\mathrm{d}y\,\approx\sum_{j}w_{j}G^{\infty}(\hat{x},y_{j}),\label{eq:uinf:approx}
\end{equation}
 where the summation is over all elements $\tau_j$ which intersect the sampling
 domain $\Omega$, the weight $w_j$ is given by $|\tau_{j}|I_{j}$,  with $|\tau_{j}|$
being the area/volume of the $j$-th element $\tau_{j}$ ($N=2$,
$3$) and $I_{j}$ the evaluation of $I(x)$ at the center of $\tau_{j}$.

%
Now multiplying \eqref{eq:uinf:approx} by $\overline{G^{\infty}}(\hat{x},x_{p})$
for any sampling point $x_{p}\in\Omega$, then integrating over the
unit sphere $\mathbb{S}^{N-1}$ and using (\ref{eq:Ginf:identity}),
we obtain
\begin{equation}
\int_{\mathbb{S}^{N-1}}u^{\infty}(\hat{x})\overline{G^{\infty}}(\hat{x},x_{p})\,\mathrm{d}s(\hat{x})\approx C\sum_{j}w_{j}\,\Im\left(G(x_{p},x_{j})\right).\label{eq:uinf:reason}
\end{equation}
 As we illustrate later (see Fig.~\ref{fig:decay}), the right-hand
side term $\Im\left(G(x_{p},x_{j})\right)$ above approaches a constant
when $x_{p}$ tends close to some point scatterer $x_{j}$ and decays
quickly as $x_{p}$ moves away from $x_{j}$. This behavior motivates us with
the following important indicator function for any sampling point
$x_{p}\in\Omega$:
\begin{equation}
\Phi^{\infty}(x_{p})=\frac{\left|\left\langle u^{\infty},\, G^{\infty}(\cdot,x_{p})\right\rangle _{L^{2}(\mathbb{S}^{N-1})}\right|}{\left\Vert u^{\infty}\right\Vert _{L^{2}(\mathbb{S}^{N-1})}\left\Vert G^{\infty}(\cdot,x_{p})\right\Vert _{L^{2}(\mathbb{S}^{N-1})}}\,,\label{eq:indicator:far}
\end{equation}
where the $L^{2}$-inner product occurring in the numerator is the
integral given by the left-hand side of (\ref{eq:uinf:reason}). 
In view of the singular behaviors of the fundamental solution $G(x, y)$ in
(\ref{eq:uinf:reason}), the indicator $\Phi^{\infty}(x_{p})$ approaches the unity when $x_{p}$
tends close to or lies within the medium scatterer and decays
quickly as $x_{p}$ moves away from the scatterer; see
Fig.~\ref{fig:decay}.
This indicating behavior will serve as the key
ingredient in our new DSM using far-field data.

We emphasize that the indicator function $\Phi^{\infty}$  in (\ref{eq:indicator:far})
was studied earlier in the orthogonality sampling method
\cite{Pot10,Gri11}, but derived from a different approach and motivation,
and under a ``smallness'' assumption.
As it has been seen,
our theoretical foundation does not require such \emph{smallness} restriction,
thus makes it possible to explain partially why the DSM works also well for
thin or long scatterers like rings and cracks. This is confirmed
by numerical tests (see Examples 6 and 7 in
Section~\ref{sec:numerics}).

Finally we make a short remark about computing the indicator function
$\Phi^{\infty}(x_{p})$ in \eqref{eq:indicator:far}. Clearly the
computing is purely explicit and direct, involving only a cheap scalar
product and a normalization operation, and it does not involve any
matrix inversion and any solution process. Moreover, normalizations
are quite cheap too. In fact we can easily see that the denominator in
\eqref{eq:indicator:far} is actually a constant and does not depend
on the sampling point $x_{p}$, so it can be computed once for all
sampling points. This is quite different from the indicator function
(\ref{eq:indicator:near}) using near-field data. Therefore, the DSM
using far-field data is computationally much cheaper than its counterpart
using near-field data.

\section{Numerical examples\label{sec:numerics}}

In this section, we will carry out a systematic evaluation of the
performance of the DSM using the far-field data for some inverse scattering
benchmark problems, including scatterers of obstacle, medium and crack
types. In particular, we shall also present the numerical reconstructions
by the DSM using the near-field data for most examples so that we
may clearly see the comparisons and distinctions between these two
different types of measurement data.

Let us first introduce our experiment settings. We shall present several
examples to illustrate the applicability and robustness of the proposed
method for determining the scatterers from both exact and noisy data.
We shall take the unitary wavelength $\lambda=1$, and the wave number
$k=2\pi$. When it is not mentioned, only one incident direction $d=(1,\,1)^{T}/\sqrt{2}$
is employed in the examples; otherwise two incidents, $d_{1}=(1,\,1)^{T}/\sqrt{2}$
and $d_{2}=(1,\,-1)^{T}/\sqrt{2}$, are used and mentioned explicitly.
In all the examples, the scattered near-field data $u^{s}$ is measured
at 50 points uniformly distributed on a circle of radius $4\lambda$
centered at the origin, and the far-field pattern $u^{\infty}$ is
observed from 50 uniform distributed angles on the unit circle $\mathbb{S}^{1}$.
The noisy data $u_{\delta}^{s}$ and $u_{\delta}^{\infty}$ are generated point-wise
by the formulae
\begin{equation}
u_{\delta}^{s}=u^{s}(x)+\epsilon\,\zeta\underset{x}{\max}|u^{s}(x)|
\quad \mbox{and} \quad
u_{\delta}^{\infty}=u^{\infty}(x)+\epsilon\,\zeta\underset{x}{\max}|u^{\infty}(x)|
\end{equation}
for both near-field and far-field  data respectively,
where $\epsilon$ refers to the relative noise level, and both real
and imaginary parts of the noise $\zeta$ follow the standard normal
distribution. Our near-field data are synthesized using
the quadratic finite element discretization in the domain $(-6,\,6)^{2}$
enclosed by a PML layer of width $1$ to damp the reflection. Local
adaptive refinement scheme within the inhomogeneous scatterer is adopted
to enhance the resolution of the scattered wave. The far-field data
are generated approximately by the integral representation \cite[Eq.~(3.64)]{CoK98}
along the circle centered at the origin with radius $5$ using the
composite Simpson's rule:
$$
    \int_0^{2\pi} 5f(5,\theta) ~\mathrm{d}\theta  \approx \frac{\pi}{15}\sum_0^{24} (f_{2i}+4 f_{2i+1} + f_{2i+2})
$$
where $f(r,\theta)$ is the integrand of Eq.~(3.64) in \cite{CoK98} in the polar form and $f_j$ is the value of $f$
evaluated at the $j$-th quadrature node $(r,\theta)=(5, 2j\pi/50)$, $j=0,1,\ldots,50$.
To ensure the sufficient accuracy of the far-field data, we refine the mesh 
successively till the relative maximum error of the data from two successive meshes 
is below $0.1\%$. The indicator function value is normalized between $0$ and $1$, 
and the sampling domain is fixed to be
$\Omega=[-2\lambda,\,2\lambda]^{2}$, which is divided into small
squares of equal width $h=0.01\lambda$. The contour plot of the indicator
function value will be displayed as an estimate to the profiles of
unknown scatterers. For multiple incident waves with directions $d_{1}$,
$d_{2}$, $\text{\ensuremath{\ldots}\ }$, $d_{\nu}$, we take the maximum
of all indicator function values point-wise:
\begin{equation}
\Phi^{\infty}(x_{p})=\max\left\{ \Phi^{\infty}(x_{p};\, d_{1}),\,\Phi^{\infty}(x_{p};\, d_{2}),\,\cdots,\Phi^{\infty}(x_{p};\, d_{\nu})\right\} ,
\end{equation}
 where $\Phi^{\infty}(x_{p};\, d_{l})$, $l=1,2,\ldots,\nu$, is the
indicator function defined by \eqref{eq:indicator:far} associated
with the $l$-th incident direction $d_{l}$. We denote the DSM using
near-field and far-field data by DSM(n) and DSM(f), respectively,
for sake of comparisons of their performances. We shall carry out
three groups of benchmark problems: medium scatterer, obstacle scatterer,
and crack scatterer.

\subsection*{DSM for medium scatterers}

We first test the four examples from \cite{IJZ12}, but using both
near-field and far-field data to compare the performances of two different
types of data.

\smallskip{}

\textbf{Example 1} (A singular point medium). The example considers
one square scatterer of width $0.02\lambda$ located at the origin,
with coefficient $\eta$ of the scatterer being $1$.

This point scatterer identification is a standard benchmark problem.
The contour plots using the DSM(n) and DSM(f) are shown in Figure$\text{\,}$\ref{fig:ex1}.
With only one incident wave probing, the location of the point scatterer
is accurately positioned even if the data is severely distorted. We
observe similar behavior of both indicator functions using near-field
and far-field data. Both DSM(n) and DSM(f) can locate the position
of the ``point'' scatterer in an accurate and stable way. The presence
of noise level up to $20\%$ in both measurement data has little effects
on the indicator functions.

\begin{figure}
\hfill{}\hspace{-0.03\textwidth}%
\begin{tabular}{c}
\includegraphics[clip,width=0.32\textwidth]{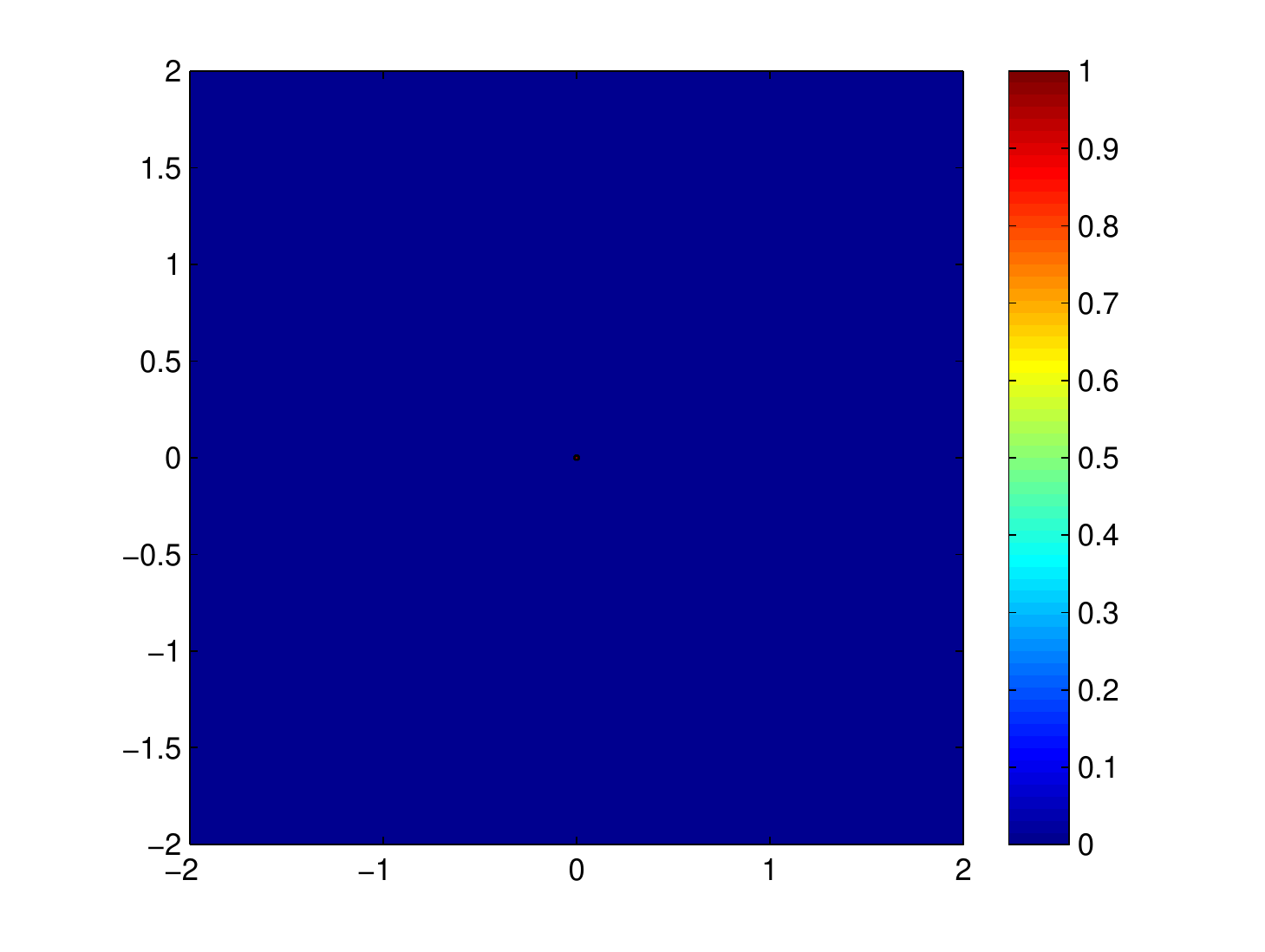}\tabularnewline
(a)\tabularnewline
\end{tabular}\negthinspace{}%
\begin{tabular}{cc}
\includegraphics[width=0.32\textwidth]{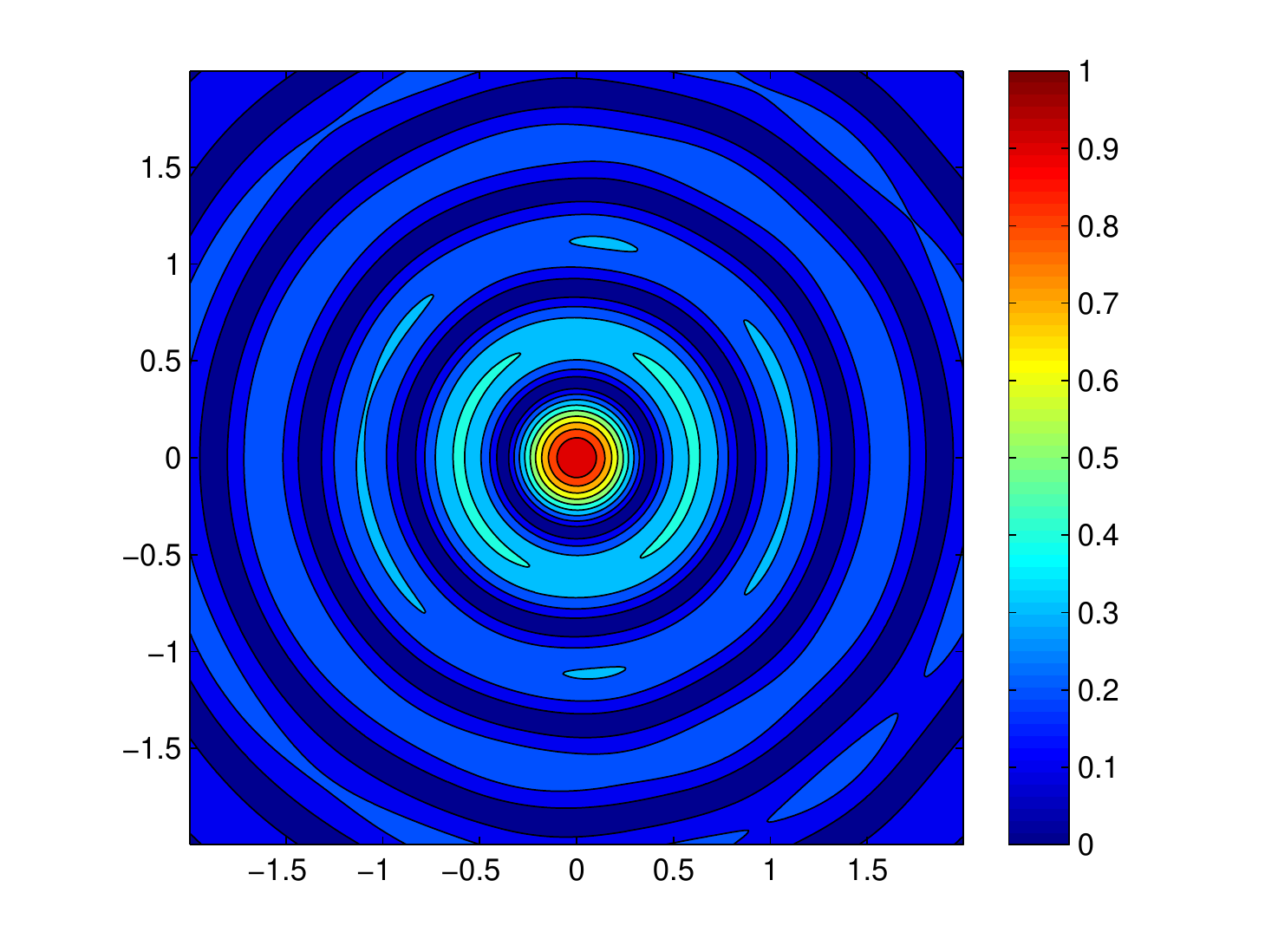} & \includegraphics[width=0.32\textwidth]{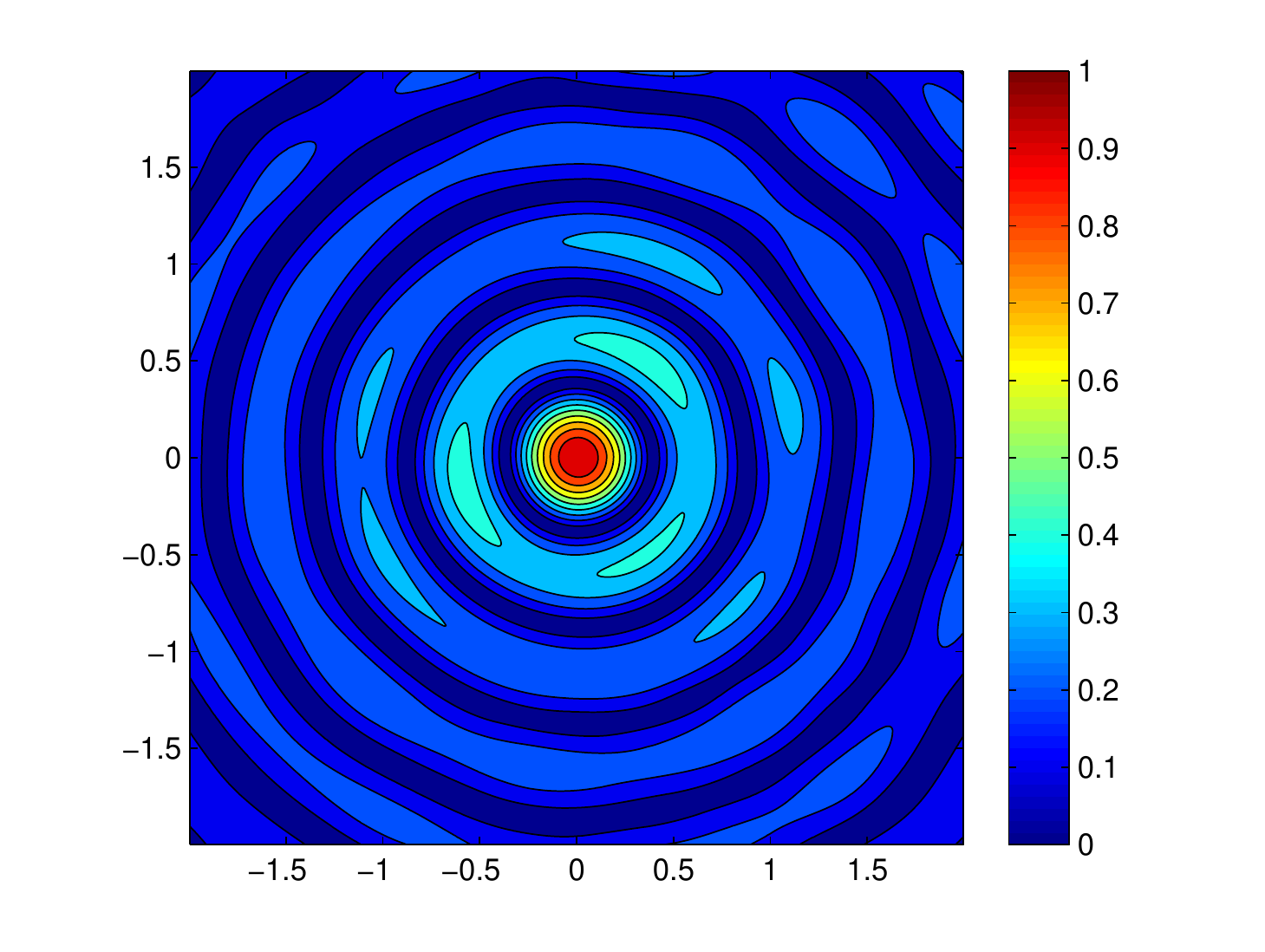}\tabularnewline
(b) & (c)\tabularnewline
\includegraphics[width=0.32\textwidth]{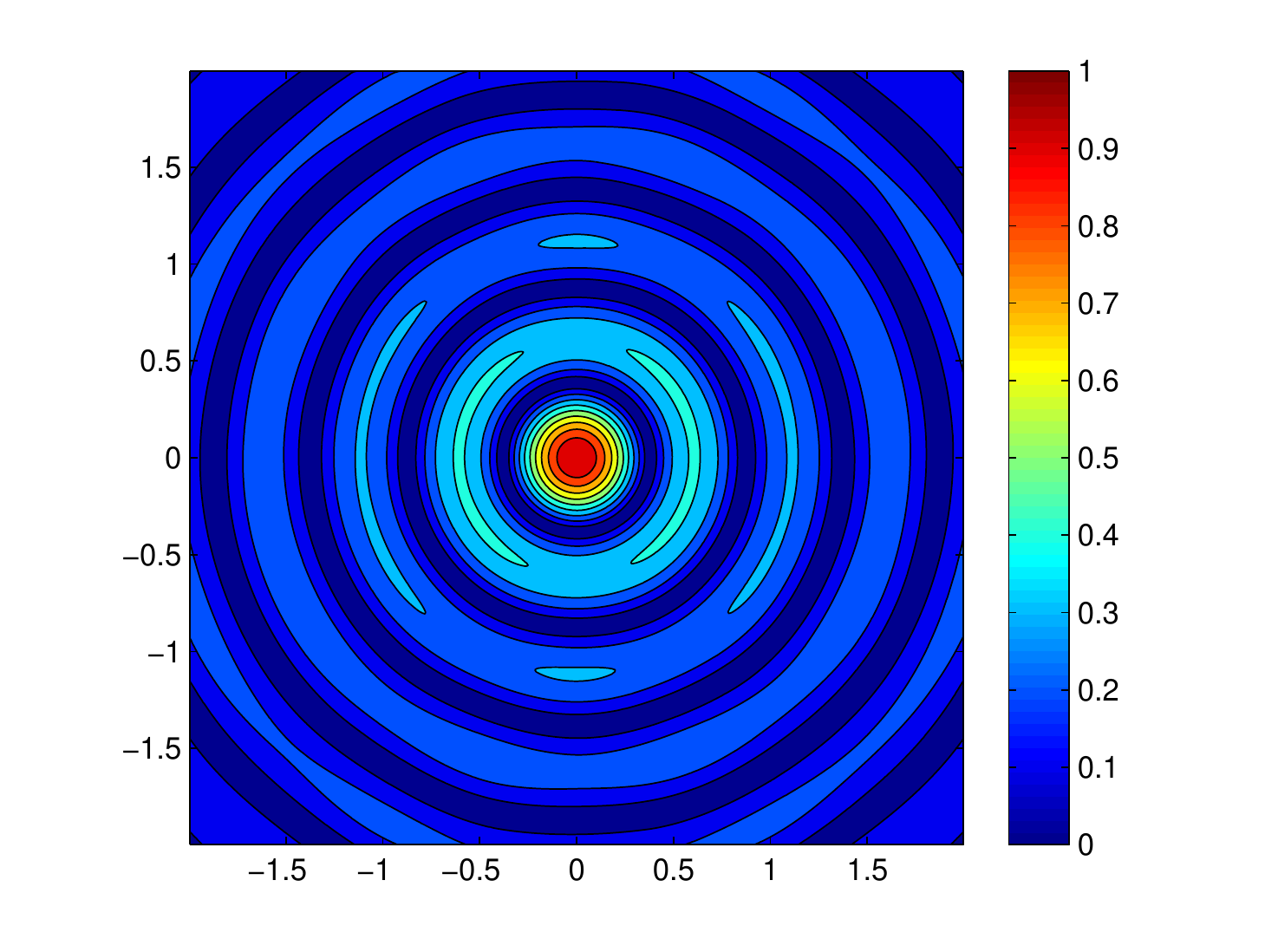} & \includegraphics[width=0.32\textwidth]{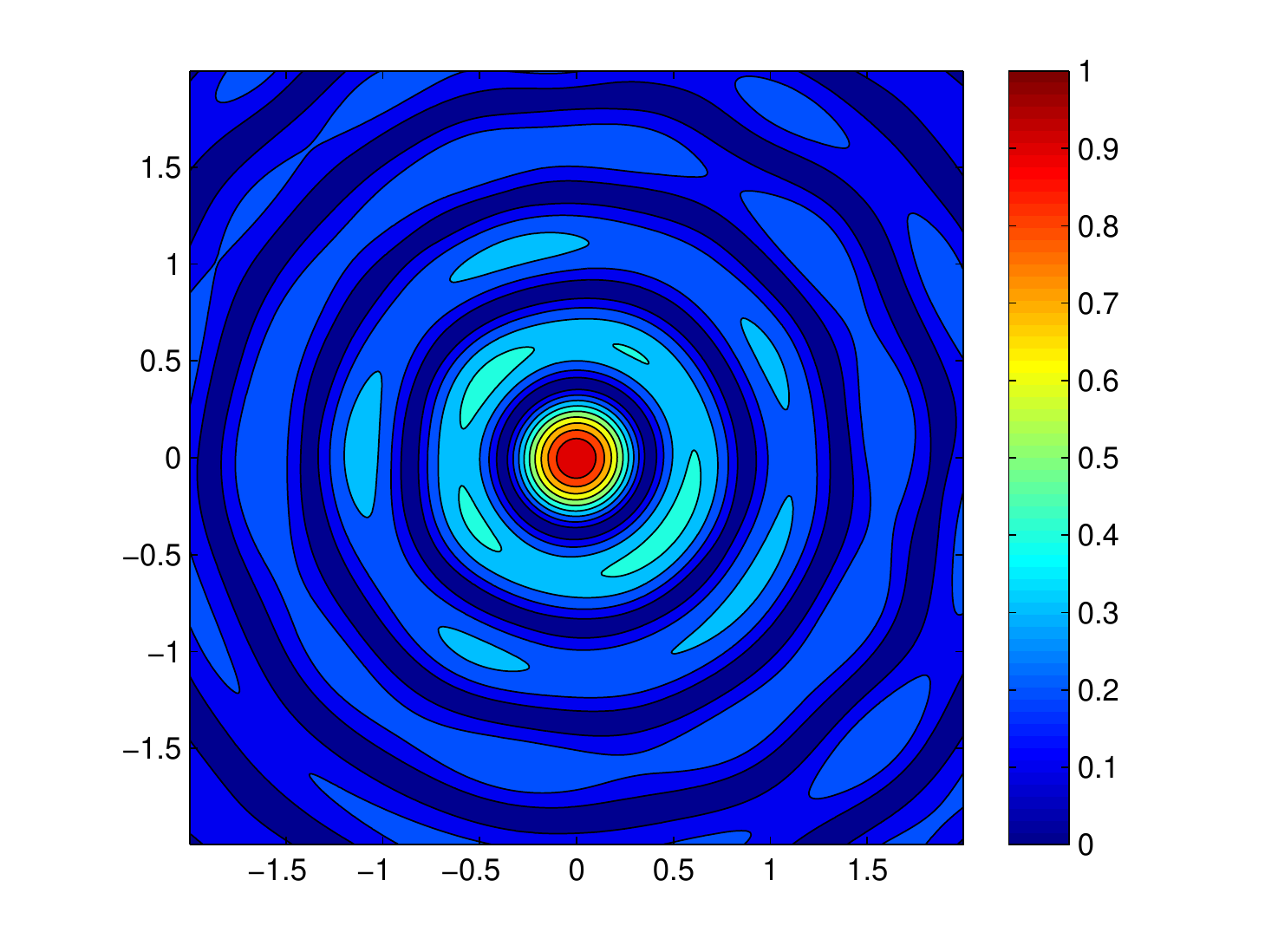}\tabularnewline
(d) & (e)\tabularnewline
\end{tabular}\hfill{}

\caption{\label{fig:ex1} Example 1: (a) true scatterer; Reconstruction results
using (b) exact near-field data, (c) noisy near-field data with $\epsilon=20\%$,
(d) exact far-field data, (e) noisy far-field data with $\epsilon=20\%$.}
\end{figure}

\smallskip{}

\textbf{Example 2} (Well-separated medium scatterers) Two square scatterers
of side length $0.3\lambda$, located at $(-0.8\lambda,\,-0.7\lambda)$
and $(0.3\lambda,\,0.8\lambda)$ respectively, are considered, both
with coefficient $\eta$ being $1$.

This example verifies that the DSM can capture multiple sources. The
indicating contour plots of the DSM using near-field and far-field
data are shown in Figure$\text{\,}$\ref{fig:ex2}. Both DSM(n) and
DSM(f) are able to identify the location of both scatterers even if
the measured data is significantly perturbed. Through studying the
contrast of the indicator contours of medium scatterers with respect
to the surrounding background medium, one can choose the cut-off value
ranging from $0.7$ to $0.8$ to truncate the profiles of the two
scatterers. Choosing a large cut-off value helps reducing those misleading
spurious medium scatterers, see yellow patches in Figure$\text{\,}$\ref{fig:ex2}(b)-(e).

\begin{figure}
\hfill{}\hspace{-0.03\textwidth}%
\begin{tabular}{c}
\includegraphics[clip,width=0.32\textwidth]{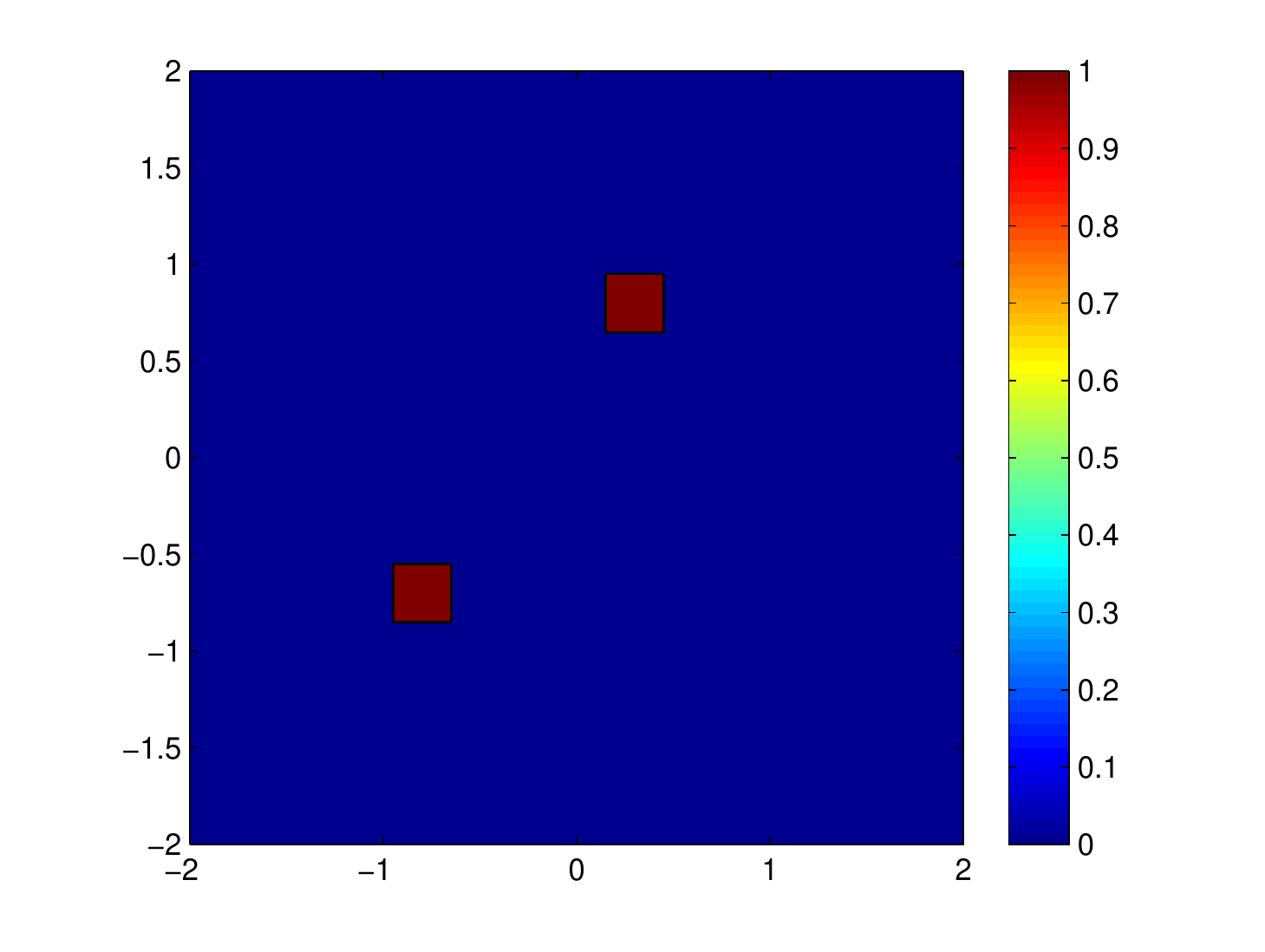}\tabularnewline
(a)\tabularnewline
\end{tabular}\negthinspace{}%
\begin{tabular}{cc}
\includegraphics[width=0.32\textwidth]{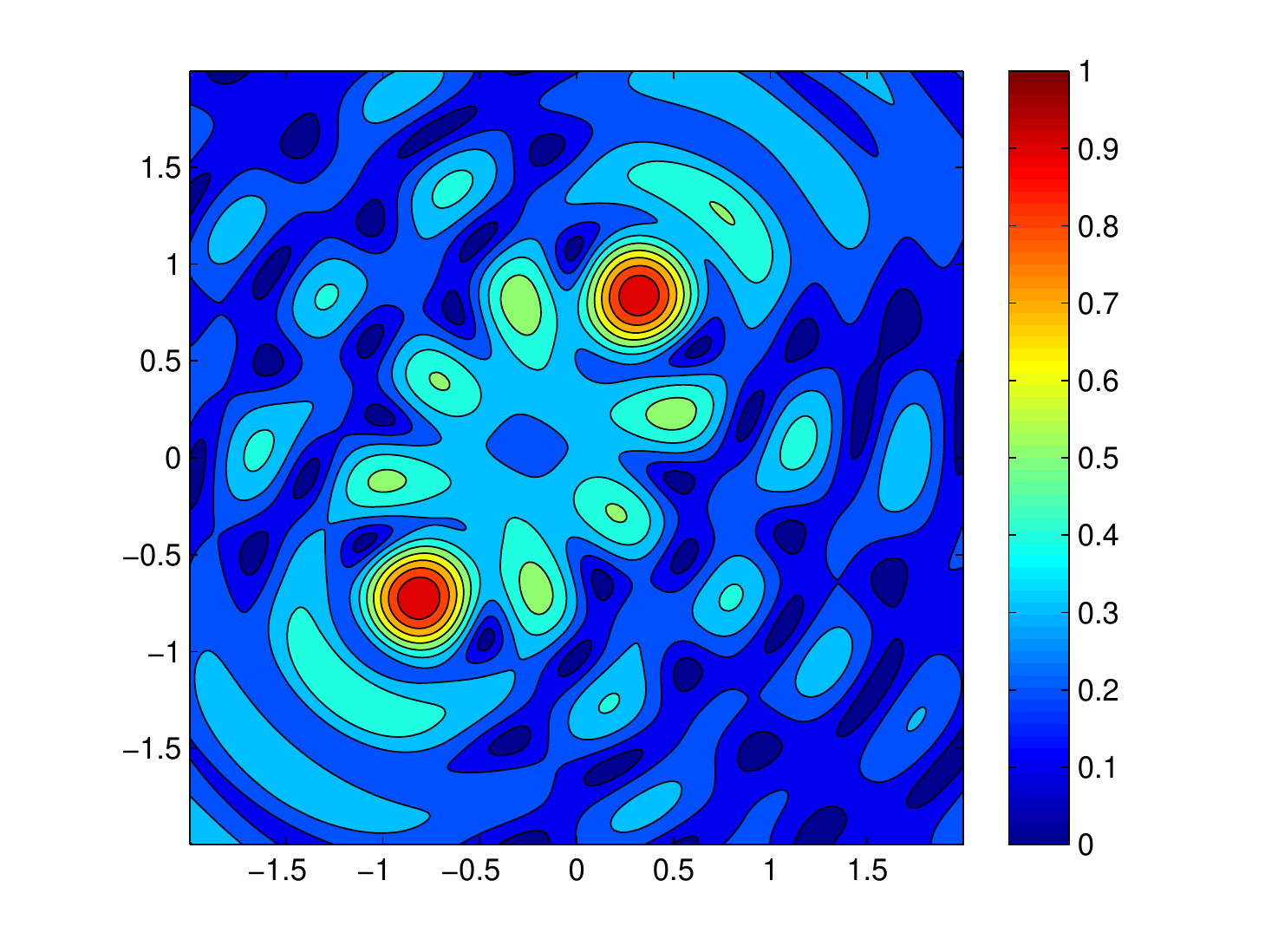} & \includegraphics[width=0.32\textwidth]{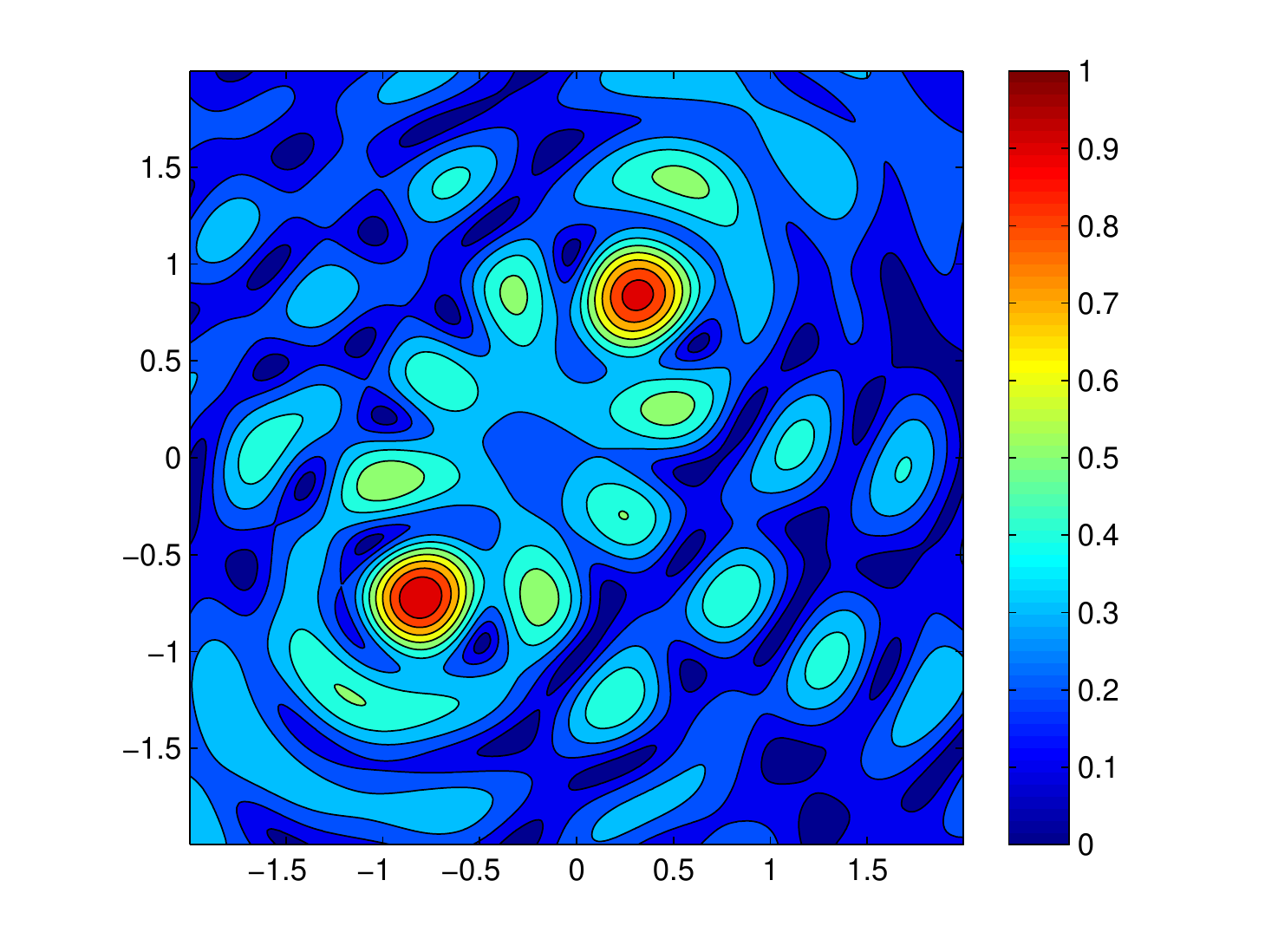}\tabularnewline
(b) & (c)\tabularnewline
\includegraphics[width=0.32\textwidth]{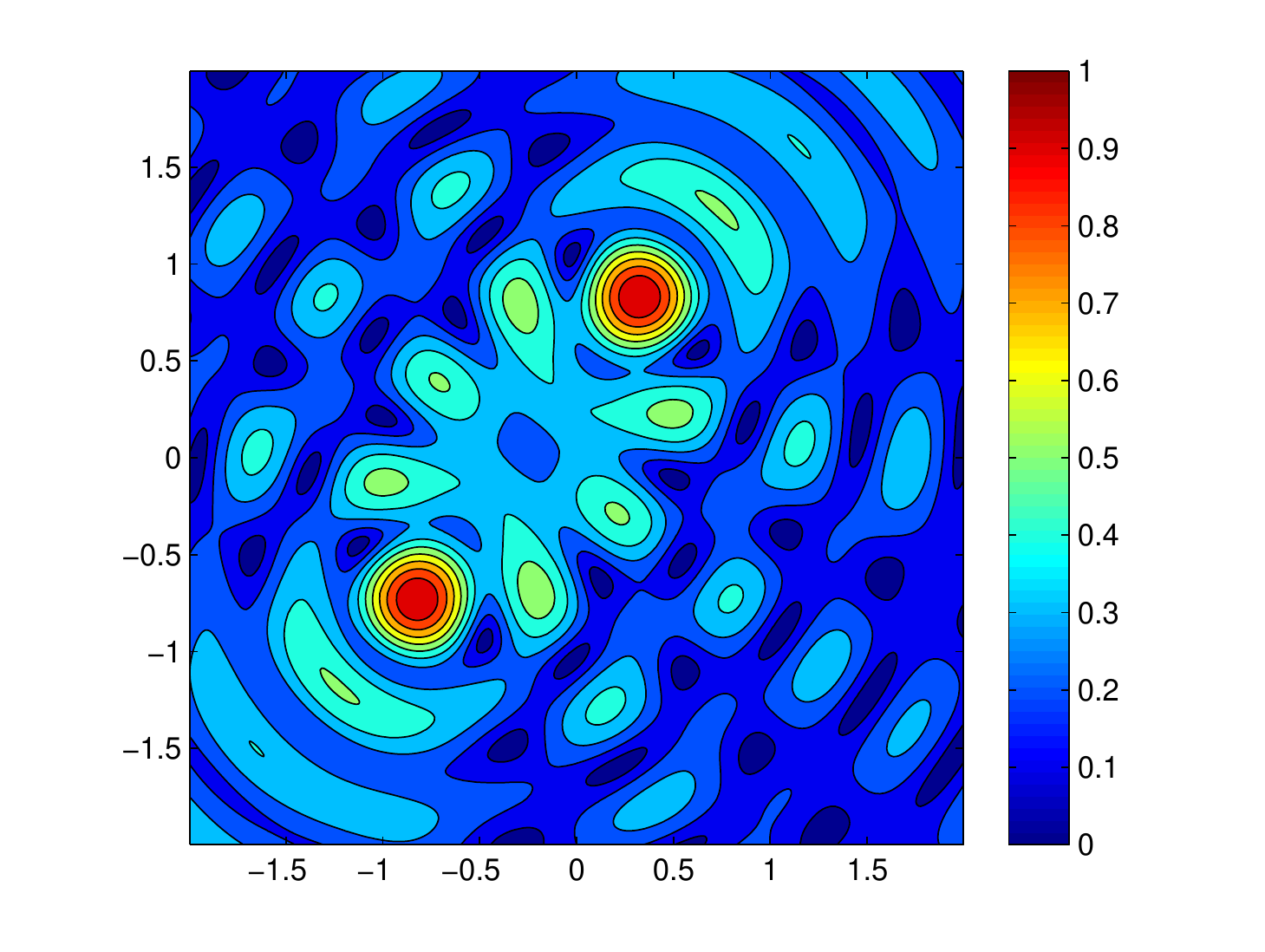} & \includegraphics[width=0.32\textwidth]{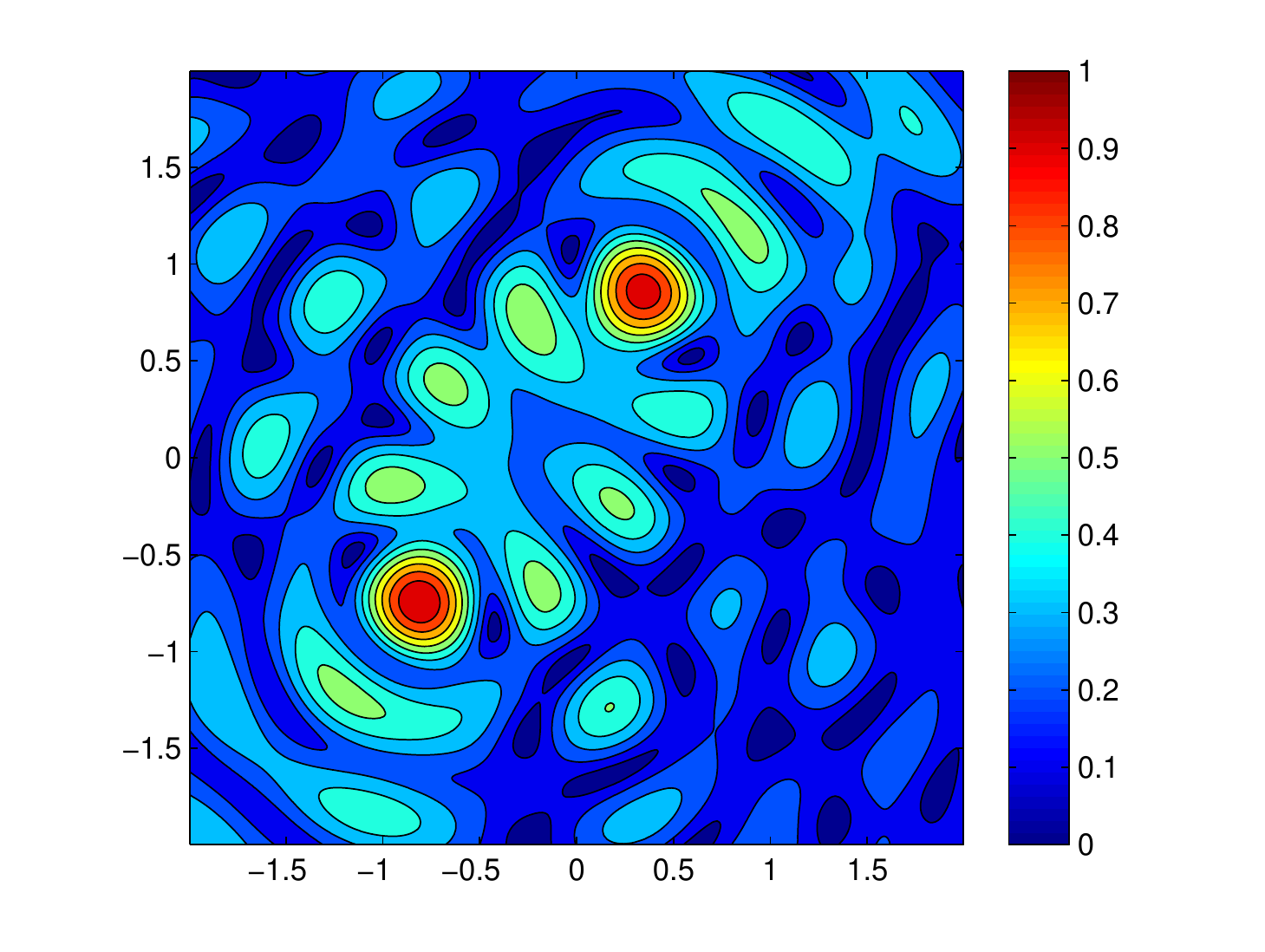}\tabularnewline
(d) & (e)\tabularnewline
\end{tabular}\hfill{}

\caption{\label{fig:ex2} Example 2: (a) true scatterer; reconstruction results
using (b) exact near-field data, (c) noisy near-field data with $\epsilon=20\%$,
(d) exact far-field data, (e) noisy far-field data with $\epsilon=20\%$.}
\end{figure}

\smallskip{}

\textbf{Example 3} (Two close medium scatterers). Here the same setting
as Example 2 is laid out, except that the two scatterers are now moved
to $(-0.25\lambda,\,0)$ and $(0.25\lambda,\,0)$, respectively.

This example investigates the resolution limit when the DSM can separate
multiple sources. The results of near-field case and far-field case
are shown in Figure$\text{\,}$\ref{fig:ex3}. Both the DSM(n) and
DSM(f) can separate these two close medium scatterers well. The location
and size of the scatterers agrees well with the true ones even under
large noise. When we further reduce the distance between those two
scatterers (less than half a wave length), DSM cannot separate the
two scatterer any longer, which is consistent with the Heisenberg
uncertainty principle.

\begin{figure}
\hfill{}\hspace{-0.03\textwidth}%
\begin{tabular}{c}
\includegraphics[clip,width=0.32\textwidth]{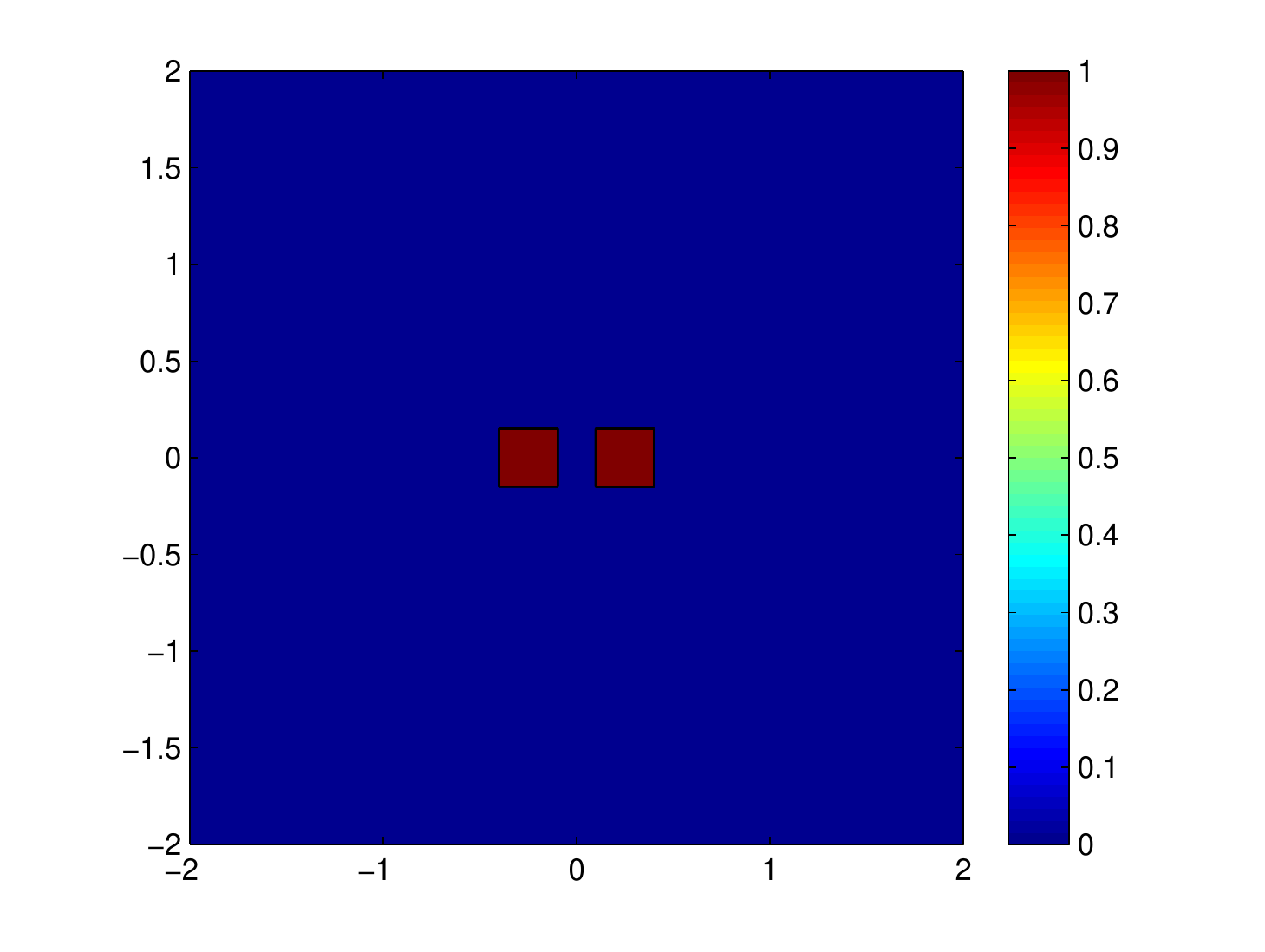}\tabularnewline
(a)\tabularnewline
\end{tabular}\negthinspace{}%
\begin{tabular}{cc}
\includegraphics[width=0.32\textwidth]{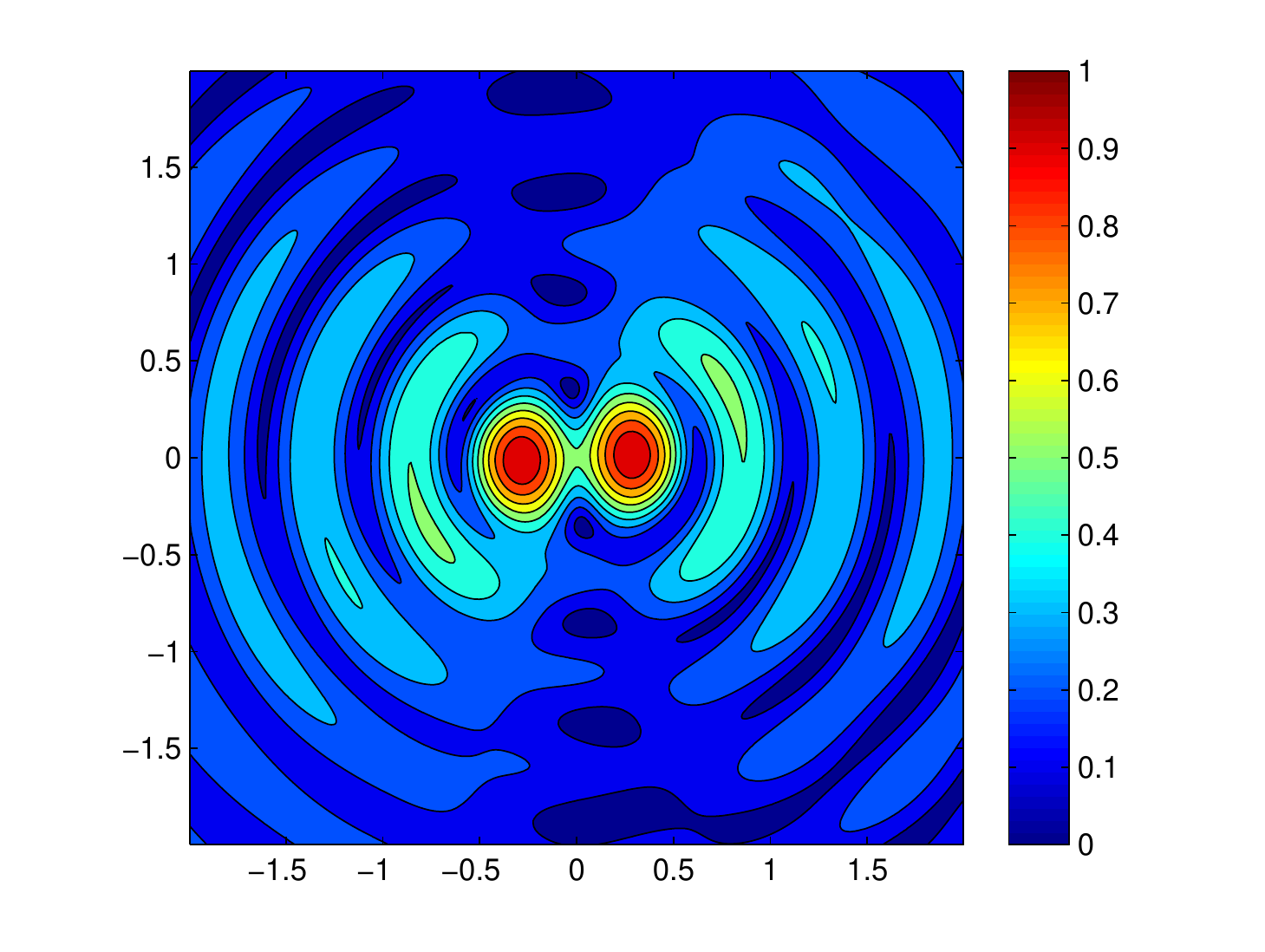} & \includegraphics[width=0.32\textwidth]{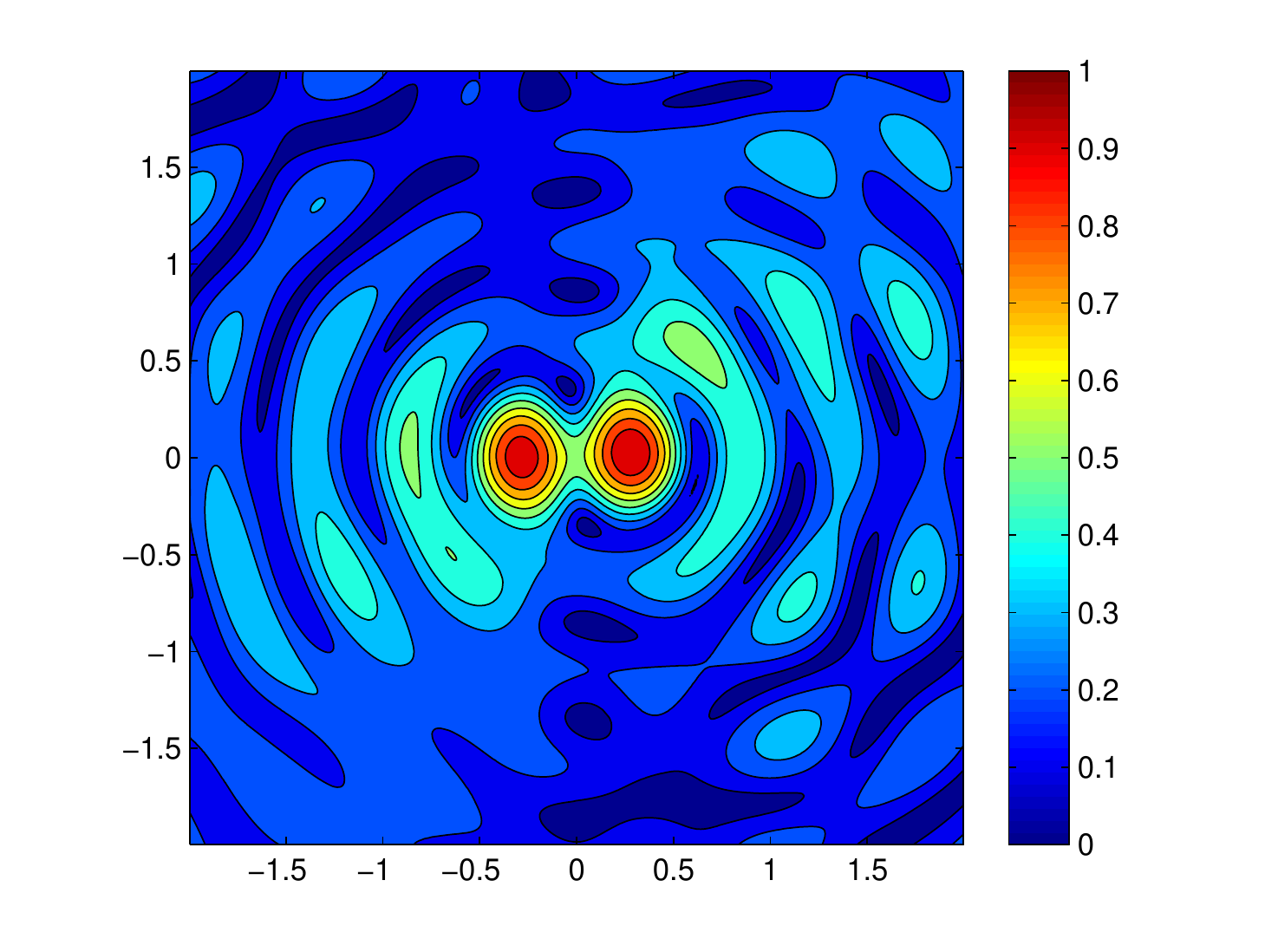}\tabularnewline
(b) & (c)\tabularnewline
\includegraphics[width=0.32\textwidth]{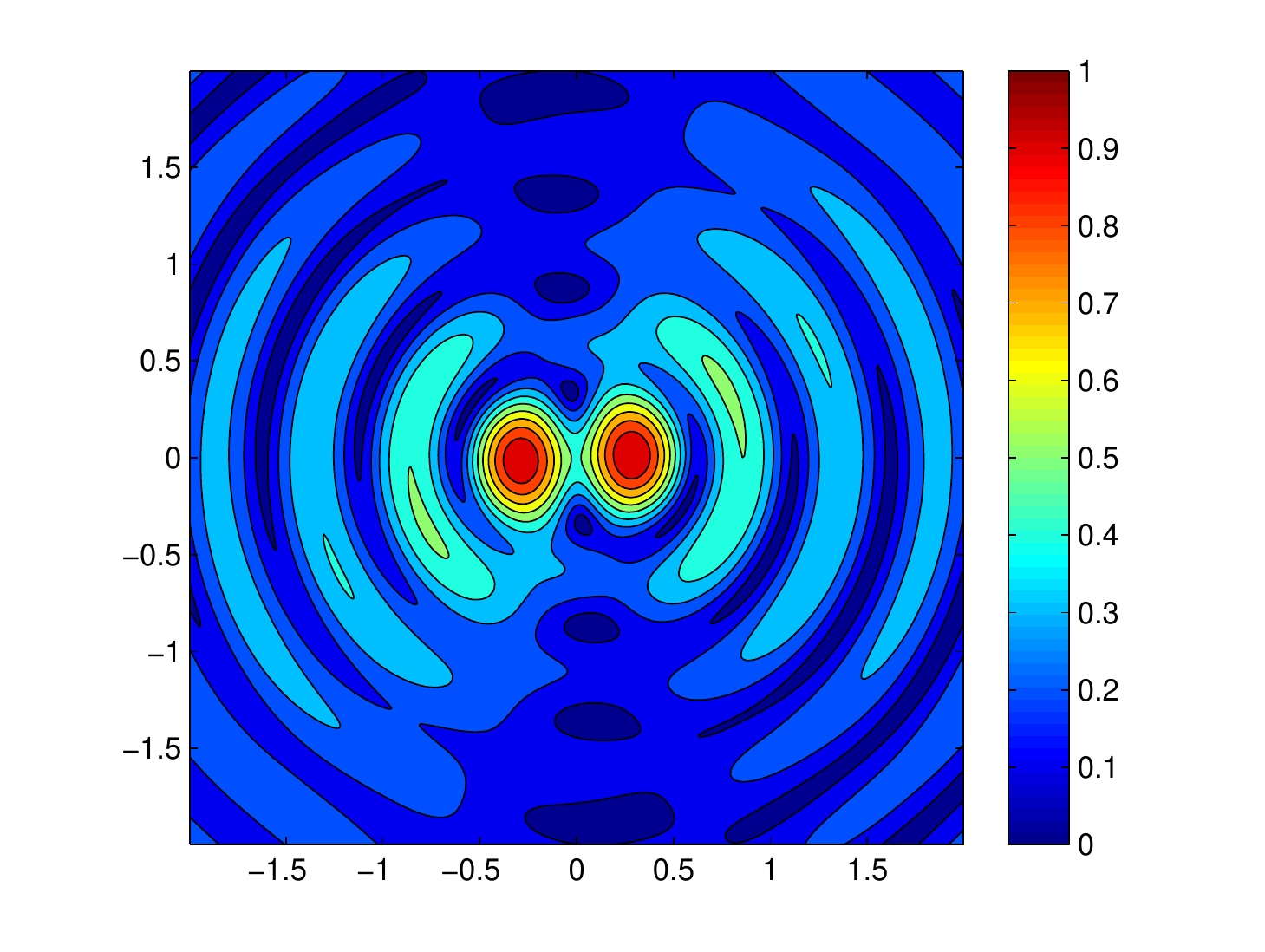} & \includegraphics[width=0.32\textwidth]{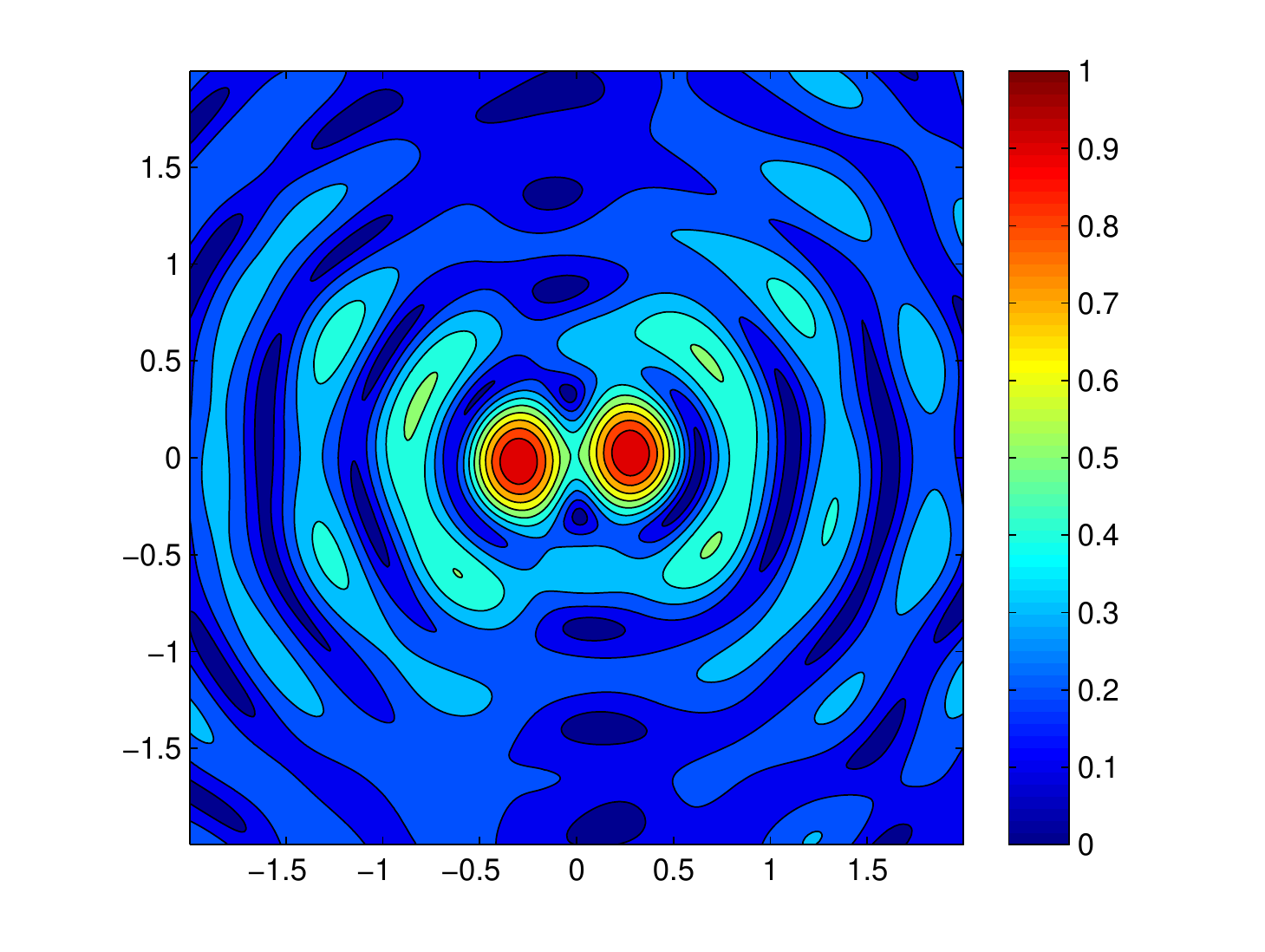}\tabularnewline
(d) & (e)\tabularnewline
\end{tabular}\hfill{}

\caption{\label{fig:ex3} Example 3: (a) true scatterer; reconstruction results
using (b) exact near-field data, (c) noisy near-field data with $\epsilon=20\%$,
(d) exact far-field data, (e) noisy far-field data with $\epsilon=20\%$.}
\end{figure}

\smallskip{}

\textbf{Example 4 (A ring-shaped medium scatterer)}. We consider a
ring-shaped square scatterer located at the origin, with the outer
and inner side lengths being $0.6\lambda$ and $0.4\lambda$ respectively
and the coefficient $\eta$ being $1$. Two incident directions $d_{1}=\frac{1}{\sqrt{2}}(1,\,1)^{T}$
and $d_{2}=\frac{1}{\sqrt{2}}(1,\,-1)^{T}$ are employed. The reconstruction
results using near-field and far-field data are shown in Figures $\text{\,}$\ref{fig:ex4:near}
and \ref{fig:ex4:far} respectively.

This is a rather challenging task to estimate the profile of the scatterer.
We see that the DSM(f) has similar performance to the DSM(n). One
single incident wave is inadequate to revolve full features of the
scatterer. Using two incident waves, The four corners of the square
ring is correctly labeled by the peaks of the superimposed indicator
function in Figure\ \ref{fig:ex4:near}(d), (g) and Figure\,\ref{fig:ex4:far}(c),
(f).

\begin{figure}
\hfill{}\includegraphics[width=0.33\textwidth]{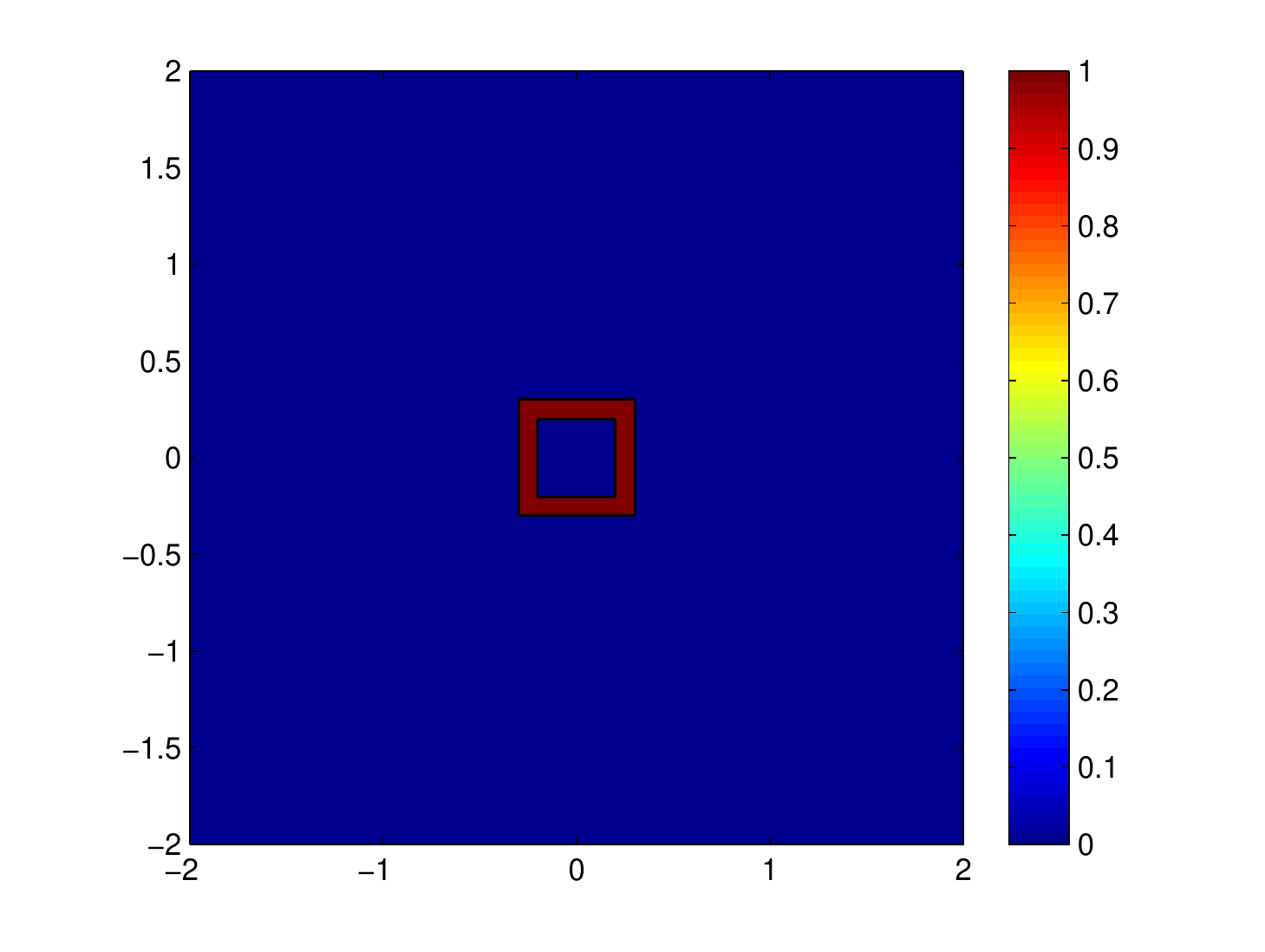}\hfill{}

\hfill{}(a)\hfill{}

\hfill{}\includegraphics[width=0.33\textwidth]{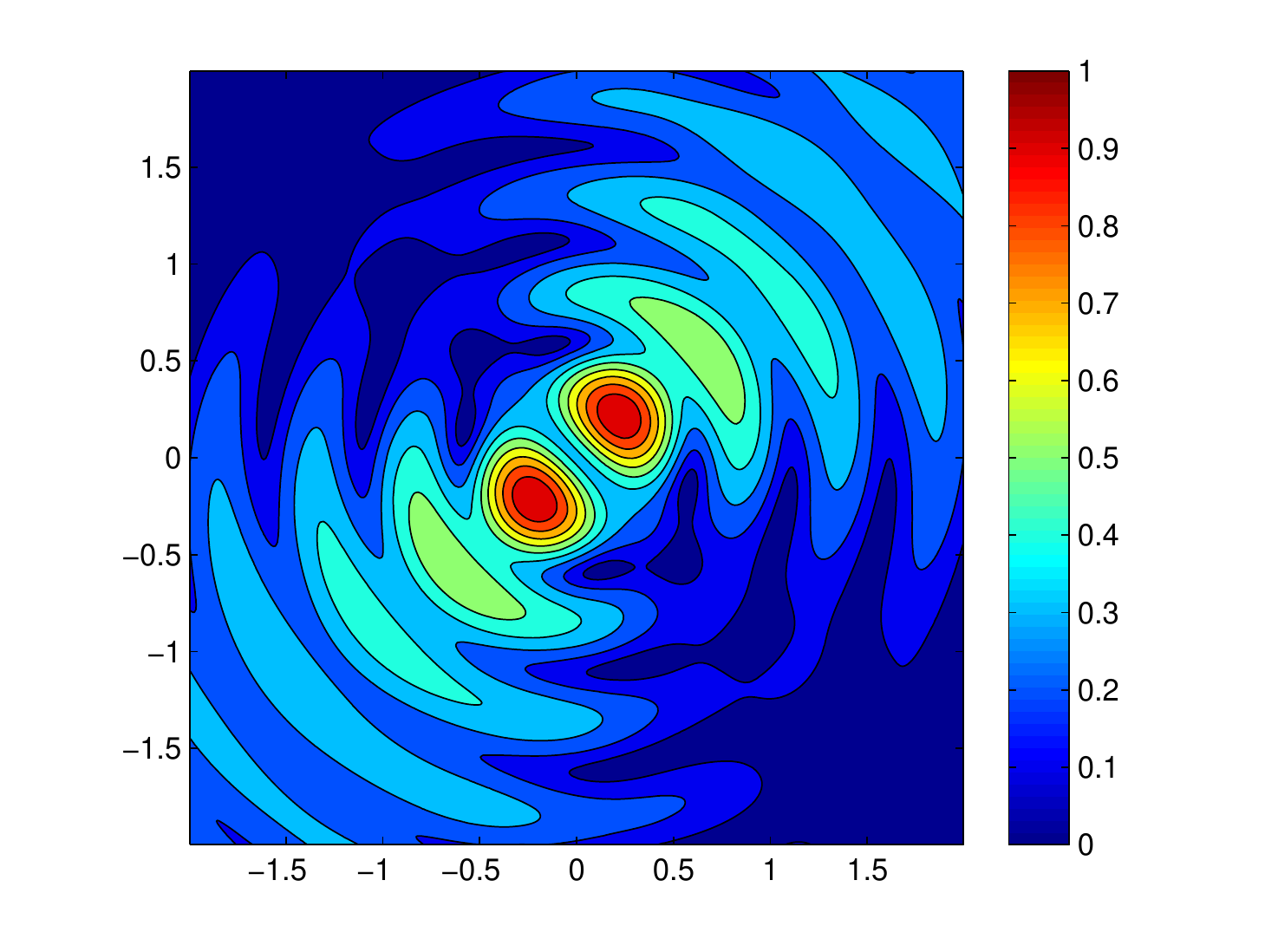}\hfill{}\includegraphics[width=0.33\textwidth]{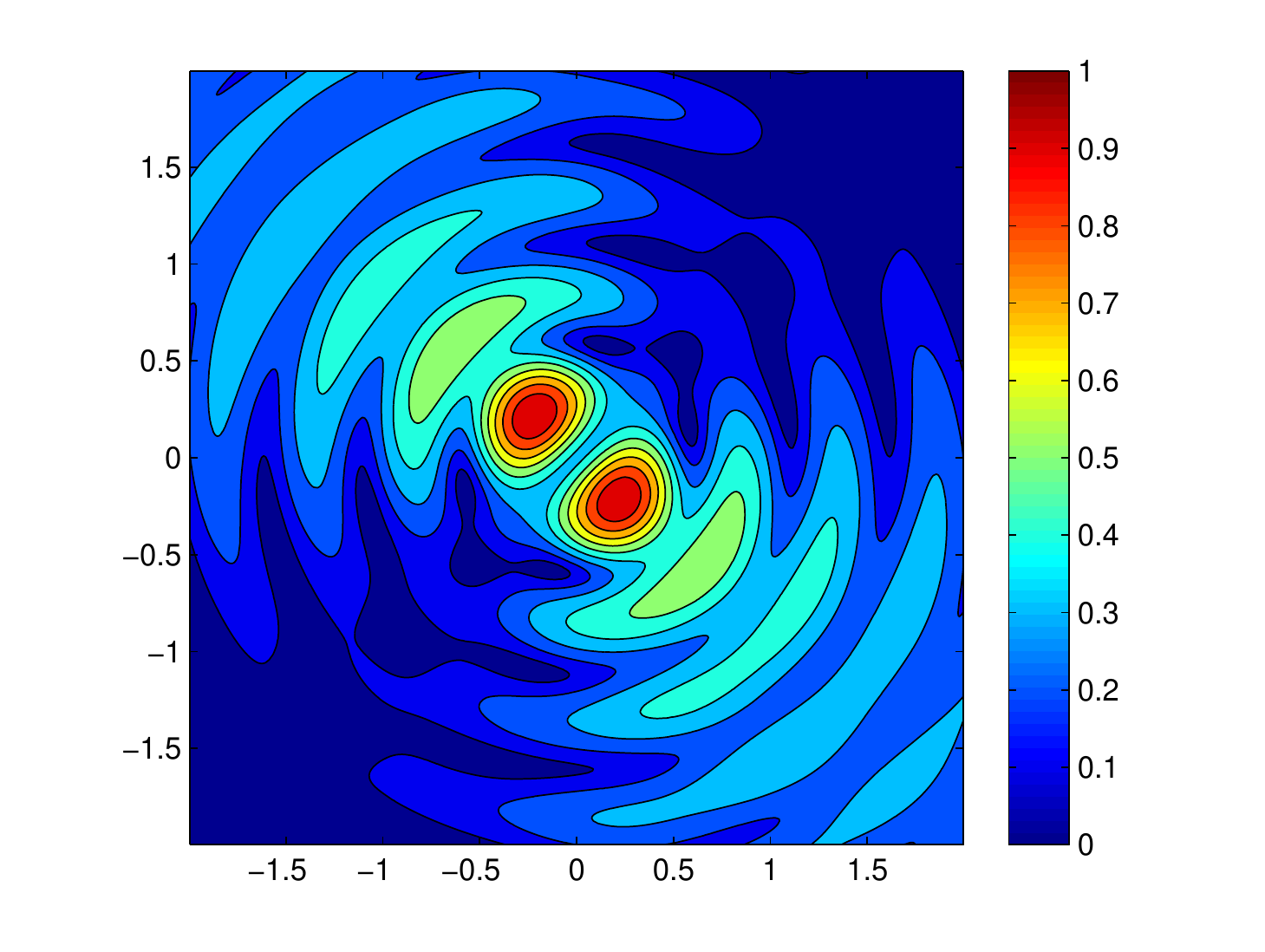}\hfill{}\includegraphics[width=0.33\textwidth]{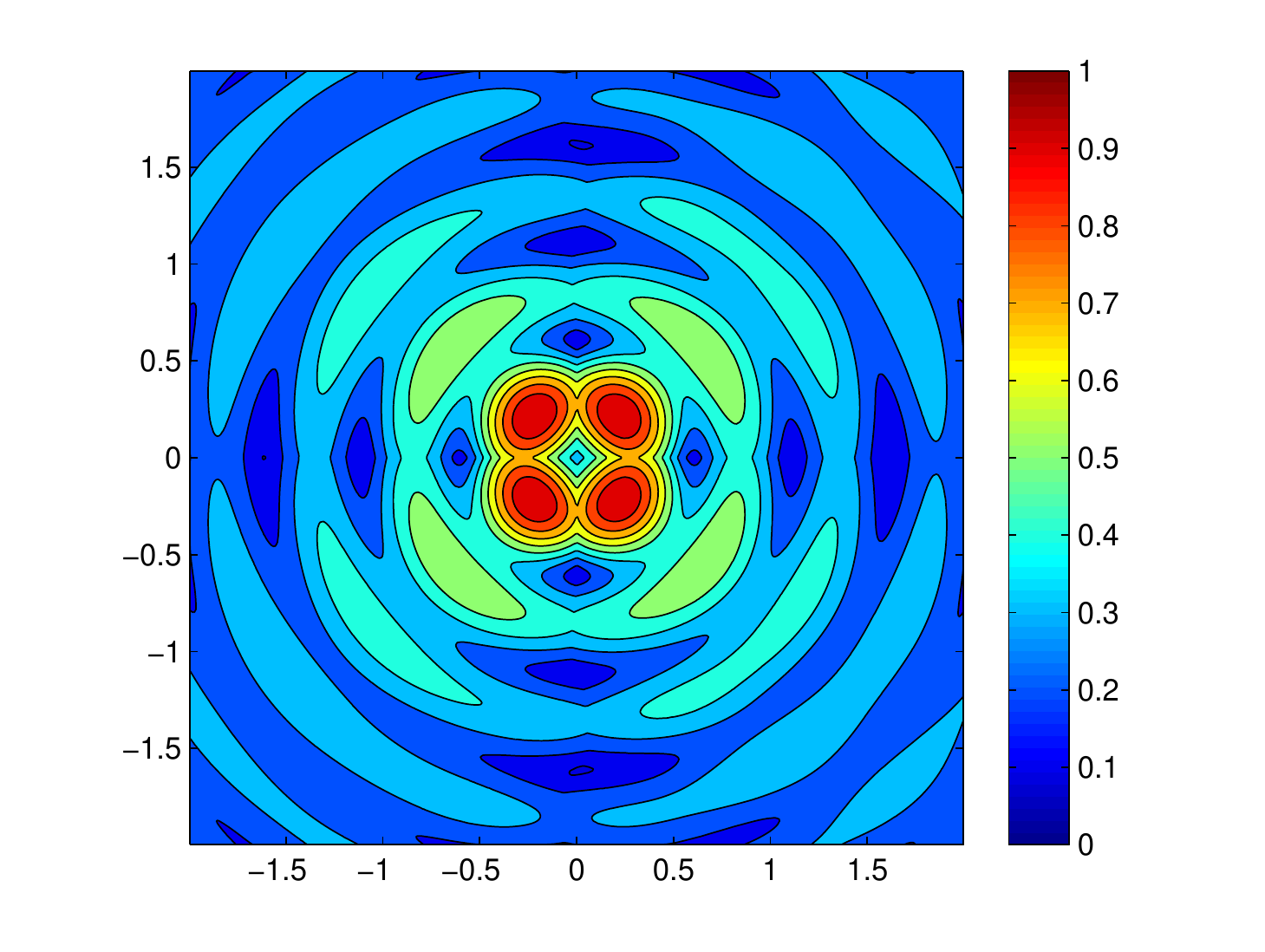}\hfill{}\hfill{}

\hfill{}(b)\hfill{}\hfill{}(c)\hfill{}\hfill{}(d)\hfill{}

\hfill{}\includegraphics[width=0.33\textwidth]{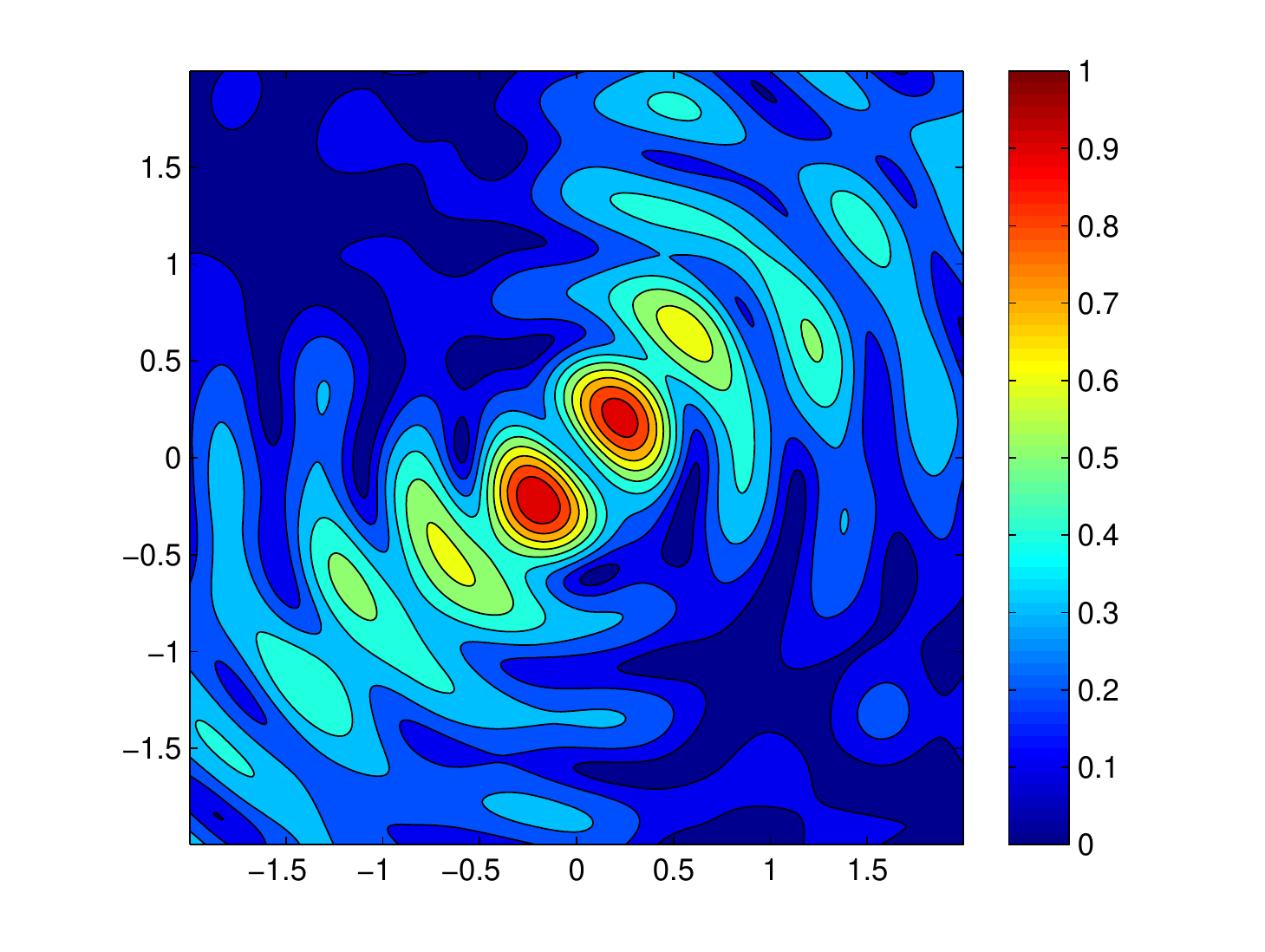}\hfill{}\includegraphics[width=0.33\textwidth]{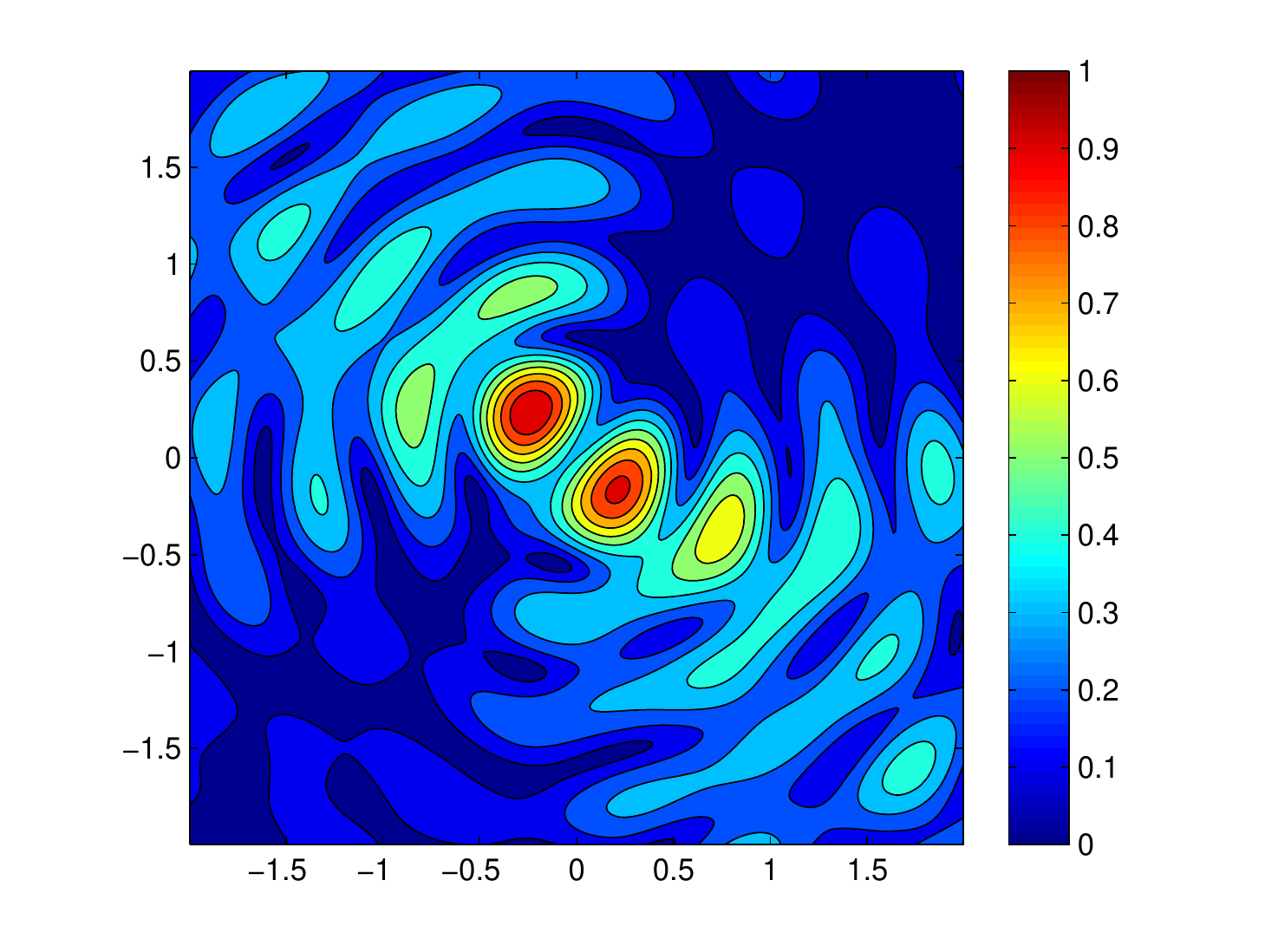}\hfill{}\includegraphics[width=0.33\textwidth]{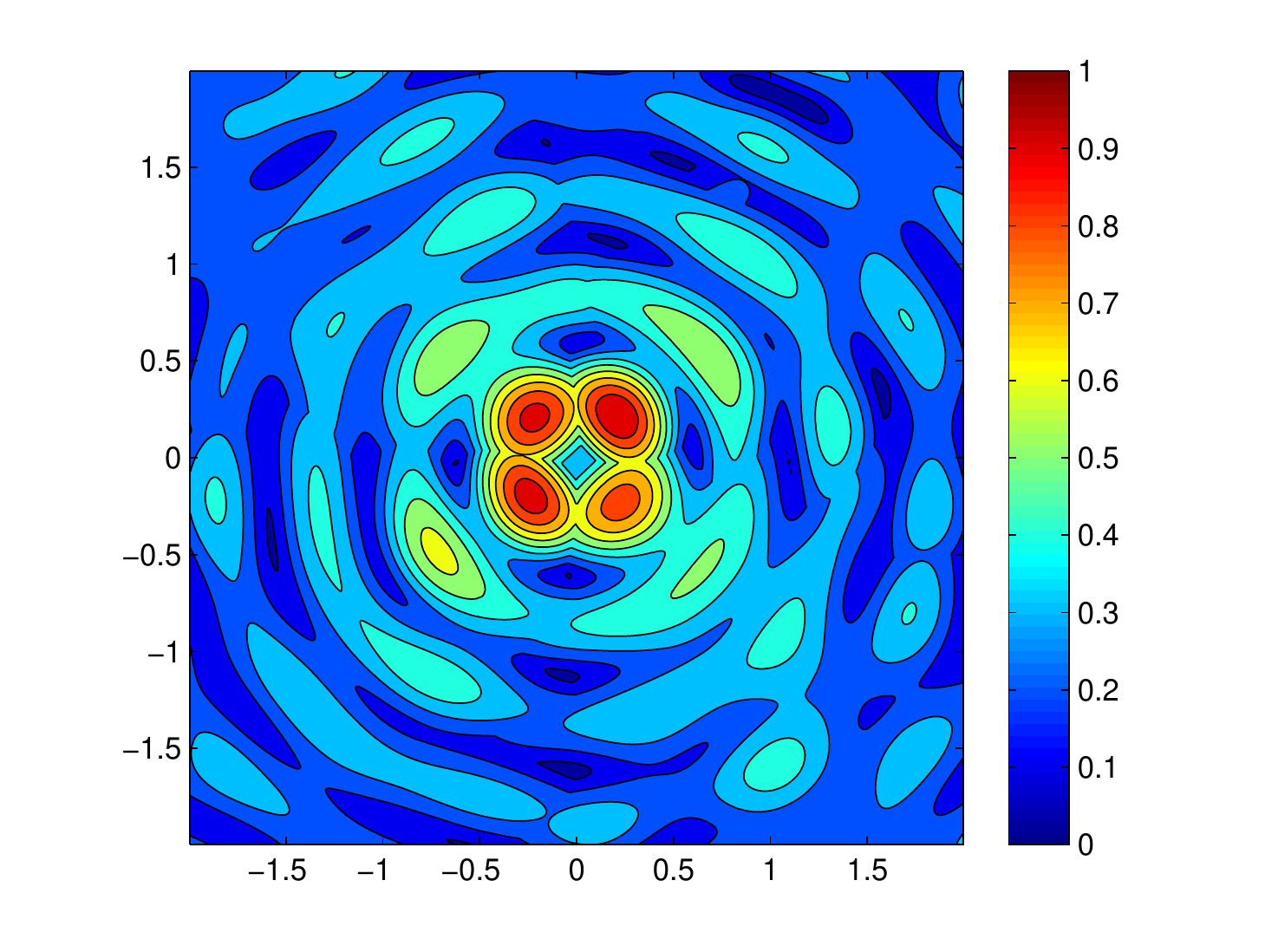}\hfill{}\hfill{}

\hfill{}(e)\hfill{}\hfill{}(f)\hfill{}\hfill{}(g)\hfill{}

\caption{\label{fig:ex4:near} Example 4: (a) true scatterer; reconstruction
results using exact near-field data (b), (c), (d), and noisy data
(e), (f), (g) with $\epsilon=20\%$; incident directions $d_{1}$,
$d_{2}$ and ($d_{1},d_{2}$) used respectively for (b) \& (e), (c)
\& (f) and (d) \& (g).}
\end{figure}

\begin{figure}
\hfill{}\includegraphics[width=0.25\textwidth]{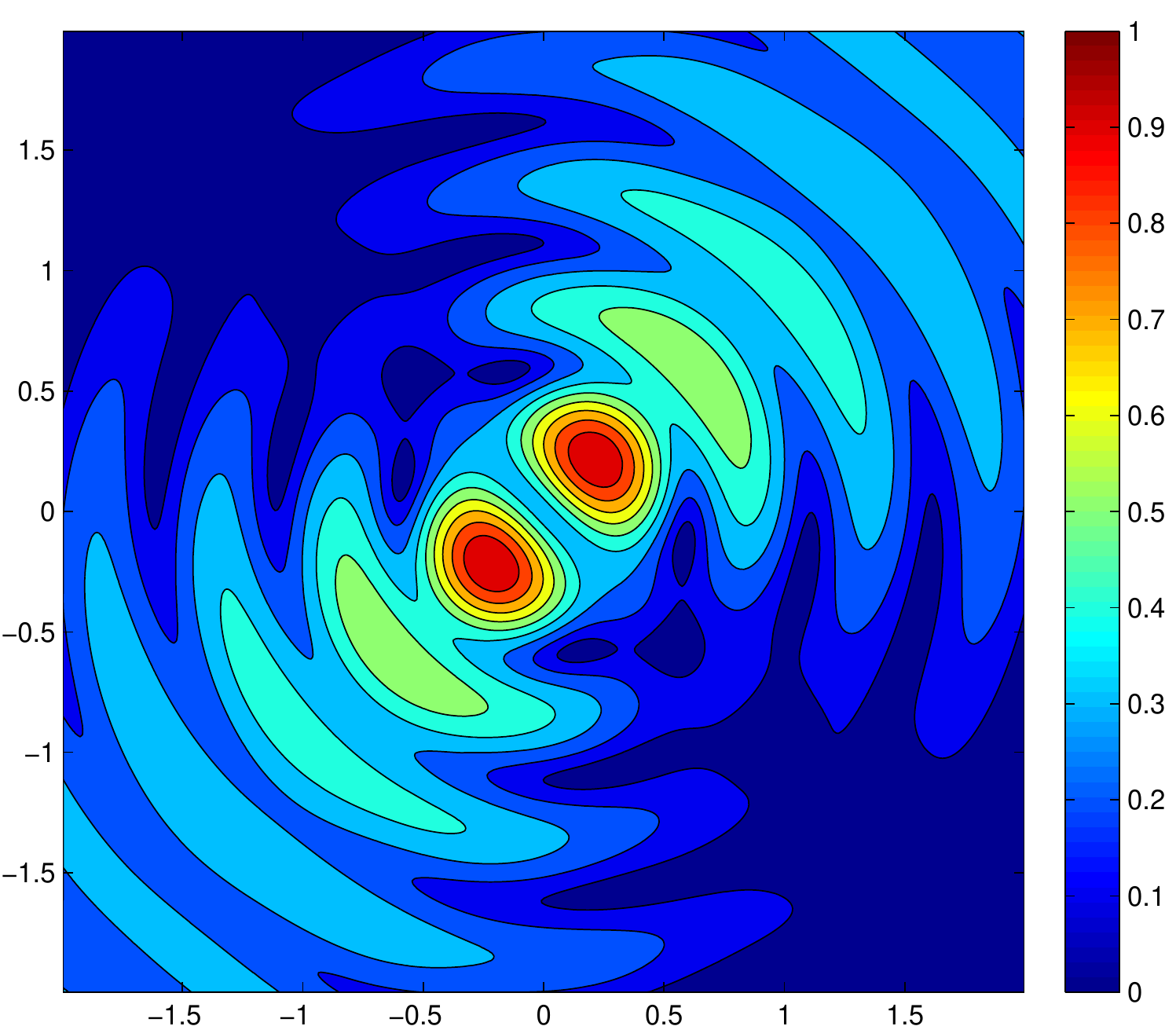}\hfill{}\includegraphics[width=0.25\textwidth]{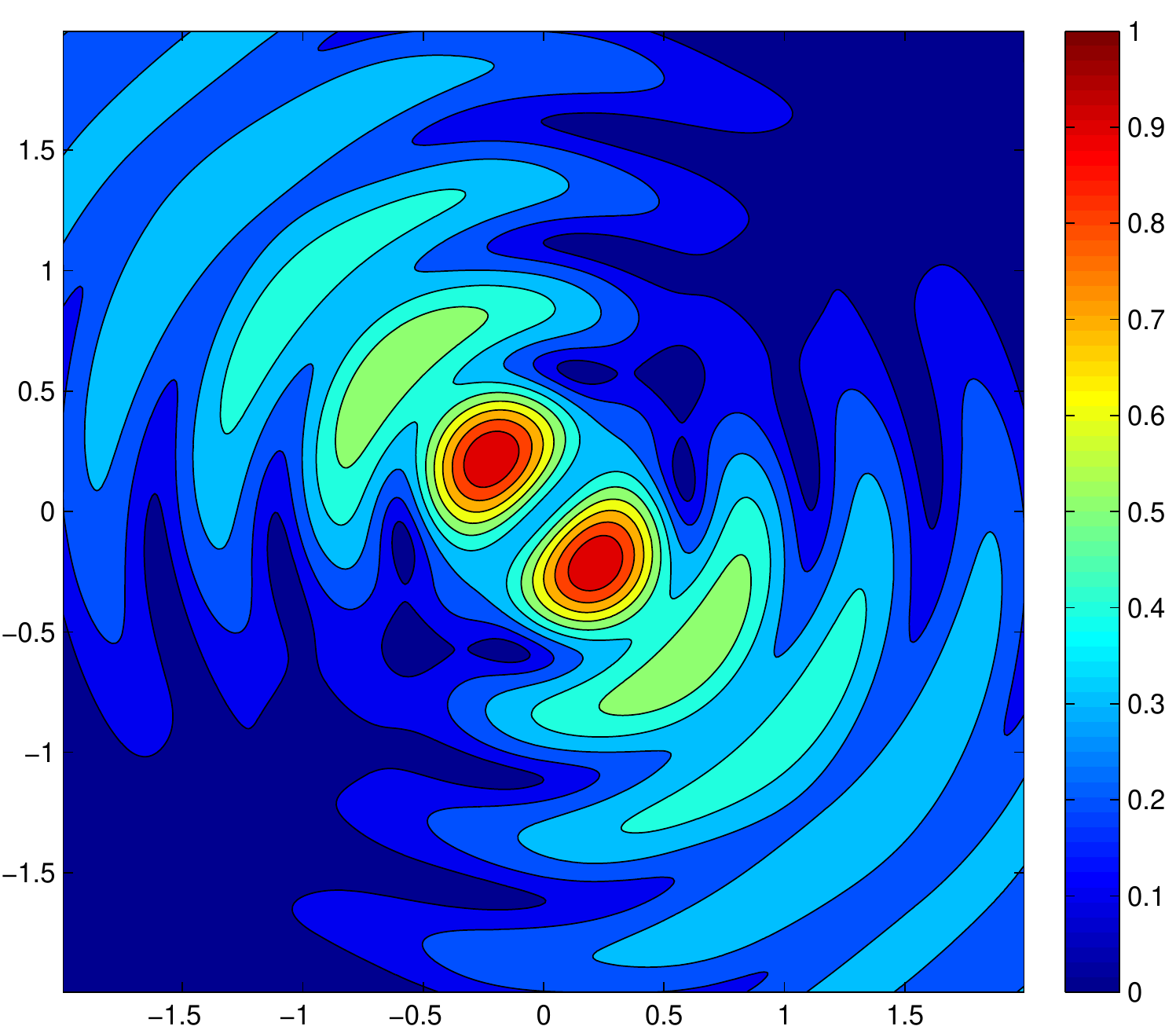}\hfill{}\includegraphics[width=0.25\textwidth]{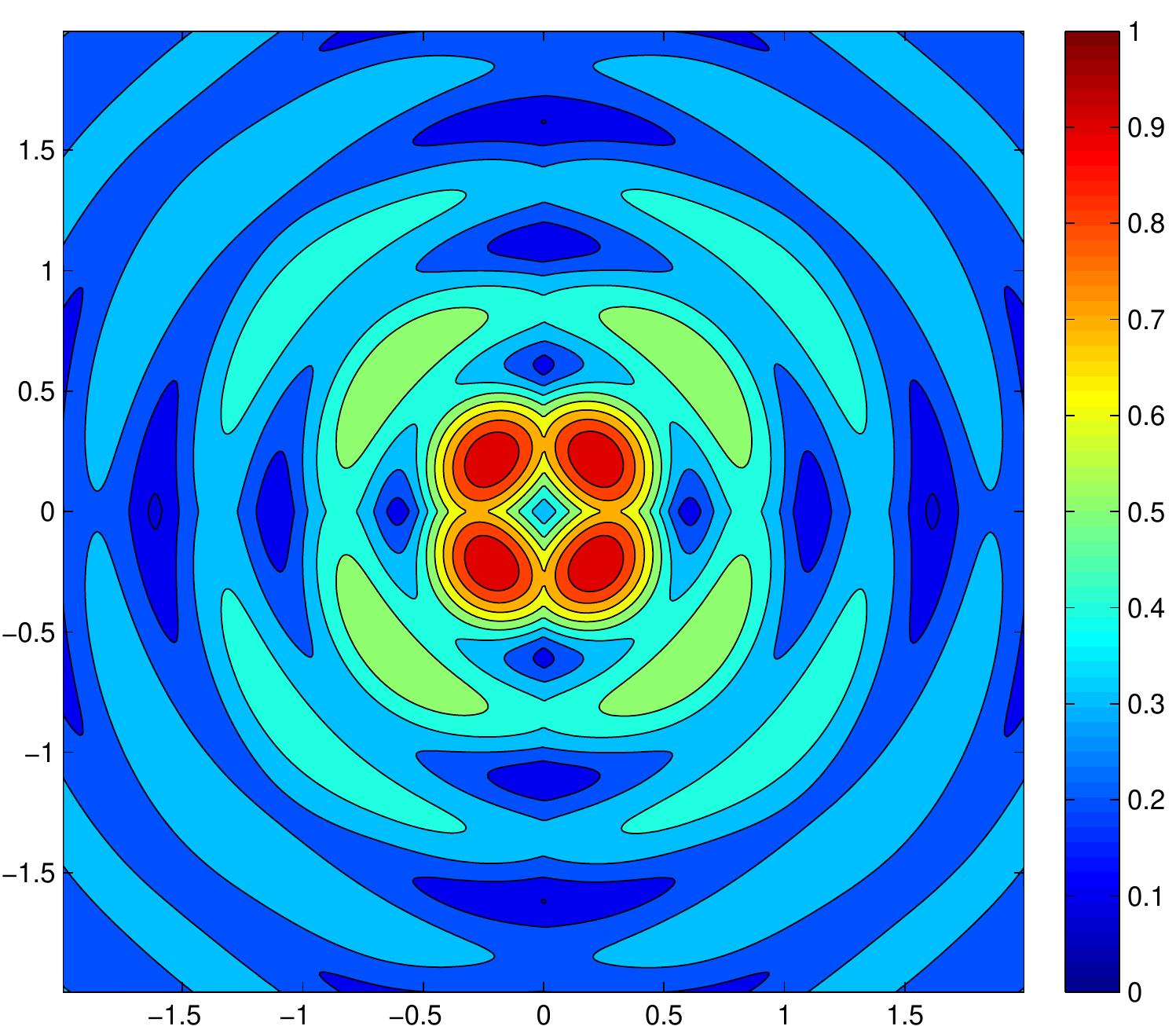}\hfill{}\hfill{}

\hfill{}\!\!\!\!\!\!\!\!\!\!\!\!\!\!\!\!(a)\hfill{}
~~~~~(b)\hfill{}~~~~~~~~ (c)\hfill{}

\hfill{}\includegraphics[width=0.25\textwidth]{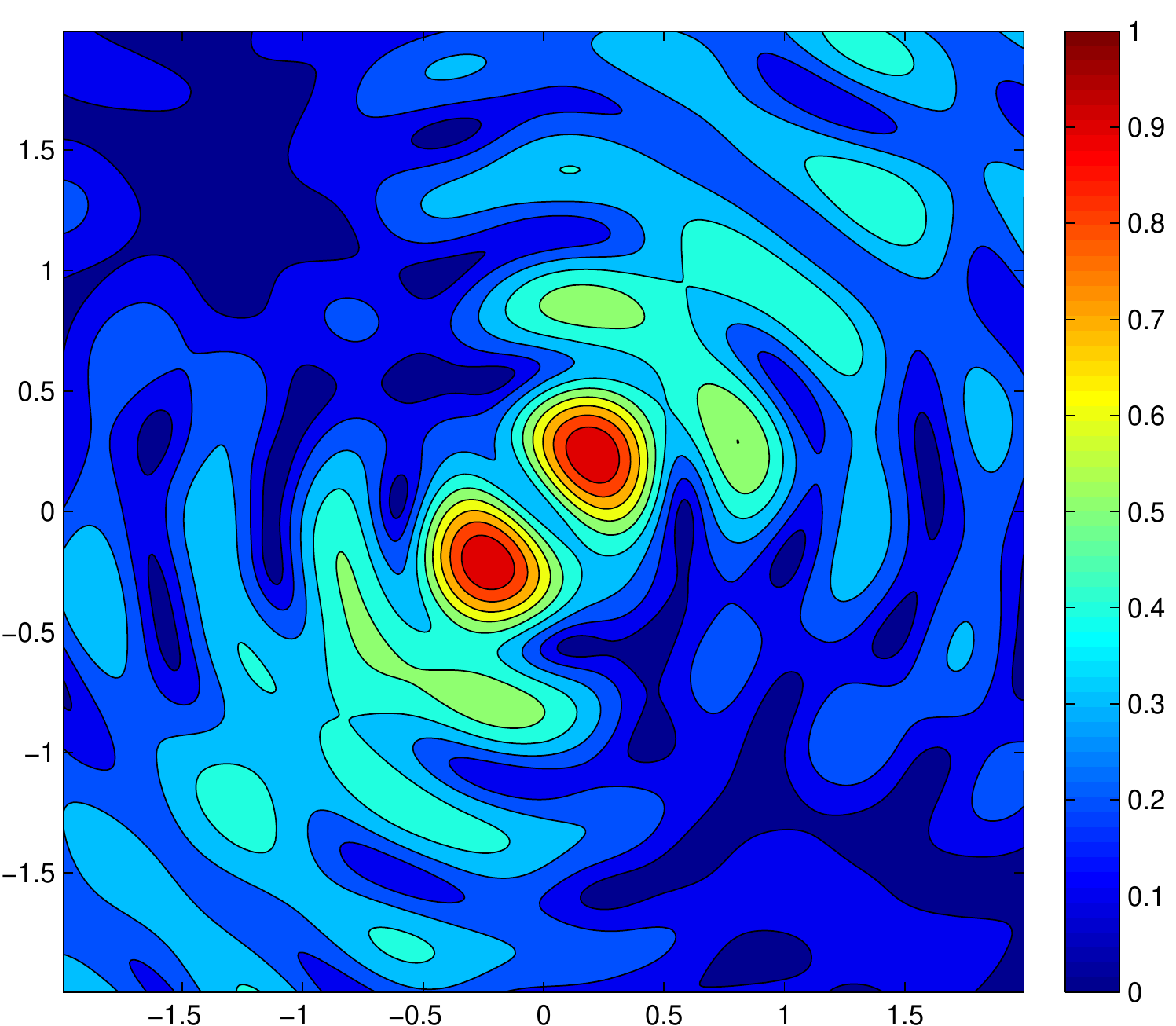}\hfill{}\includegraphics[width=0.25\textwidth]{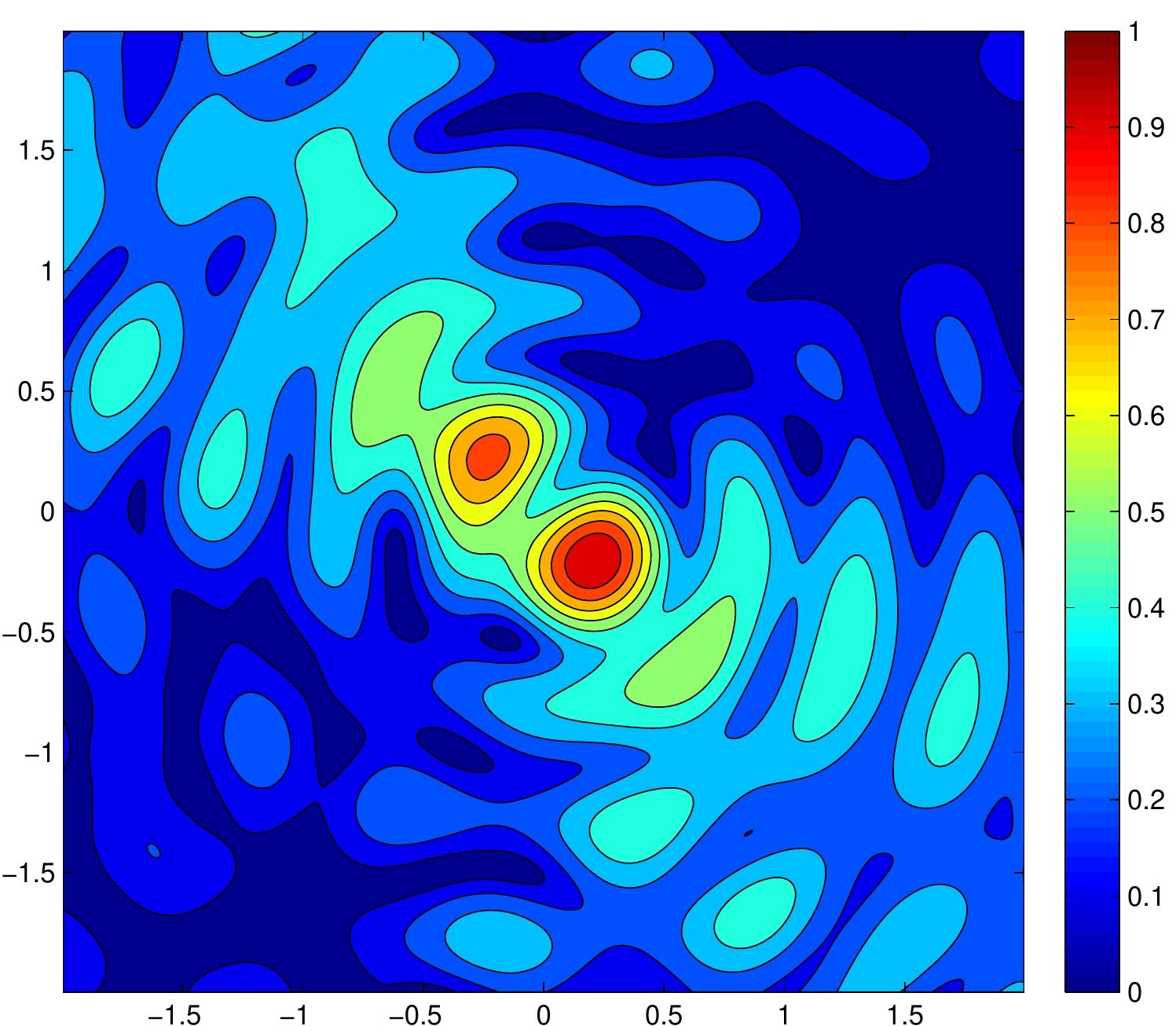}\hfill{}\includegraphics[width=0.25\textwidth]{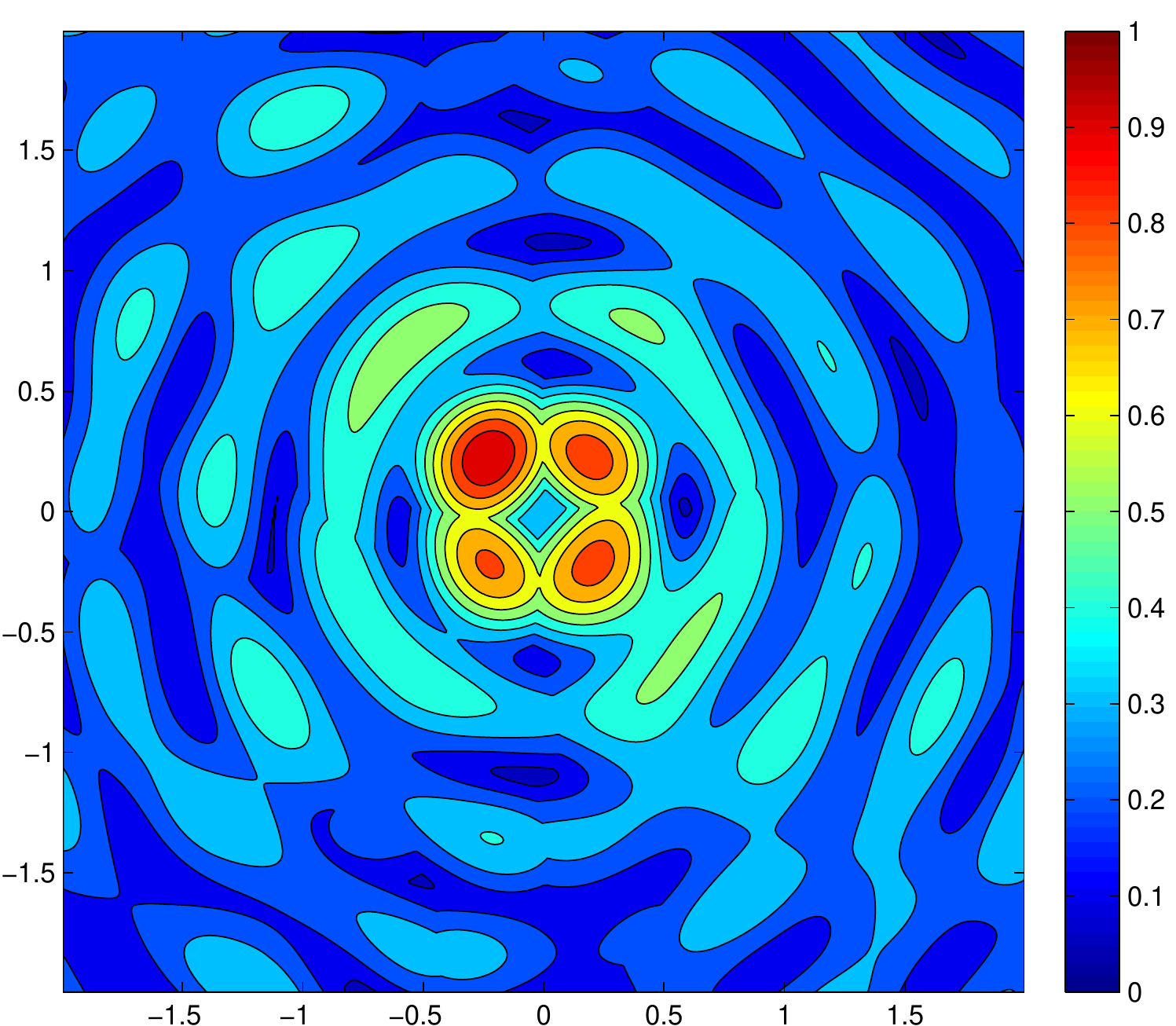}\hfill{}\hfill{}

\hfill{}\!\!\!\!\!\!\!\!\!\!\!\!\!\!\!\!(d)\hfill{}
~~~~~(e)\hfill{}~~~~~~~~ (f)\hfill{}

\caption{\label{fig:ex4:far} Example 4: reconstruction results using
exact far-field data (a), (b), (c), and noisy data (d), (d), (f)
with $\epsilon=20\%$; incident directions $d_{1}$, $d_{2}$ and
($d_{1},d_{2}$) used respectively for (a) \& (d), (b) \& (e) and (c)
\& (f).}
\end{figure}

\smallskip{}

From the above four examples, we see that the DSM(f) has similar performance
as the DSM(n).  Both methods work stably and can tolerate strong noise.
These estimated locations and index values can serve as good initial
guesses for further reconstructions of more accurate locations and
inhomogeneity profile $\eta(x)$ of scatterers through computationally
much more demanding resolution processes such as nonlinear optimizations.

\subsection*{Physical interpretation}

Before we study more challenging examples,  we try to provide a physical justification
of the DSM using both near-field and far-field data based on some numerical observations.

Let us first consider the
key factor $\Im\left(G(x_{p},x_{j})\right)$ in the correlation measure
\eqref{eq:uinf:reason}, which can be explicitly written (with an
appropriate scaling factor) as:
\begin{equation}
C_{N}\Im\left(G(x_{p},x_{j})\right)=\begin{cases}
{\displaystyle J_{0}(k|x_{j}-x_{p}|)}, & N=2;\\
\\
{\displaystyle \frac{\sin(k|x_{j}-x_{p}|)}{k|x_{j}-x_{p}|}}, & N=3,
\end{cases}\label{eq:Imag:G}
\end{equation}
 where $C_{N}=4$ and $4\pi$ for $N=2$ and $3$, respectively. Figure\ \ref{fig:decay}
shows that the leading term $C_{N}\,\Im\left(G(x_{p},x_{j})\right)$
of the correlation measure achieves the maximum when $x_{p}$ approaches
$x_{j}$, and decays quickly as the distance between them increases.

\begin{figure}
\hfill{}\includegraphics[width=0.48\textwidth]{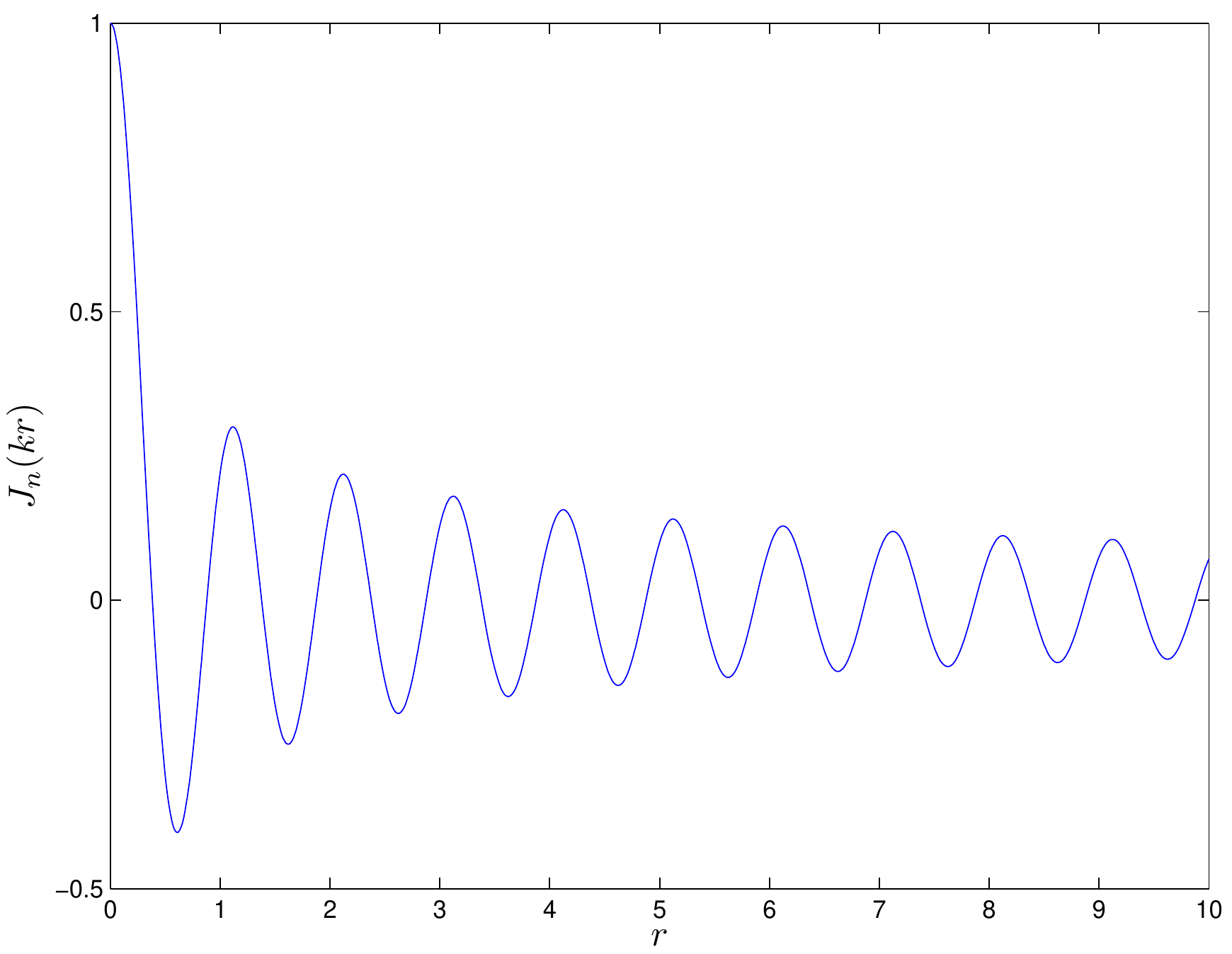}\hfill{}\includegraphics[width=0.48\textwidth]{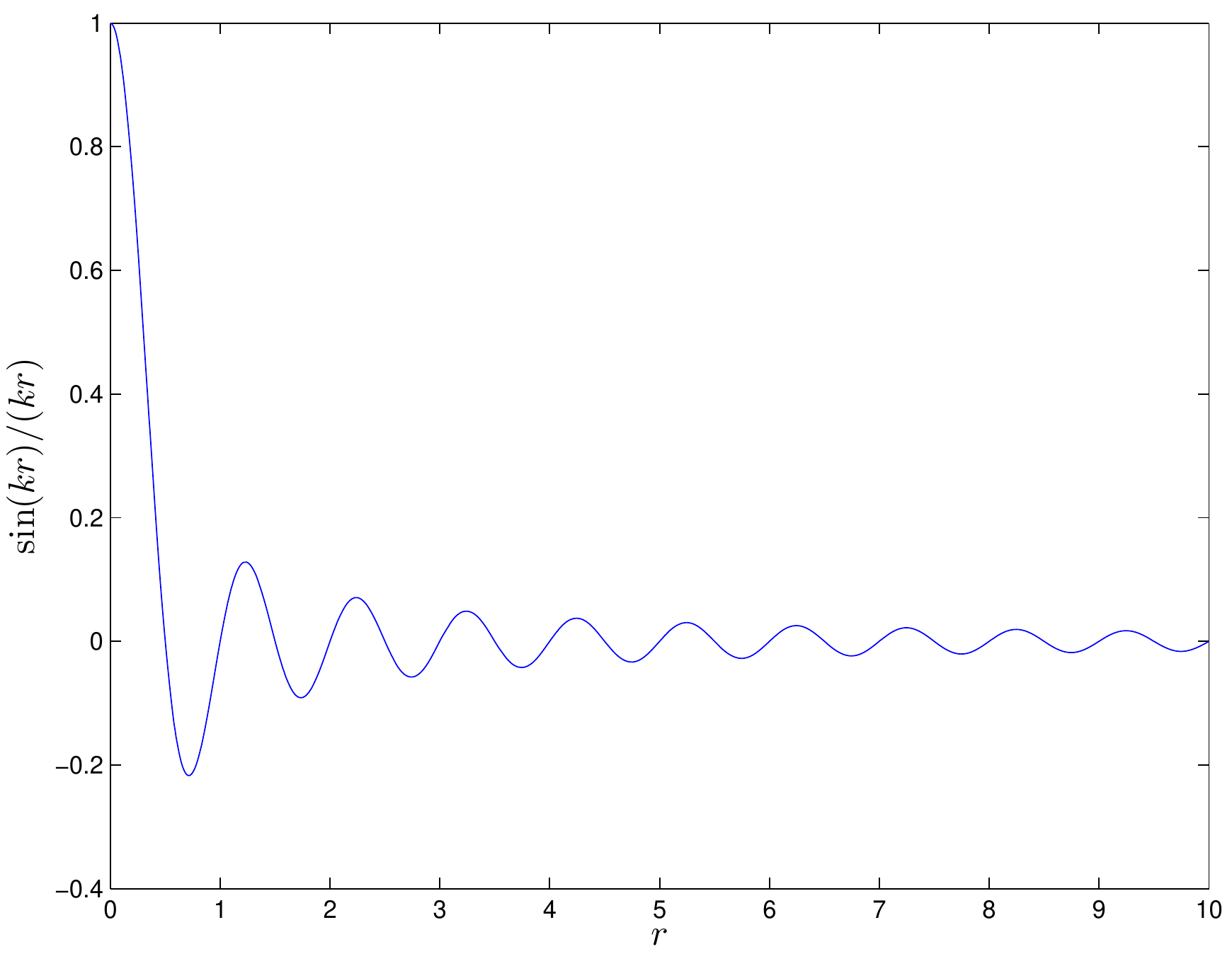}\hfill{}\hfill{}

\hfill{}(a)\hfill{}\hfill{}(b)\hfill{}

\caption{\label{fig:decay}Decay patterns of
$C_{N}\,\Im\left(G(x_{p},x_{q})\right)$ as $r=|x_{p}-x_{q}|$
increases in two dimensions (a) and three dimensions (b) when $k=2
\pi$. }
\end{figure}

Let us now first explain the DSM using far-field data. From Figure\ \ref{fig:decay},
it can be seen clearly that there is much stronger pattern correlation
between $u^{\infty}(\hat{x})$ and $G^{\infty}(\hat{x},x_p)$ in \eqref{eq:uinf:reason} when $r=|x_{p}-x_{j}|$ is small, namly  the sampling point $x_{p}$
is  sufficiently close to the center $x_{j}$ of some inhomogeneous sampling element
where $I_{j}$ does not vanish due to the non-vanishing $\eta(x_{j})$. Let us take Example
1 to illustrate this key relation.
We plot in Figure~\ref{fig:ratio:uinf:Ginf} the pointwise  complex ratio of the far field data $u^{\infty}(\hat{x})$
to the far field pattern $G^{\infty}(\hat{x},x_p)$  of the  fundamental solutions located at the sampling
point $x_p$ being the origin.
This complex ratio is computed with respect to 50 observation angles in $[0,2\pi]$ in the polar form with its polar radius and phase angle depicted in the left and right plots, respectively. It can be easily observed that both the polar radius  and the phase angle in Figure~\ref{fig:ratio:uinf:Ginf} approaches some constants. In other words,  $u^{\infty}(\hat{x})$ and $G^{\infty}(\hat{x},x_p)$  are almost linear dependent with each other.

\begin{figure}
\hfill{}\includegraphics[width=0.47\textwidth]{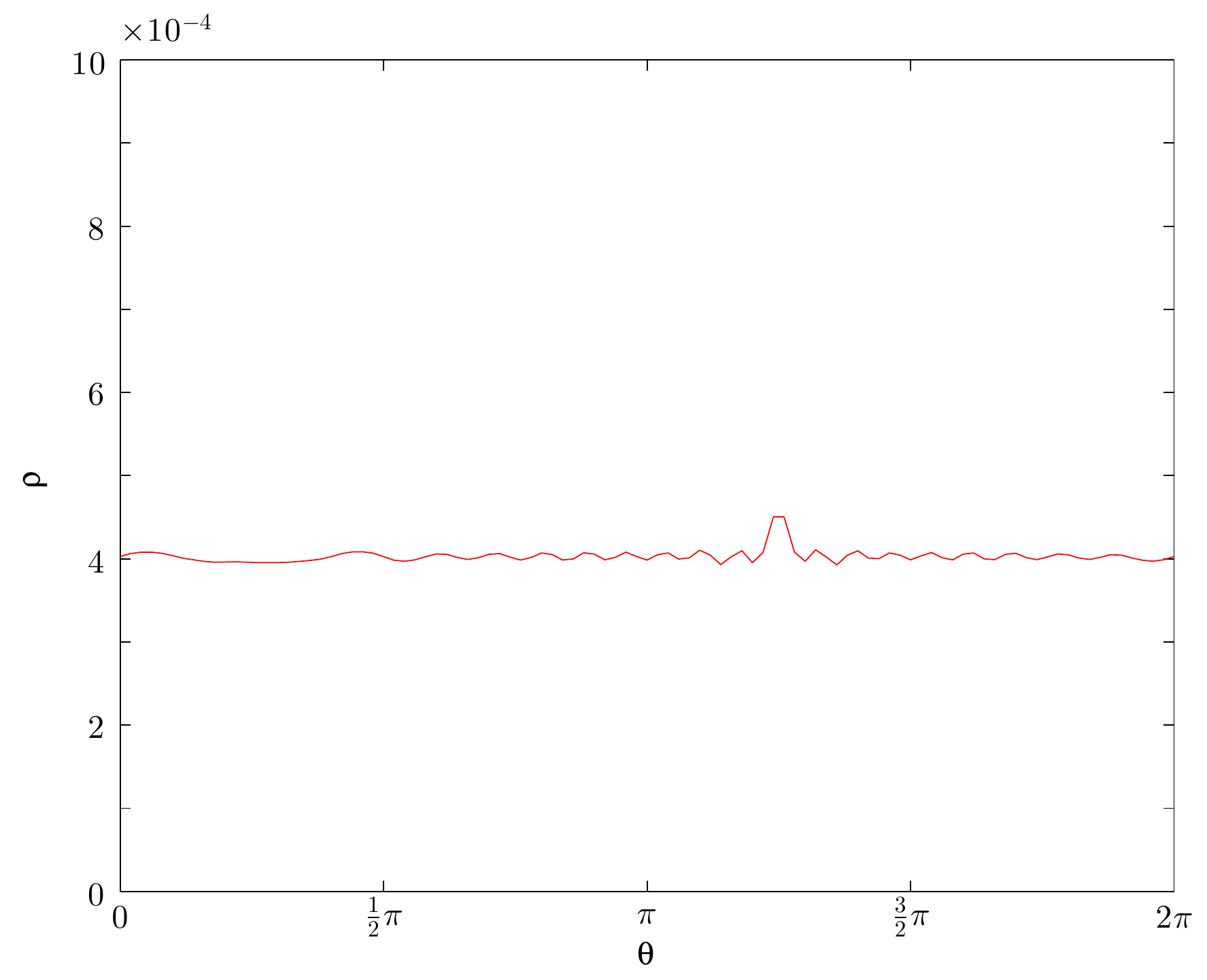}\hfill{}\includegraphics[width=0.49\textwidth]{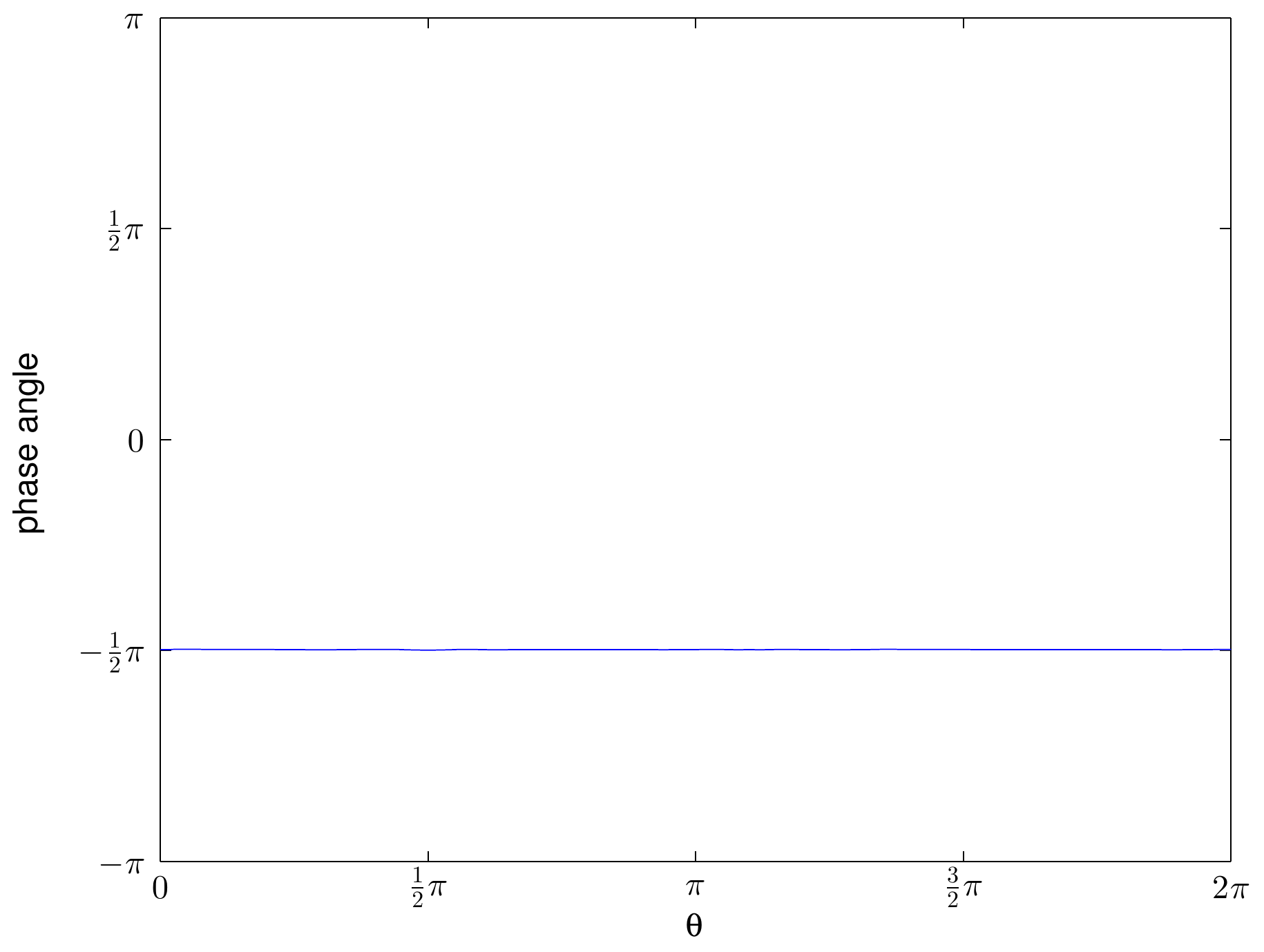}\hfill{}\hfill{}

\hfill{}(a)\hfill{}\hfill{}(b)\hfill{}

\caption{\label{fig:ratio:uinf:Ginf} Pointwise complex ratio ((a) polar radius; (b) phase angle) of $u^{\infty}(\hat{x})$ to  $G^{\infty}(\hat{x},x_p)$ versus all the observation angle $\theta$. Here $x_p$ is chosen to be the origin, the center of the small inhomogeneous scatterer in Example 1.}
\end{figure}

For multi-component small-scaled medium scatterers with sufficiently large distance
between each other, multiple scattering is weak and thus
ignored, so  the exact  scattered wave field
$u^{s}$ can be approximated by its reduced counterpart, namely a superposition of small point sources  induced
by the incident plane wave illuminating the scatterers.  Thanks to the superposition principle
of the wave phenomena, some scattered wave by an individual small
medium scatterer shares common decay pattern with $G(x,x_{p})$ when $x_{p}$
falls within the inhomogeneous medium, which leads to the parallelism
of some component of $u^{\infty}(\hat{x})$ and $G^{\infty}(\hat{x},x_{p})$
in the trace space on $\mathbb{S}^{N-1}$ as shown in Figure~\ref{fig:ratio:uinf:Ginf}. Physically speaking, it amounts
to exciting point sources when a plane wave impinges on some small-scaled
media (size of media is less then a half wavelength). Then the far-field
pattern $G^{\infty}(\hat{x},x_{p})$ of the fundamental solution becomes
a major component of the far-field $u^{\infty}$ when $x_{p}$ is
located within the scatterer. Therefore, if the indicator function $\Phi^{\infty}(x_{p})$
approaches $1$ for some sampling point $x_{p}$, it is plausible
to claim that $x_{p}$ is within the scatterer obstacle or medium.

The reasoning for the DSM using near-field data can be viewed as
a direct consequence of
that of its far-field version. Again we use Example 1
for illustration by taking  the sampling point $x_p$ to be the origin. Under the assumption of sufficiently large
distance between the measurement surface $\Gamma$ and the small scatterer inside
(normally four or five times the wavelength), the traces of $u^{s}(x)$ and
$G(x,x_{p})$ on $\Gamma$ depend approximately linearly on their
asymptotic amplitudes $u^{\infty}(\hat{x})$ and
$G^{\infty}(\hat{x},x_{p})$, respectively, up to some complex scaling
factors.  This can be verified by checking the pointwise complex ratio of $u^{s}(x)$
to $u^{\infty}(\hat{x})$ and that of $G(x,x_{p})$ to $G^{\infty}(\hat{x},x_{p})$, respectively,
as showin in
Figure~\ref{fig:ratio:us:Gs}. By recalling the nearly linear dependence of  $u^{\infty}(\hat{x})$ and  $G^{\infty}(\hat{x},x_{p})$ and the transitivity of the three linear approximations aforementioned, it leads naturally
to the nearly linear dependence of the traces of $u^{s}(x)$ and
$G(x,x_{p})$ on $\Gamma$, which explain the effectiveness of the indicator function \eqref{eq:indicator:near} in \cite{IJZ12}.

For multi-component scatterers, as in the far field case, the scattering field $u^s$
may well reduce to a superposition of small point sources  induced
by the incident plane wave illuminating the scatterers plus some high order terms due to
the multiple scattering.
When the sampling point $x_p$ falls within a medium scatterer, one
can obtain a nearly linear dependence of a significant component of the scattering field $u^s$
on the trace of  the point source $G(x,x_{p})$ restricted on $\Gamma$.     After
normalization, such physical correlation of the near fields $u^s$ and $G(x,x_{p})$ leads to a fairly large
value (nearly one) of the indicator function $\Phi(x_{p})$ for the
sampling point $x_{p}$.

\begin{figure}
\hfill{}\includegraphics[width=0.48\textwidth]{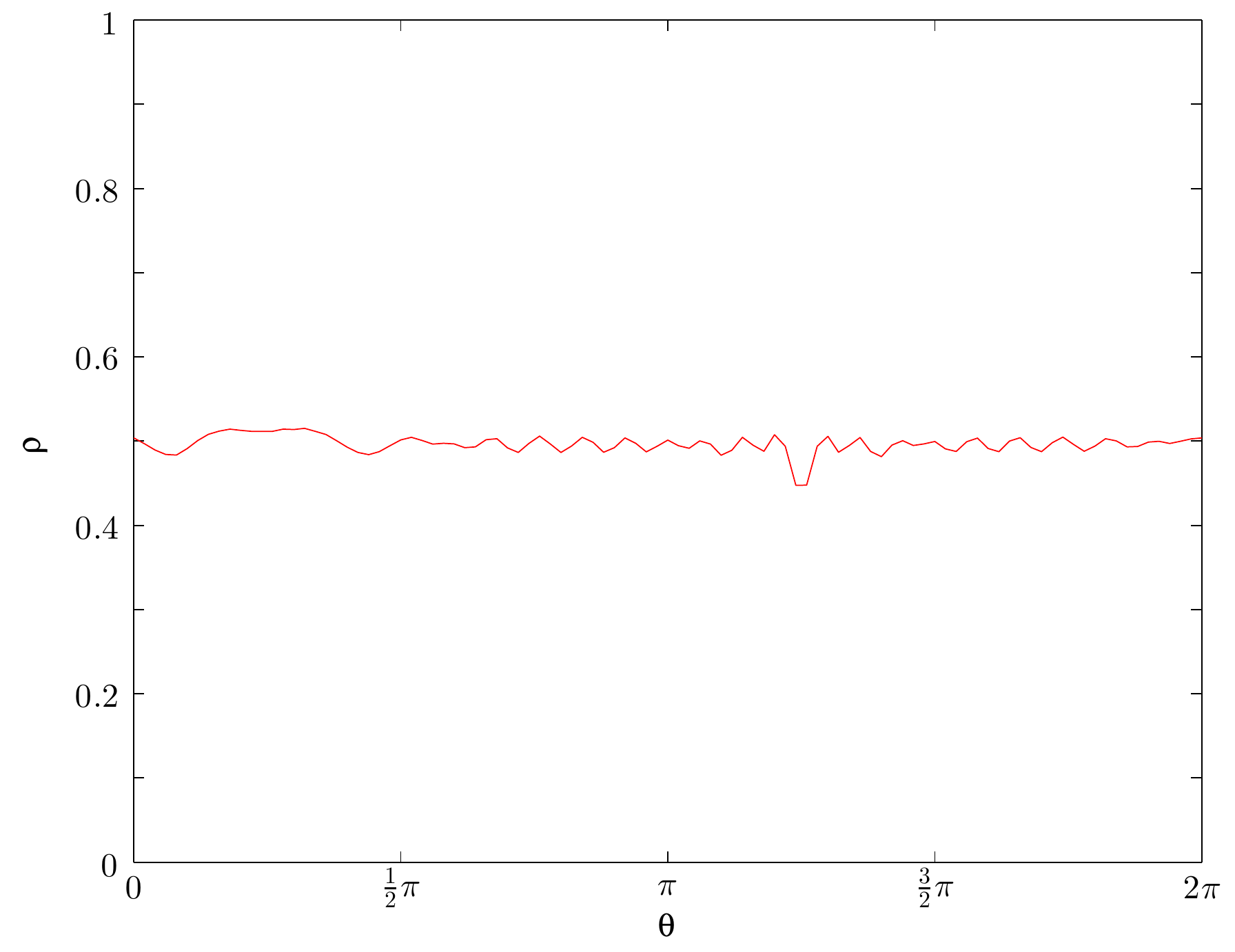}\hfill{}\includegraphics[width=0.48\textwidth]{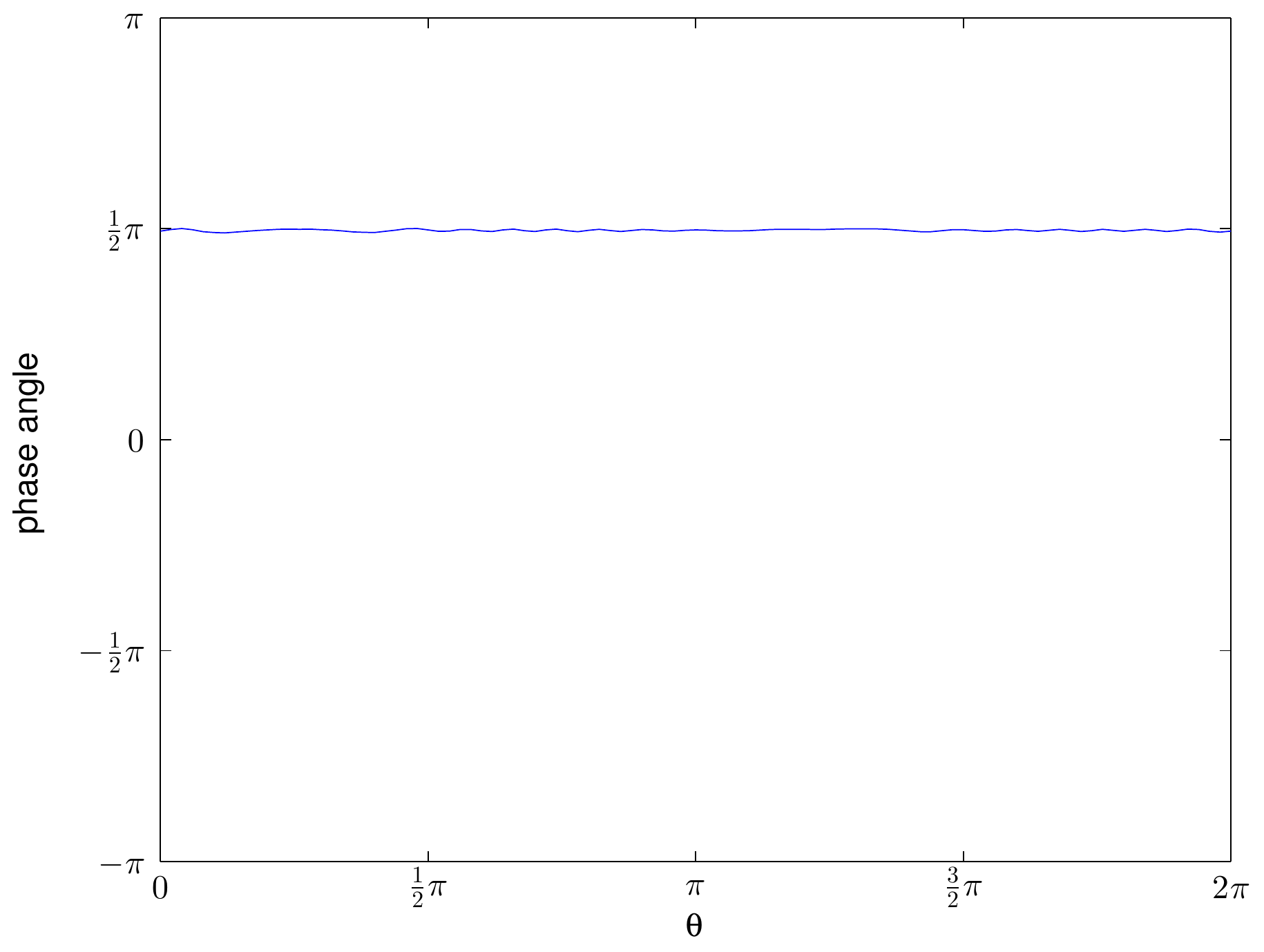}\hfill{}\hfill{}

\hfill{}(a)\hfill{}\hfill{}(b)\hfill{}

\hfill{}

Complex ratio ((a) polar radius; (b) phase angle) of $u^{s}({x})$ to $u^{\infty}(\hat{x})$  versus  the observation angle $\theta$.

\hfill{}\includegraphics[width=0.48\textwidth]{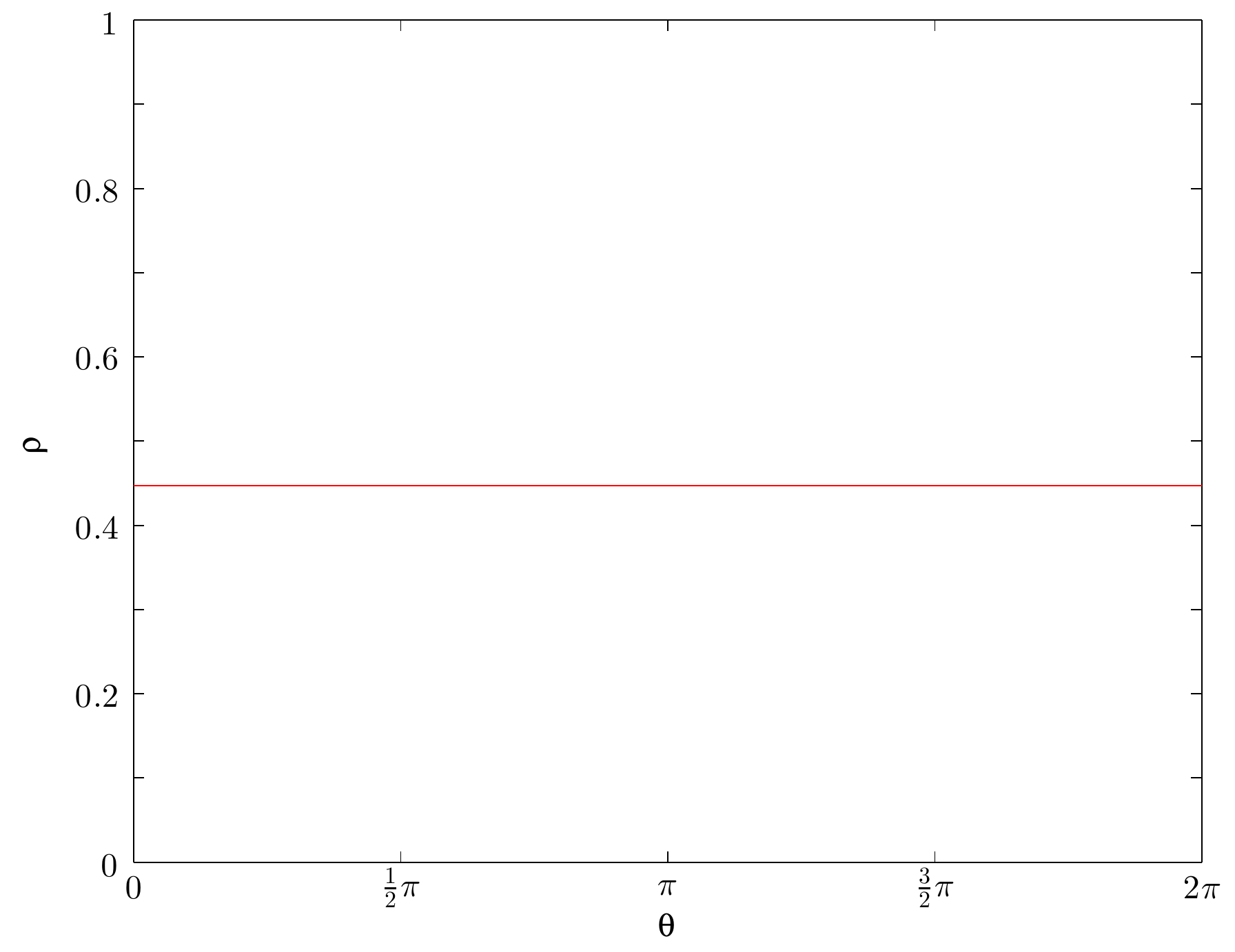}\hfill{}\includegraphics[width=0.48\textwidth]{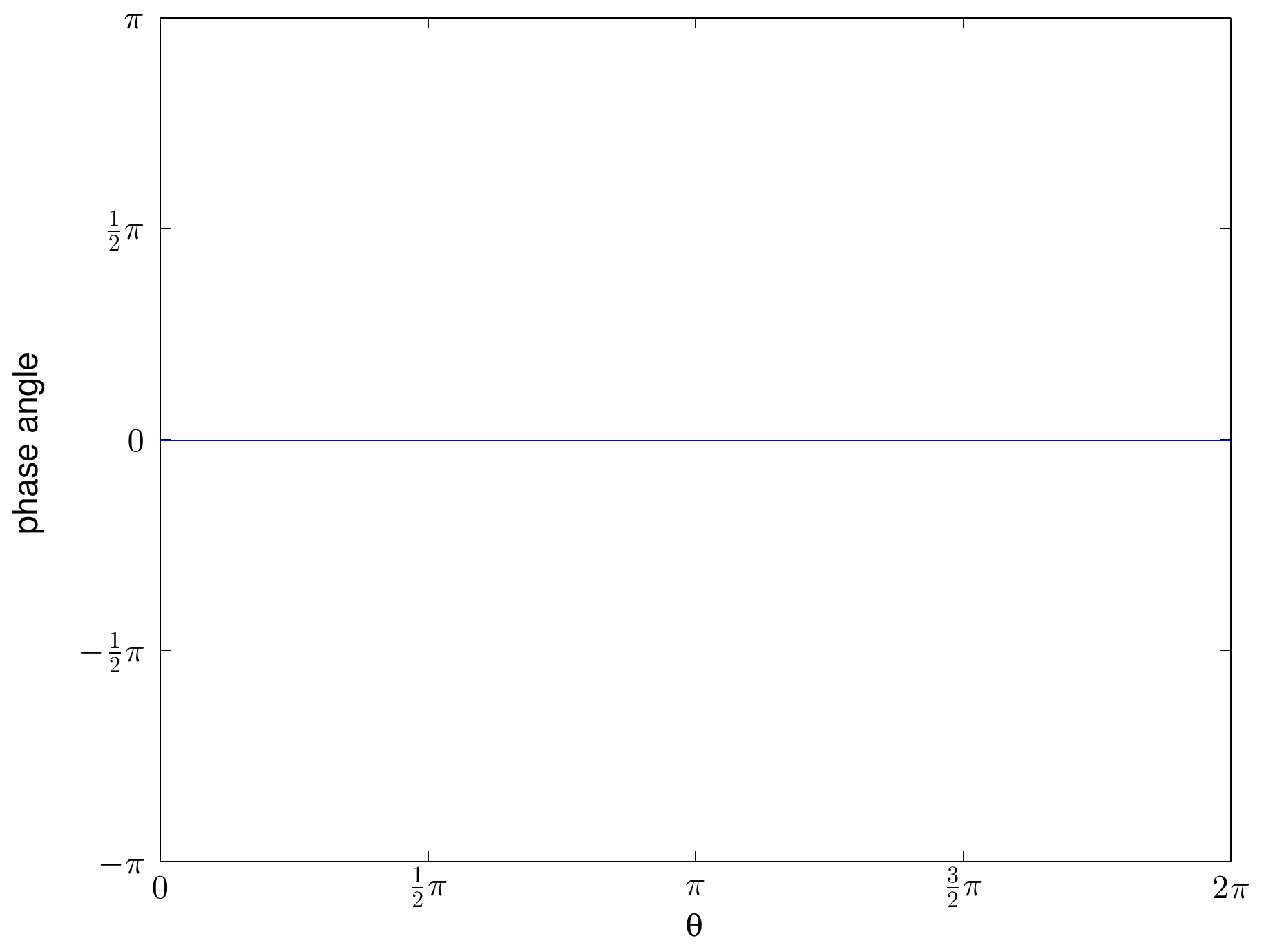}\hfill{}\hfill{}

\hfill{}(c)\hfill{}\hfill{}(d)\hfill{}

Complex ratio ((c) polar radius; (d) phase angle) of $G({x},x_p)$ to $G^{\infty}(\hat{x},x_p)$ versus  the observation angle $\theta$.

\caption{\label{fig:ratio:us:Gs}  Illustration of approximately linear dependence of the traces of $u^{s}(x)$ and
$G(x,x_{p})$ on $\Gamma$  on their
asymptotic amplitudes $u^{\infty}(\hat{x})$ and
$G^{\infty}(\hat{x},x_{p})$, respectively. Here $x_p$ is chosen to be the origin, the center of the small inhomogeneous scatterer in Example 1.}
\end{figure}

\subsection*{DSM for mixed obstacle and medium scatterers}

The next example considers a mixed scatterer composed of an obstacle
and an inhomogeneous medium with variable material coefficient.

\textbf{Example 5 (A mixed scatterer with obstacle and medium)}. The
example considers two square scatterers of side length $0.3\lambda$.
The two scatterers are located at $(-0.8\lambda,\,-0.7\lambda)$ (a
sound-soft obstacle) and $(0.3\lambda,\,0.8\lambda)$ (an inhomogeneous
medium), respectively. The lower left one is an obstacle, which is
excluded from the domain and denoted by a white square. We vary the
coefficient $\eta$ in the upper right inhomogeneous square to study
the wave interaction of obstacle and medium scatterers.

This example verifies the capability of both types of DSMs to identify
obstacles and also shows the condition under which the DSM can identify
both obstacles and media. The results for $\eta=1$ and $n^{2}=10+10\mathrm{i}$
using the DSM(n) and DSM(f) are shown in Figures\ \ref{fig:ex10:1}
and \ref{fig:ex10:1000} respectively. For $\eta=1$, it amounts to a refractive index $n=\sqrt{1+1/(2\pi)^2}$ for the right upper square medium scatter, which is
very  close to the background medium with unit refractive index. The
scattering effect of the inhomogeneous medium scatterer is so small compared to the scattered
wave due to the obstacle square. Therefore, one can only locate correctly
the position of the obstacle and loses track of the information of
the inhomogeneous medium scatterer.

\begin{figure}
\hfill{}\hspace{-0.03\textwidth}%
\begin{tabular}{c}
\includegraphics[clip,width=0.32\textwidth]{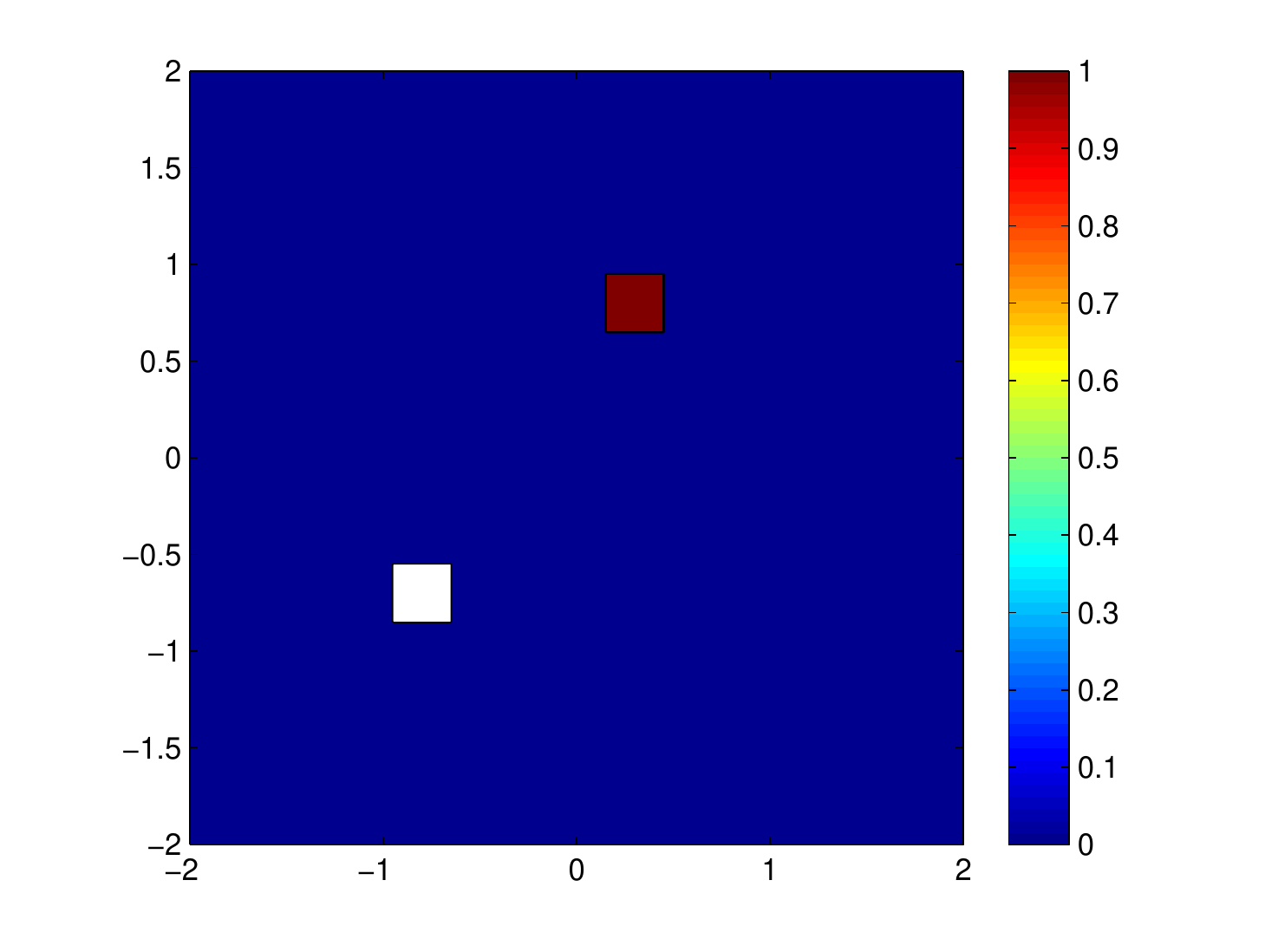}\tabularnewline
(a)\tabularnewline
\end{tabular}\negthinspace{}%
\begin{tabular}{cc}
\includegraphics[width=0.32\textwidth]{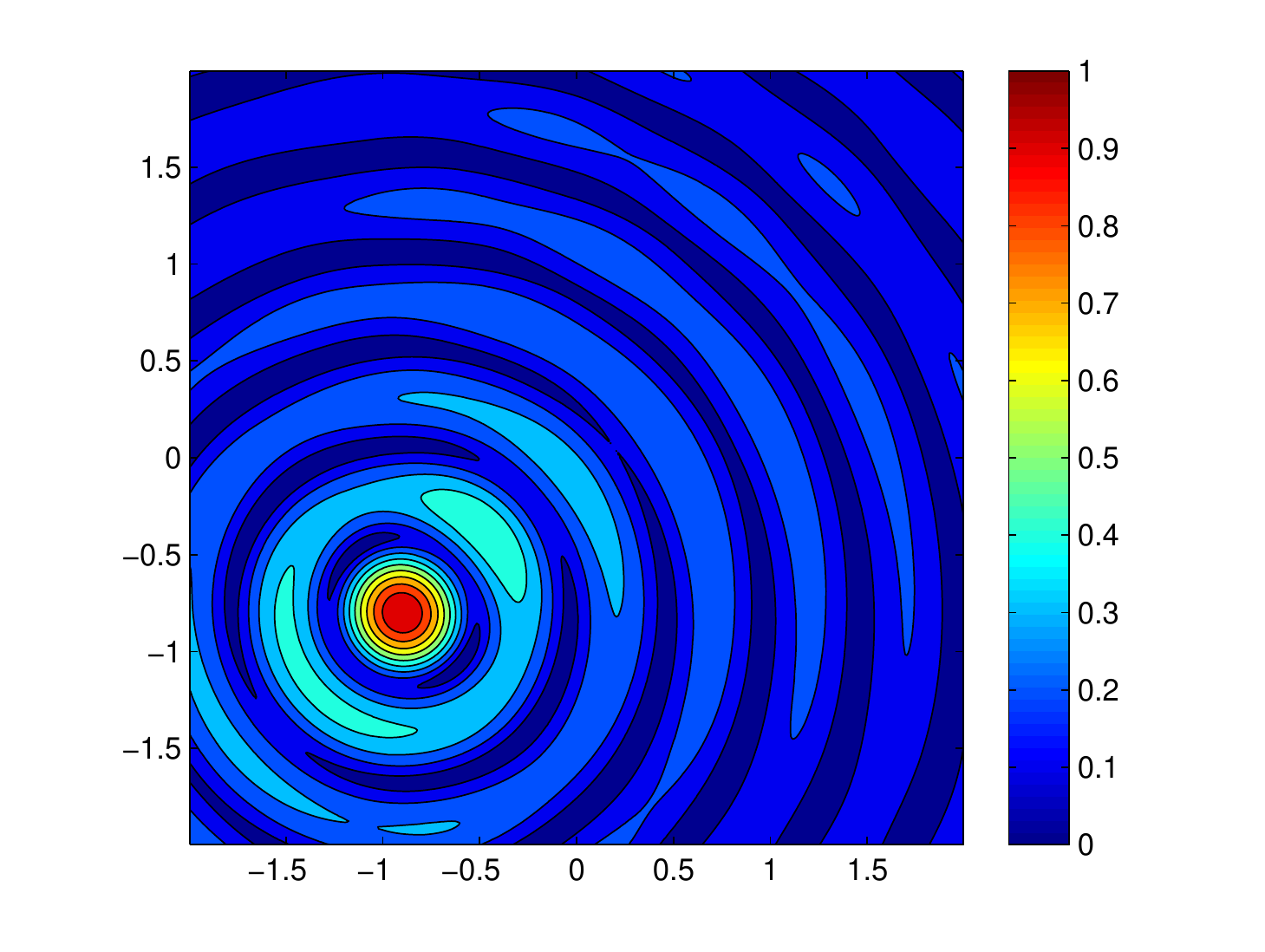} & \includegraphics[width=0.32\textwidth]{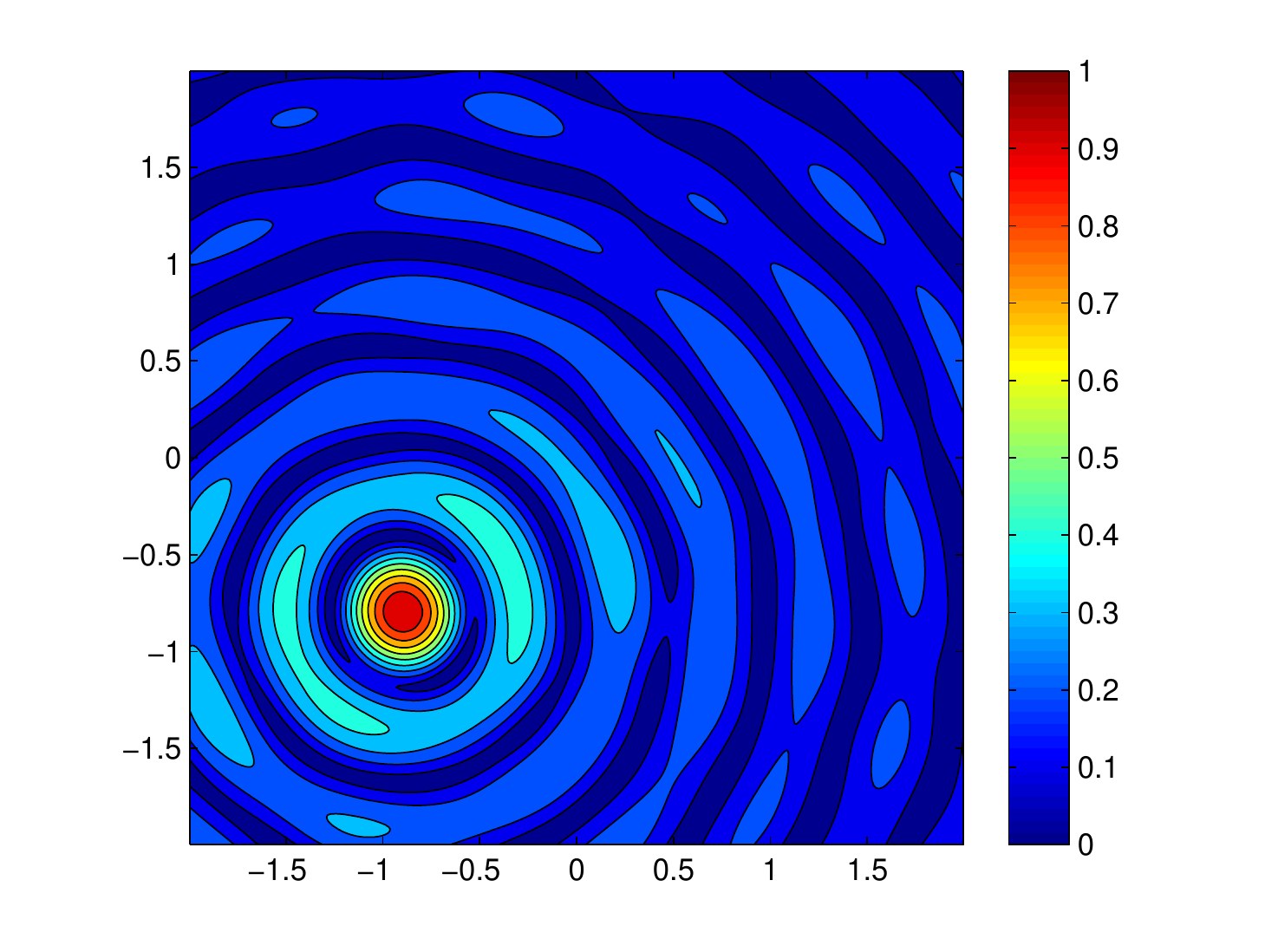}\tabularnewline
(b) & (c)\tabularnewline
\includegraphics[width=0.32\textwidth]{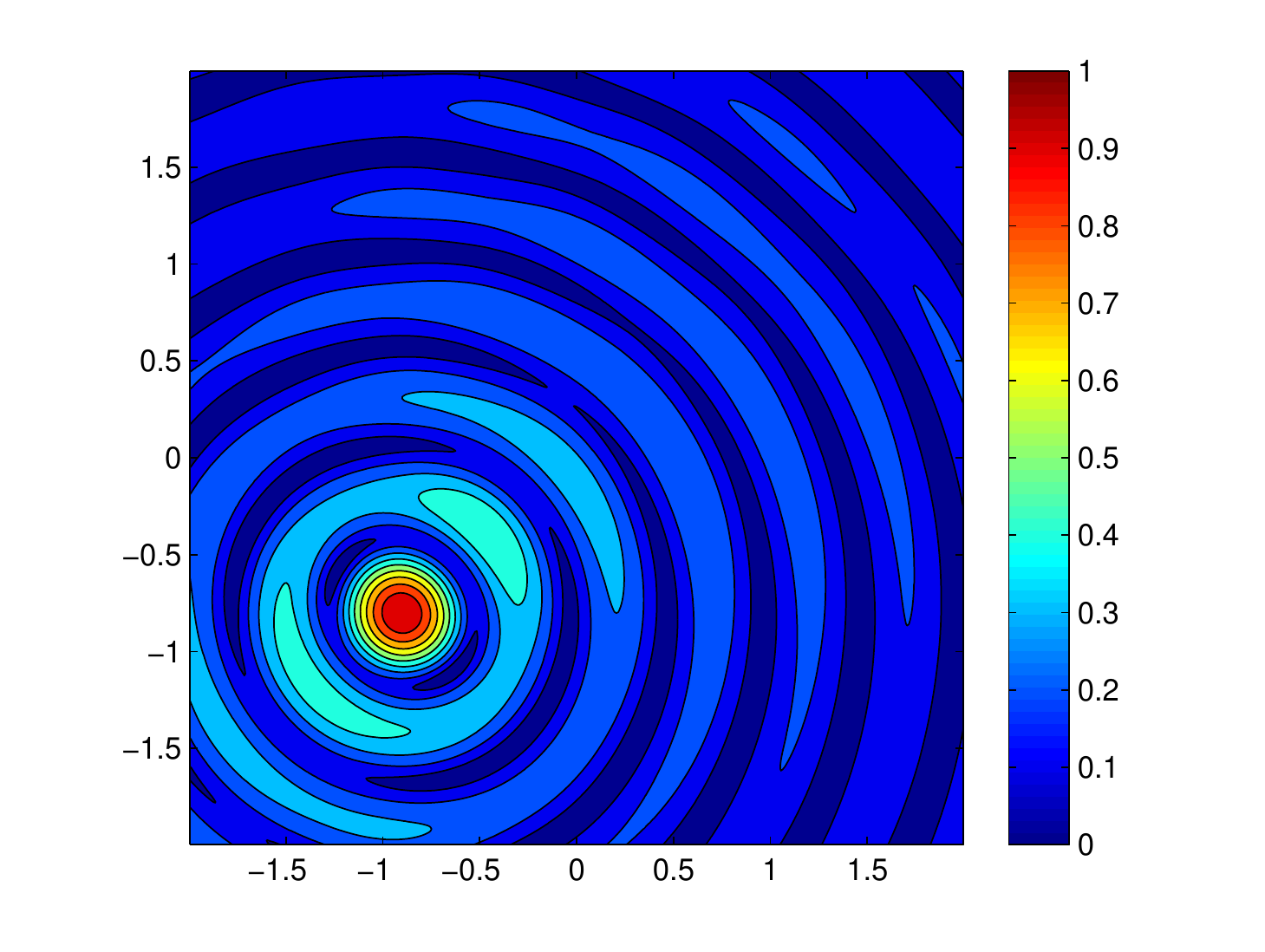} & \includegraphics[width=0.32\textwidth]{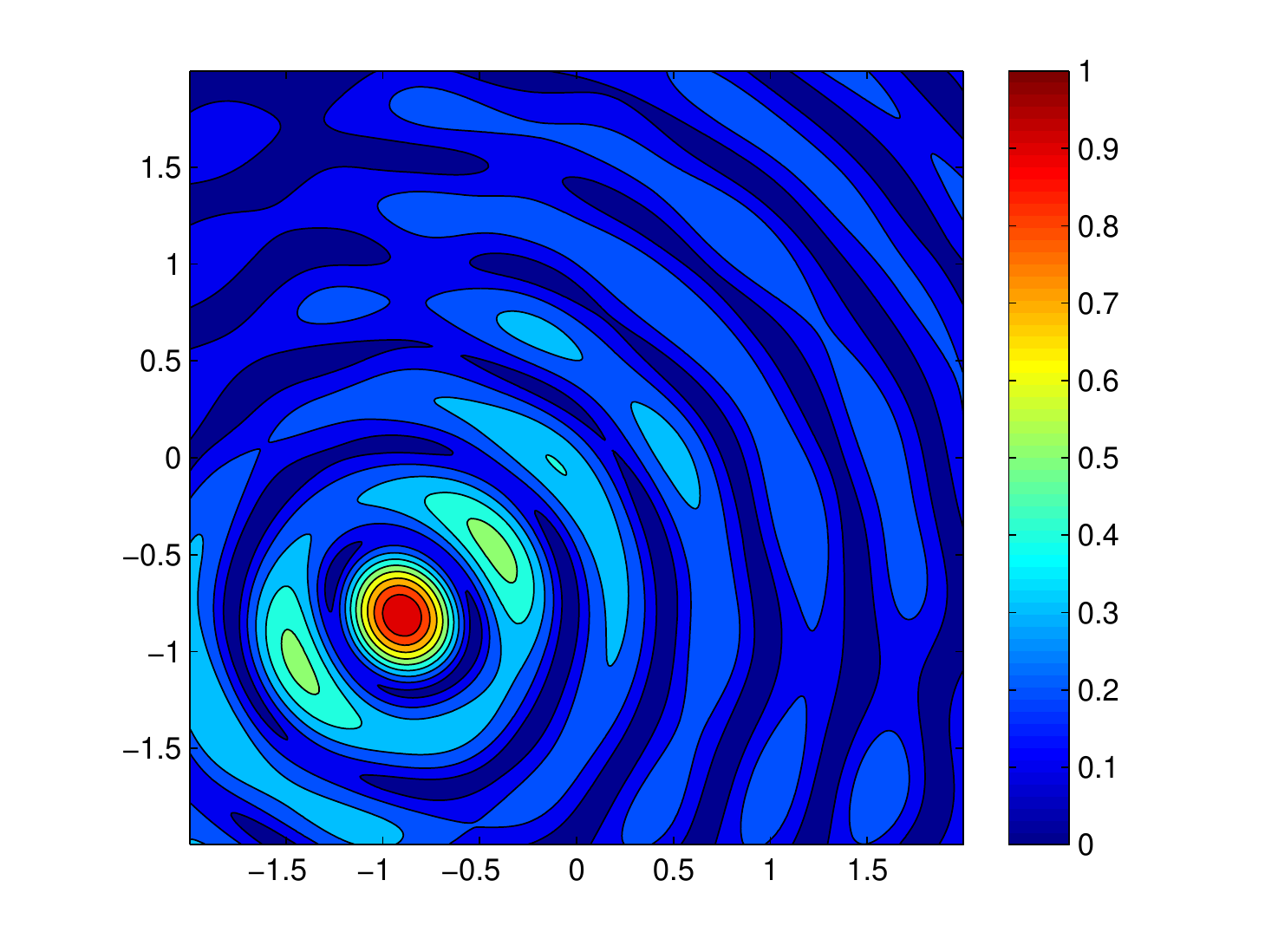}\tabularnewline
(d) & (e)\tabularnewline
\end{tabular}\hfill{}

\caption{\label{fig:ex10:1} Example 5 with $\eta=1$: (a) true scatterer;
reconstruction results using (b) exact near-field data, (c) noisy
near-field data with $\epsilon=20\%$, (d) exact far-field data, (e)
noisy far-field data with $\epsilon=20\%$.}
\end{figure}

However, when we increase the refractive index $n^{2}$ to be
$10+10\mathrm{i}$, the multiple scattering interaction of the scattered wave field due
to the obstacle and medium scatterers becomes more evident. Only in
such circumstances the location of the medium scatterer emerges
gradually as a red patch in the correct northeast location. In other
words, the DSM performs well for mixed types of scatterers under the
assumption that the interaction of scattered wave due to each
individual scatterer (be it an obstacle or a medium scatterer) must be strong enough. Such good performance is
quite robust with respect to large noise level.

\begin{figure}
\hfill{}%
\begin{tabular}{cc}
\includegraphics[width=0.32\textwidth]{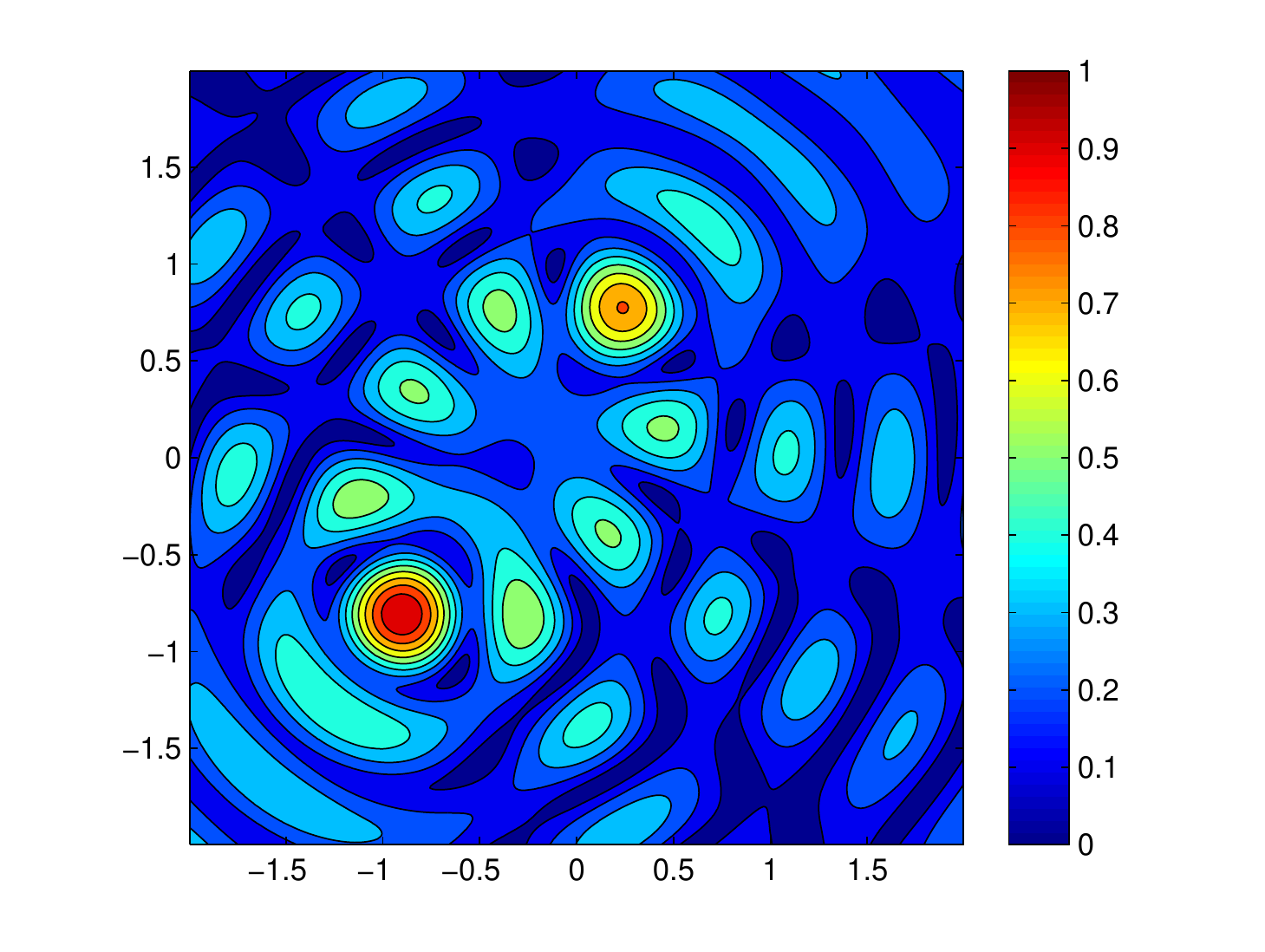} & \includegraphics[width=0.32\textwidth]{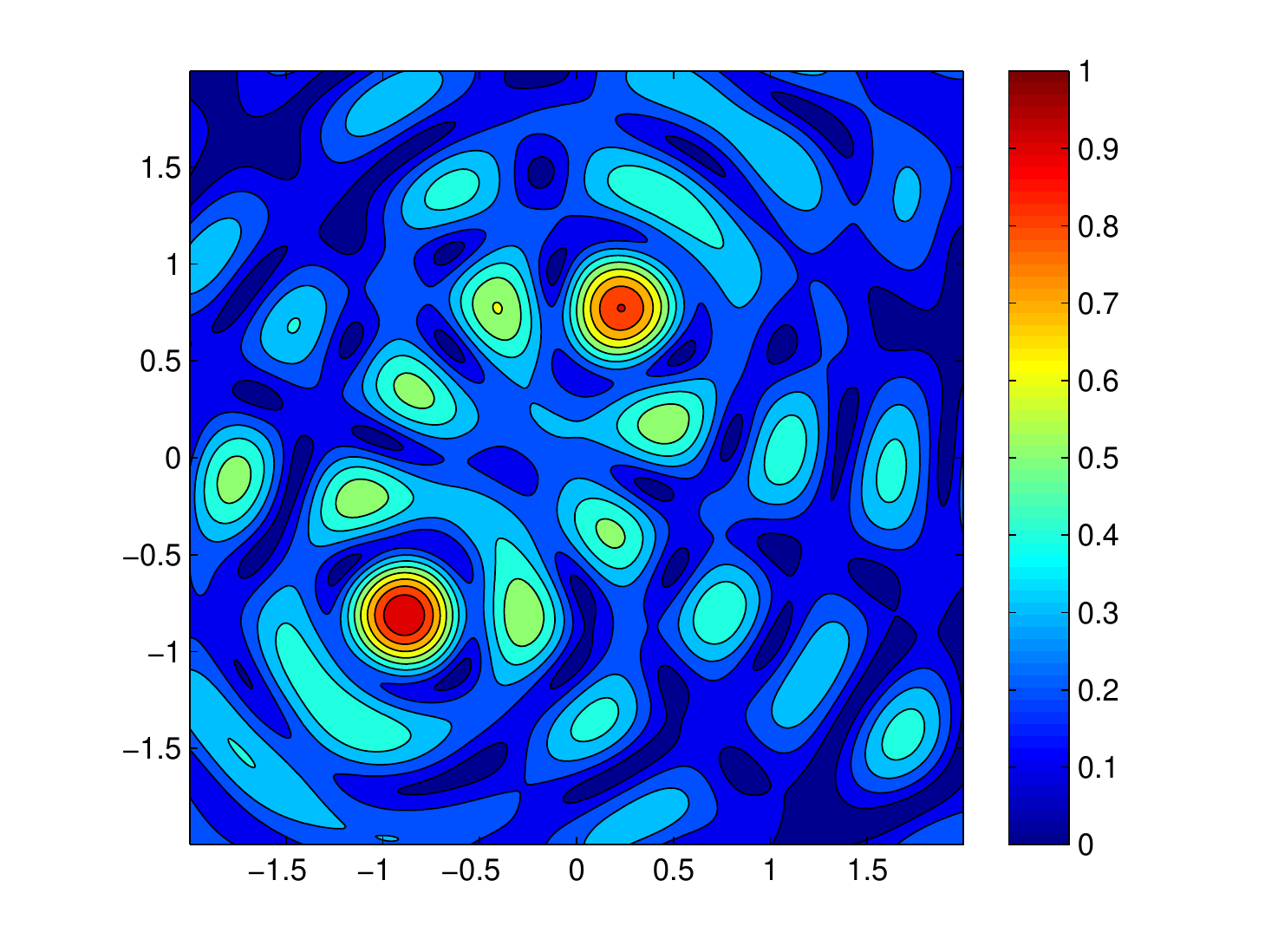}\tabularnewline
(a) & (b)\tabularnewline
\includegraphics[width=0.32\textwidth]{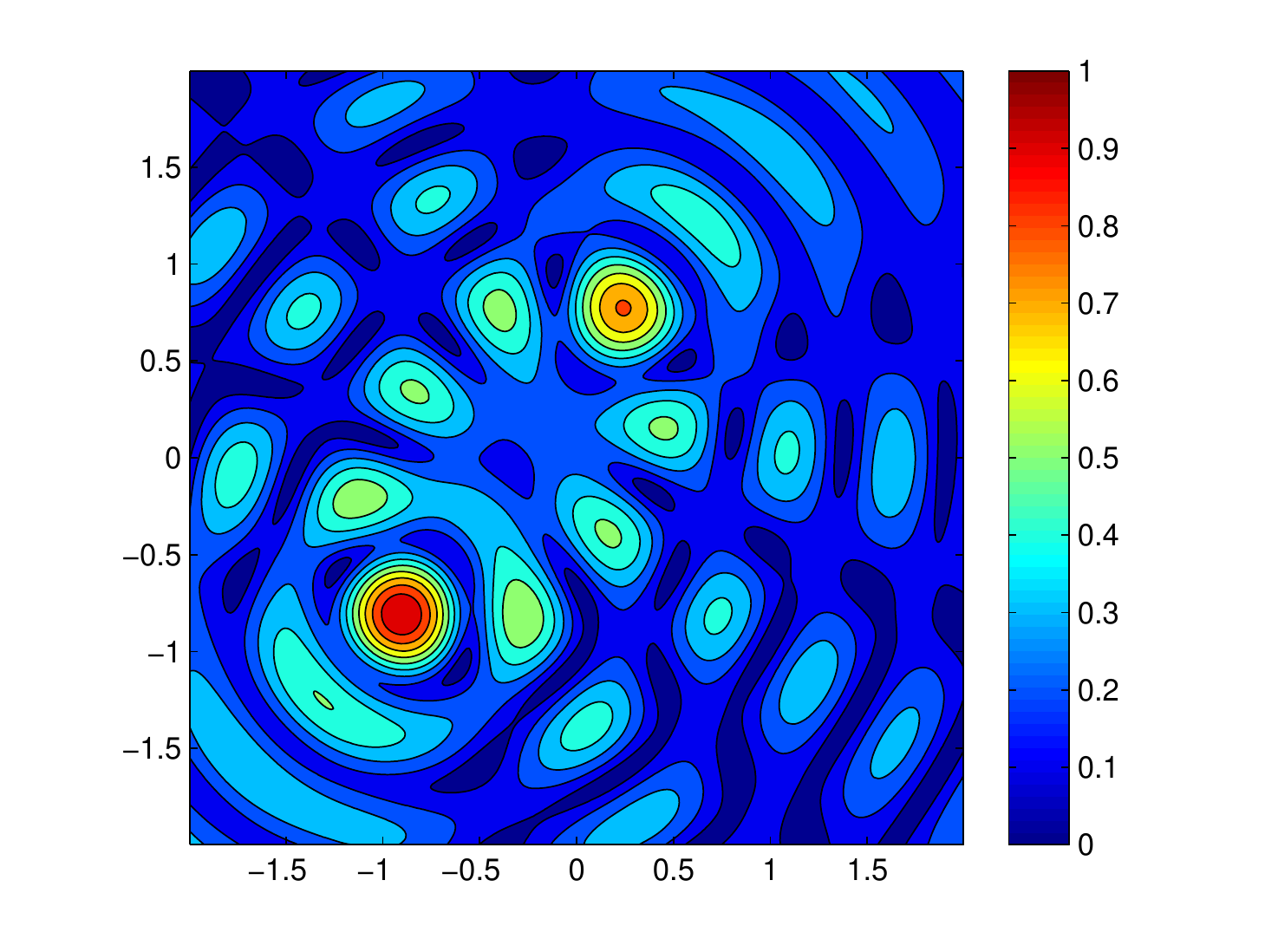} & \includegraphics[width=0.32\textwidth]{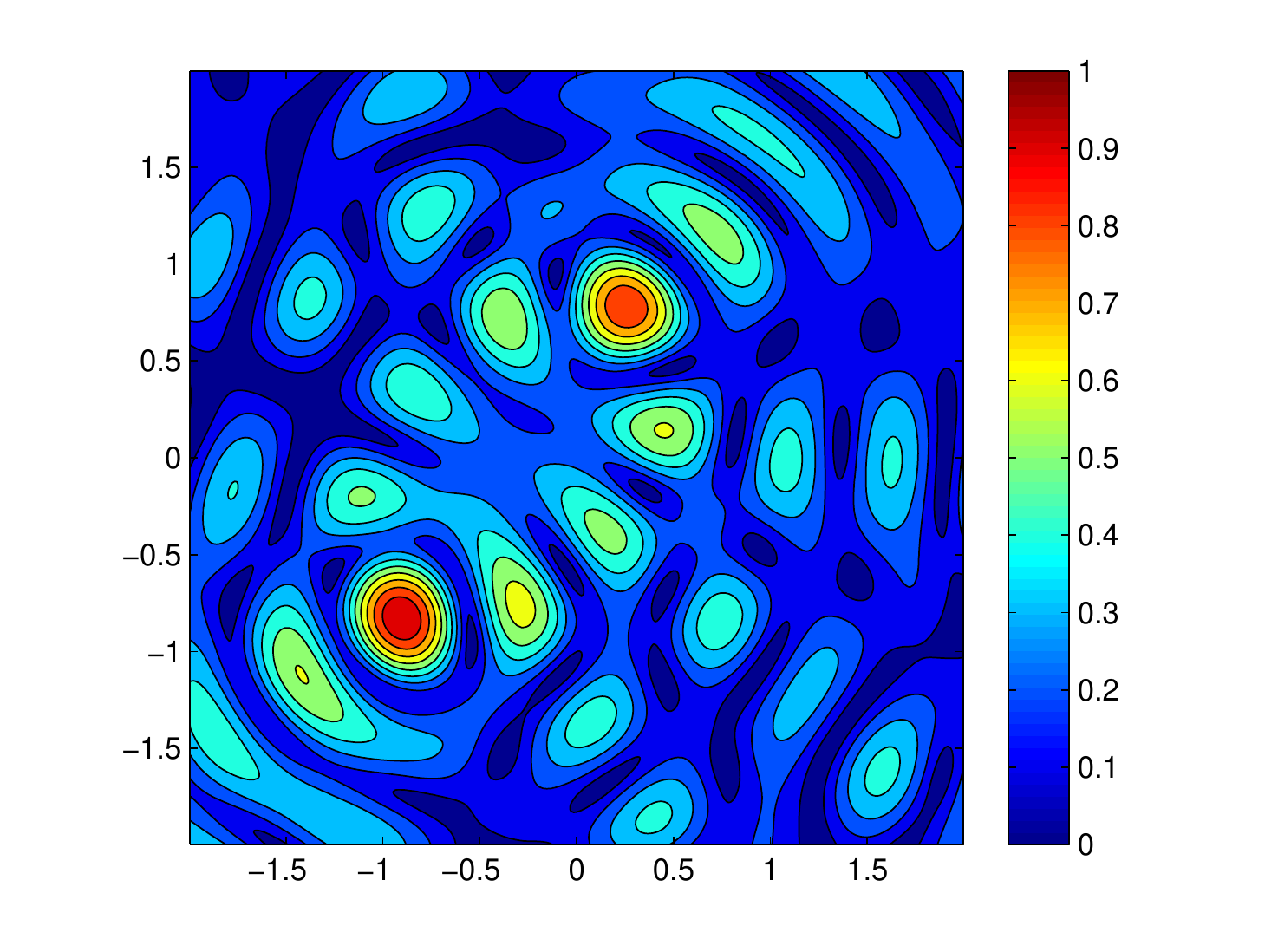}\tabularnewline
(c) & (d)\tabularnewline
\end{tabular}\hfill{}

\caption{\label{fig:ex10:1000} Example 5 with $n^{2}=10+10\mathrm{i}$, :
reconstruction results using (a) exact near-field data, (b) noisy
near-field data with $\epsilon=20\%$, (c) exact far-field data, (d)
noisy far-field data with $\epsilon=20\%$.}
\end{figure}

\smallskip{}

\subsection*{DSM for cracks }

Now we are going to test two examples with cracks and demonstrate
an interesting and promising application of the DSM to such an important
scattering scenario.

\smallskip{}

\textbf{Example 6 (A horizontal crack)}. This example considers a
thin crack of thickness $0.1\lambda$, centered at the origin and
parallel to the $x$-axis lying in the middle of the domain with length
1. The medium scatterer with coefficient $\eta=1$ is illuminated
by a horizontal incident plane wave, see Figures\ \ref{fig:ex6-1}(a)
for the configuration.

Cracks are the most illusive scatterer type to be identified as it
has a very small thickness, and it is non-trivial even with the data
from multiple incident directions. The results using the near-field
and far-field data from just one incident field are shown in Figure\ \ref{fig:ex6-1}.

In the noise-free cases, both DSM(n) and DSM(f) can successfully
determine the location and length for the crack, and the identified
geometry is roughly a red long bar with correct length. It is important
to observe that the DSM also works for this difficult example with
one incident, even with large noise, say $\epsilon=20\%$; see Figure\,\ref{fig:ex6-1}
(c), (e). But in this case, the DSM(f) performs better than the DSM(n),
which seems to break the crack into two pieces.

\begin{figure}
\hfill{}\hspace{-0.03\textwidth}%
\begin{tabular}{c}
\includegraphics[clip,width=0.32\textwidth]{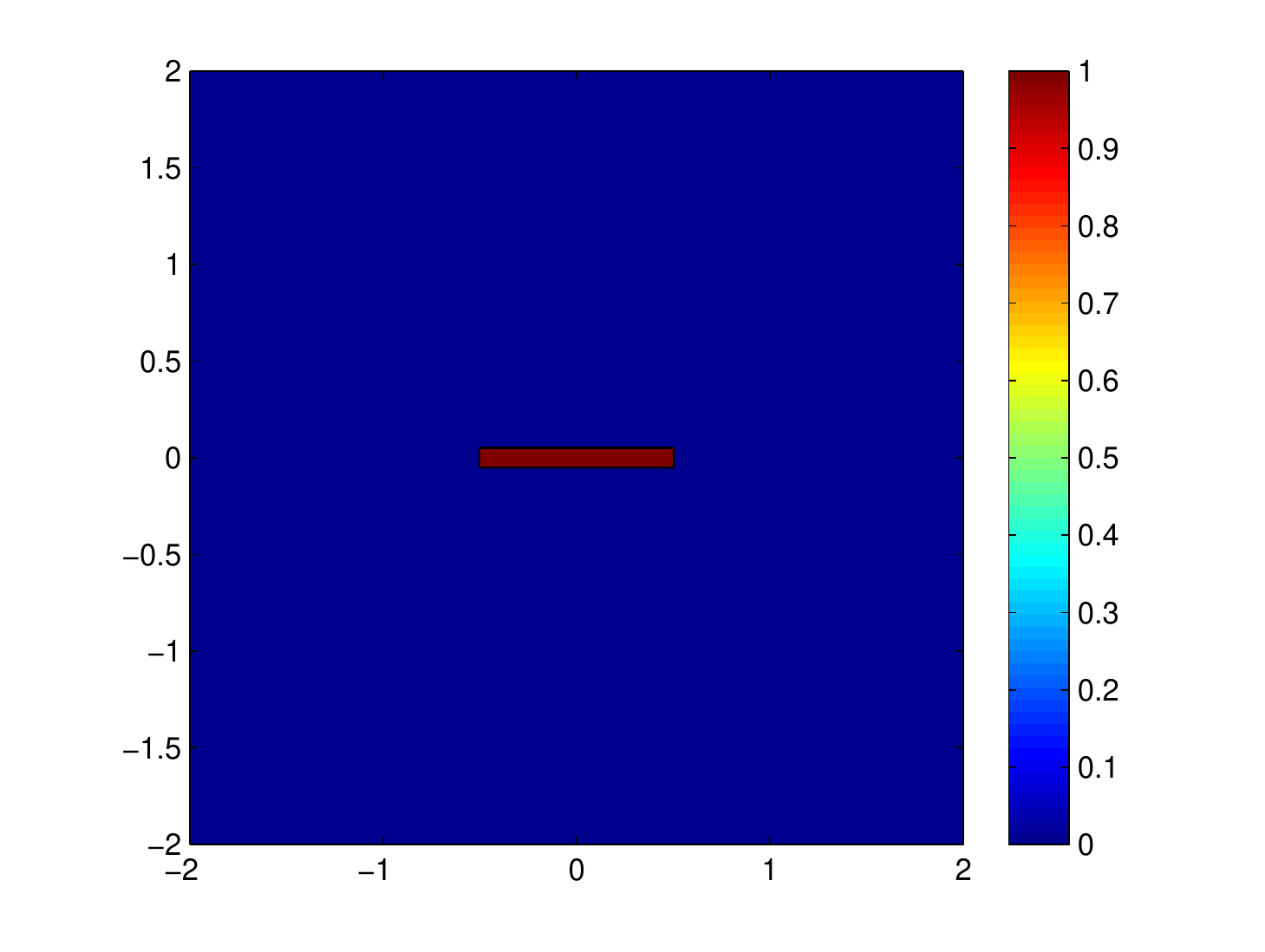}\tabularnewline
(a)\tabularnewline
\end{tabular}\negthinspace{}%
\begin{tabular}{cc}
\includegraphics[width=0.32\textwidth]{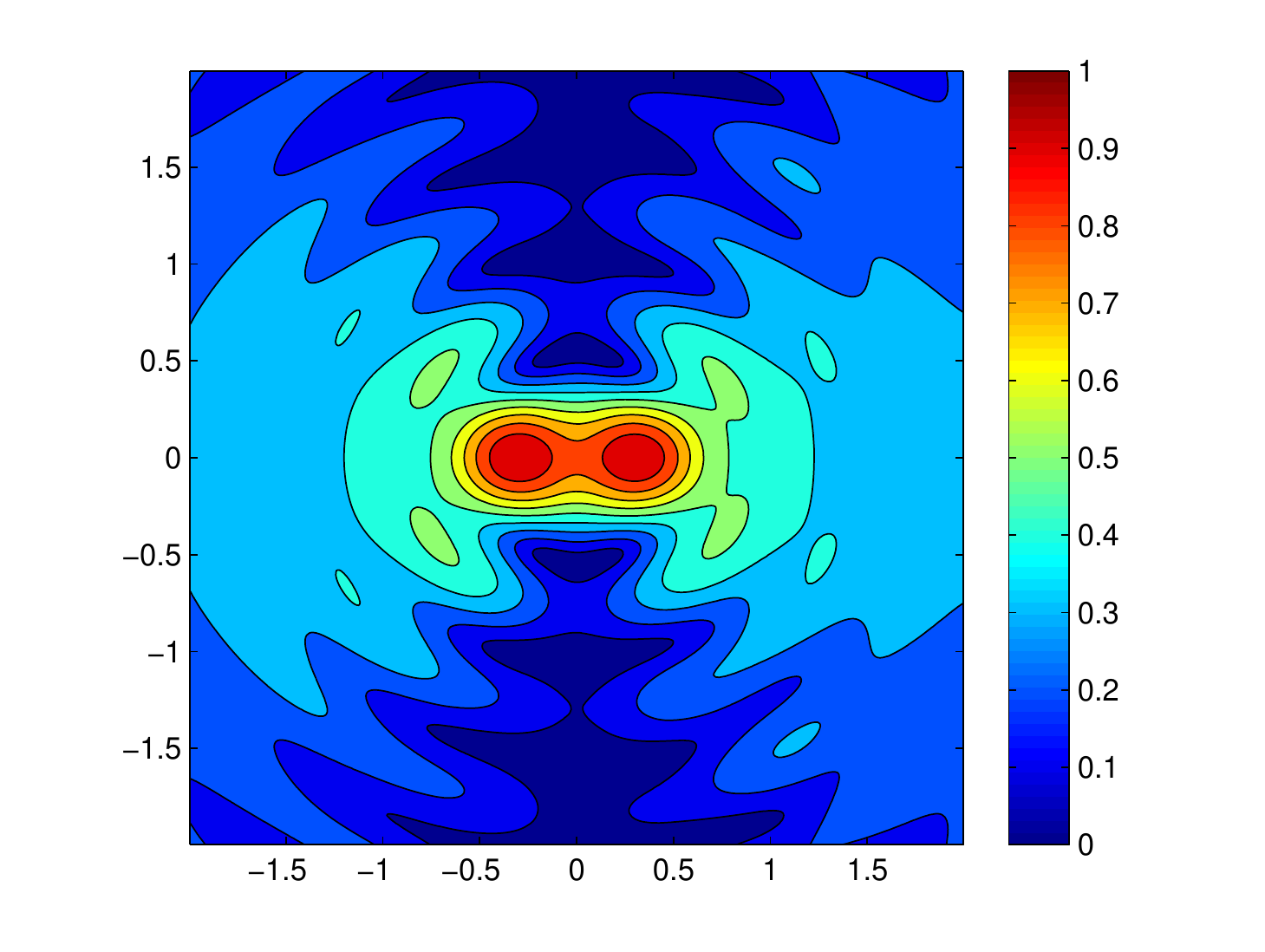} & \includegraphics[width=0.32\textwidth]{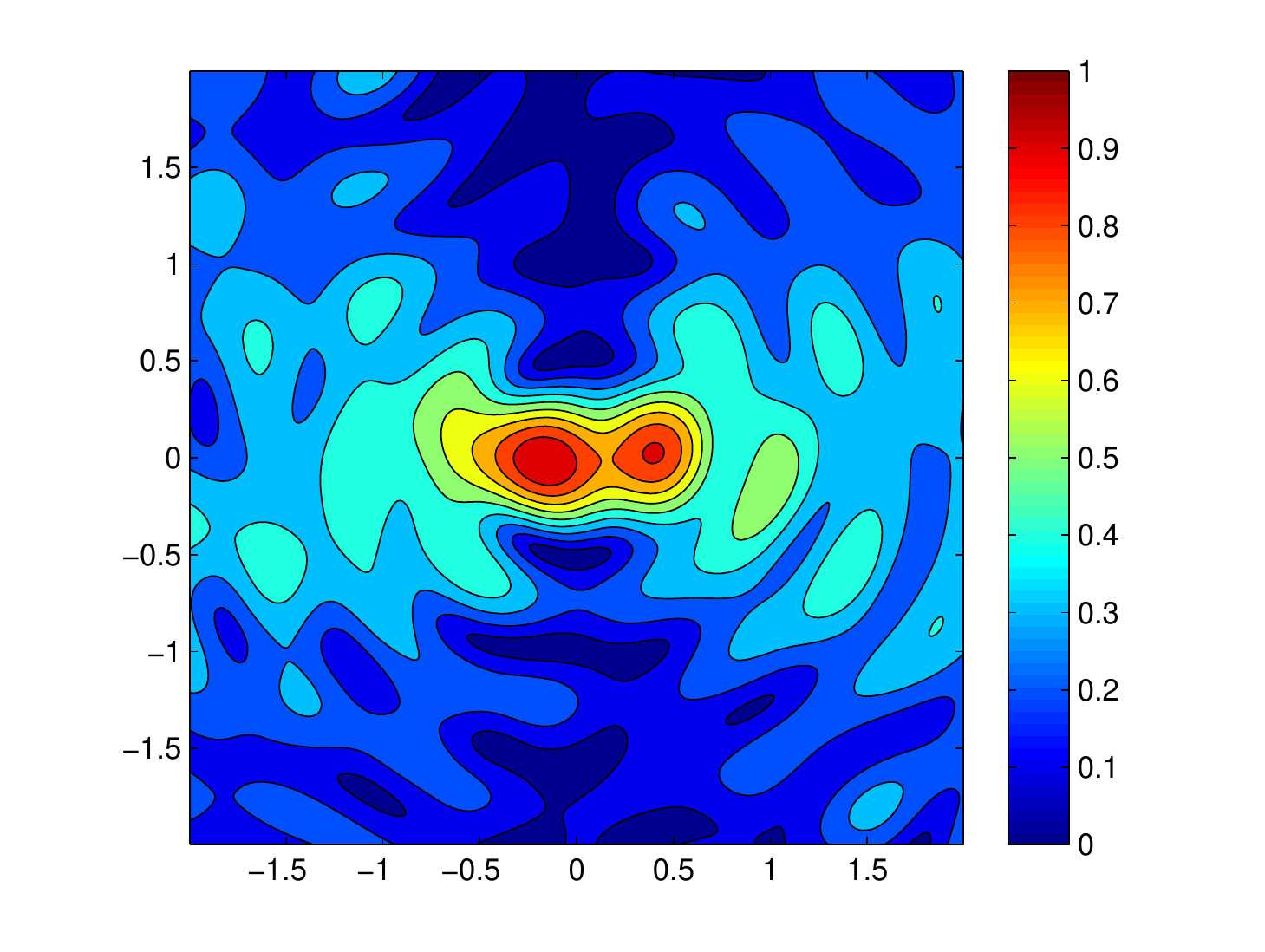}\tabularnewline
(b) & (c)\tabularnewline
\includegraphics[width=0.32\textwidth]{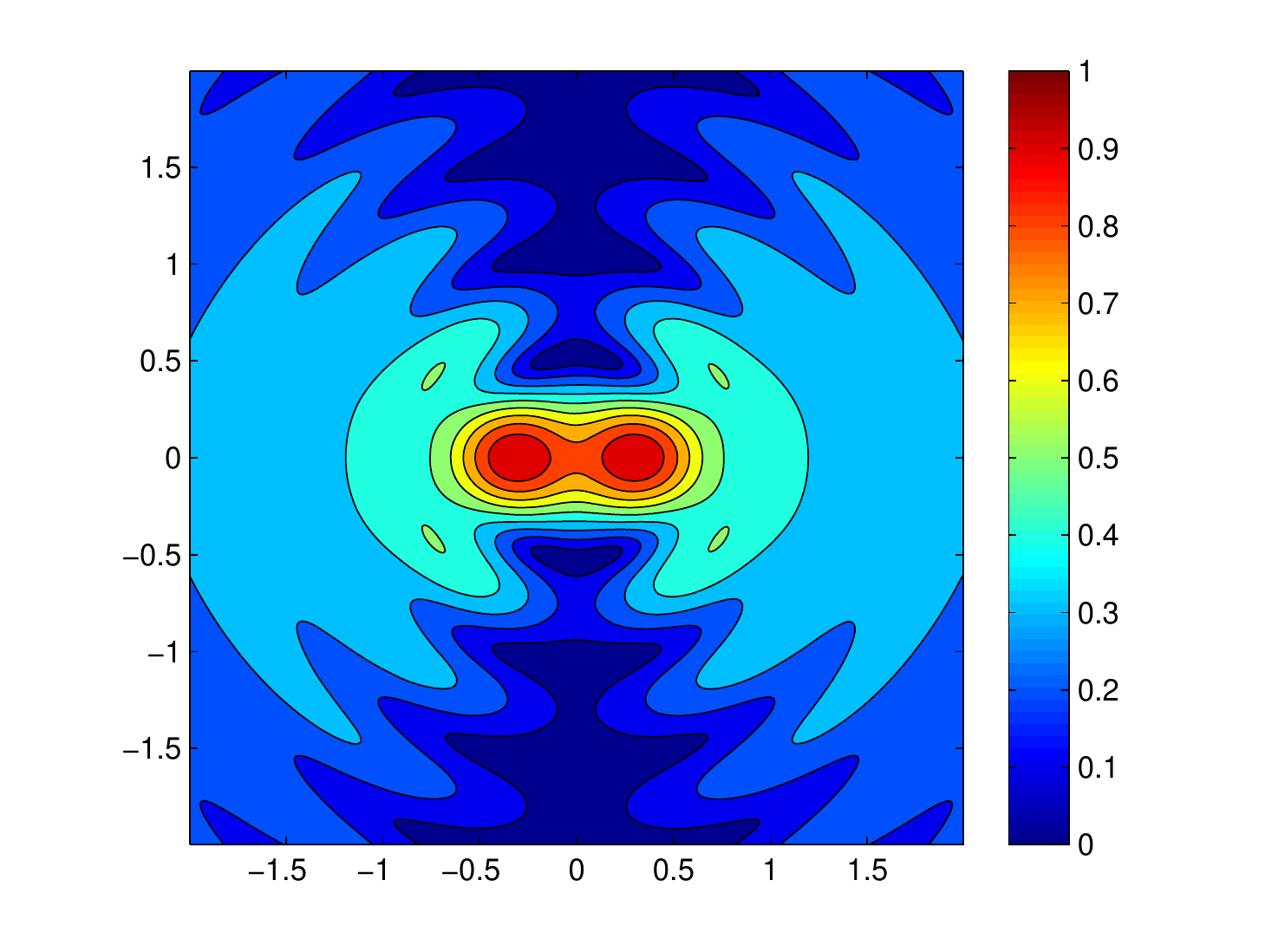} & \includegraphics[width=0.32\textwidth]{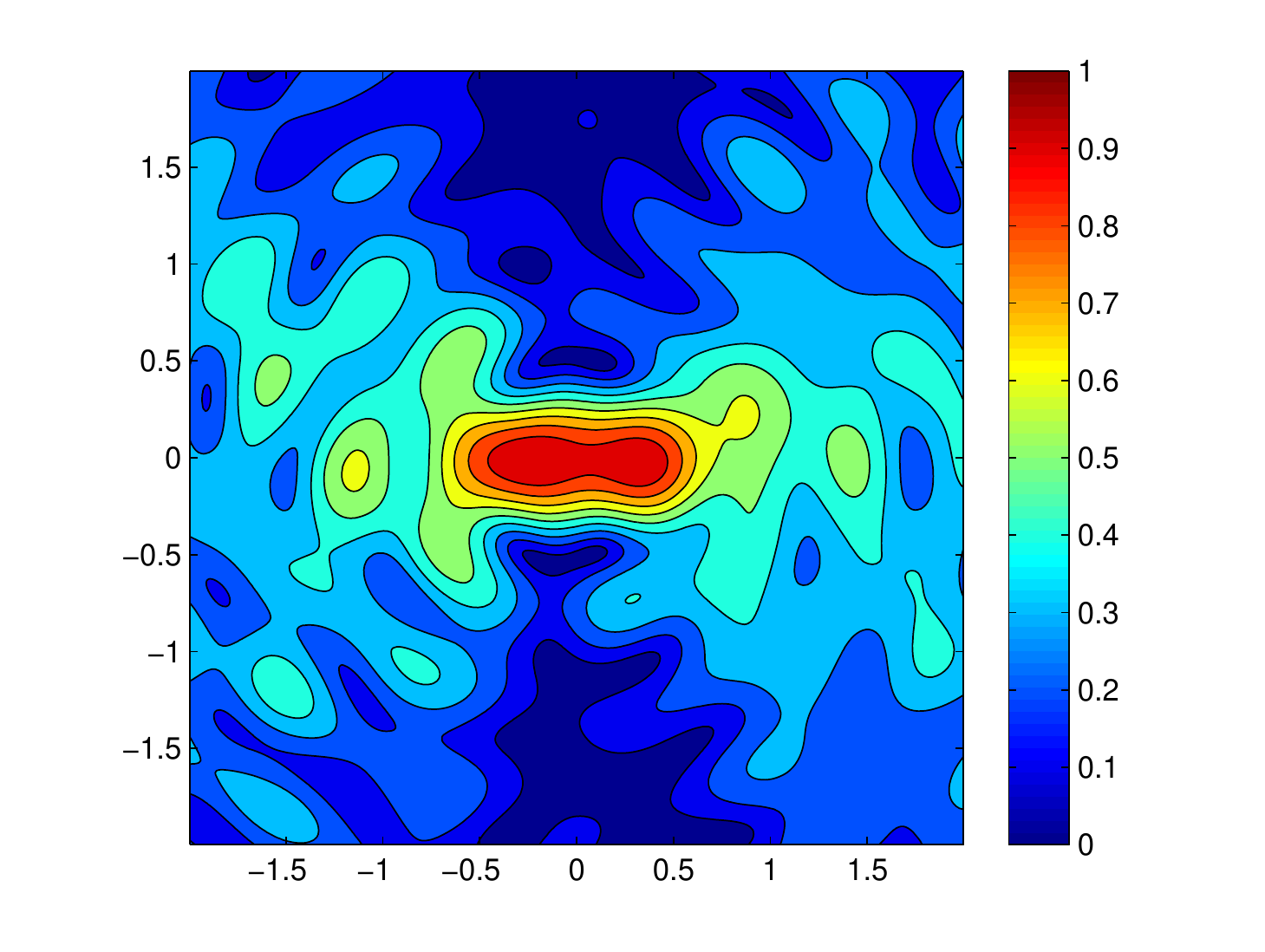}\tabularnewline
(d) & (e)\tabularnewline
\end{tabular}\hfill{}

\caption{\label{fig:ex6-1} Example 6: (a) true crack scatterer; reconstruction
results using (b) exact near-field data, (c) noisy near-field data
with $\epsilon=20\%$, (d) exact far-field data, (e) noisy far-field
data with $\epsilon=20\%$.}
\end{figure}

\smallskip{}

\textbf{Example 7 (An L-shape crack)}. The example considers an L-shape
crack with thickness $0.1\lambda$ and length $2\lambda$. The coefficient
$\eta$ in the inhomogeneous region is 1.

This L-shaped crack is more challenging than the crack in Example
6. We use the incident wave along direction $d_{2}=(1,-1)^{T}/\sqrt{2}$
and the reconstruction is given in Figure\ \ref{fig:ex7-2}. One
can see that an L-shaped dark yellow bar is identified by choosing
the cut-off value about $0.7$ in the noise-free case for both DSM(n)
and DSM(f). But with $20\%$ noise, the DSM(f) provides a much better
profile than the DSM(n) in terms of the length and connectedness of
the crack.

Finally we test with two incident wave directions $d_{1}$ and $d_{2}$,
and report the results in Figure\ \ref{fig:ex7-3}. One can see that
both DSM(n) and DSM(f) can yield significantly enhanced images than
those obtained using a single wave in direction $d_{1}$ or $d_{2}$. Moreover,
the recovered L-shaped geometry degenerates rather slowly as we increase
the noise level from $5\%$ to $20\%$, see Figure\ \ref{fig:ex7-3}.
This example demonstrates the robustness and effectiveness of the DSM as a potential
promising technique for crack detections.

\begin{figure}
\hfill{}\hspace{-0.03\textwidth}%
\begin{tabular}{c}
\includegraphics[clip,width=0.32\textwidth]{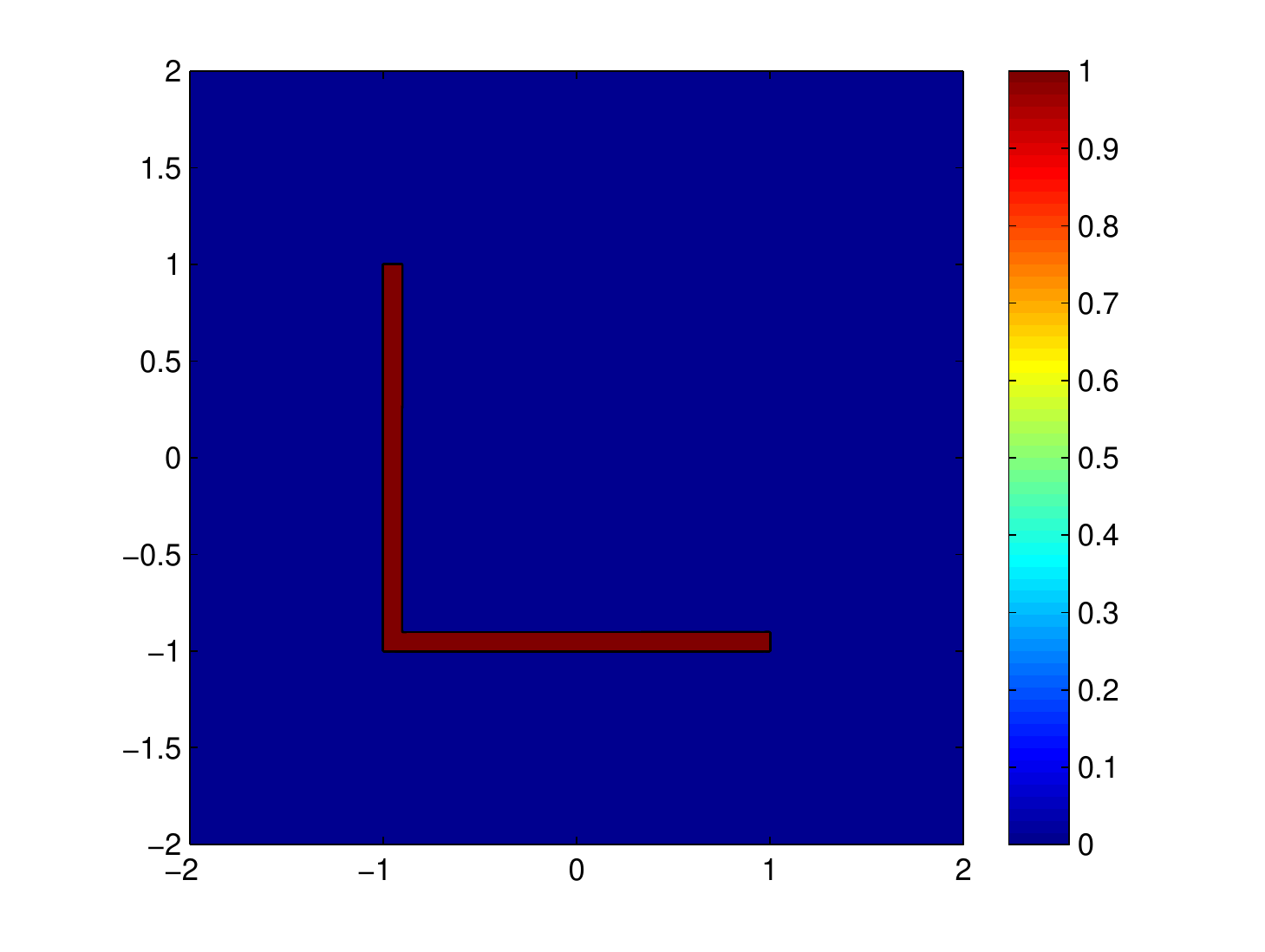}\tabularnewline
(a)\tabularnewline
\end{tabular}\negthinspace{}%
\begin{tabular}{cc}
\includegraphics[width=0.32\textwidth]{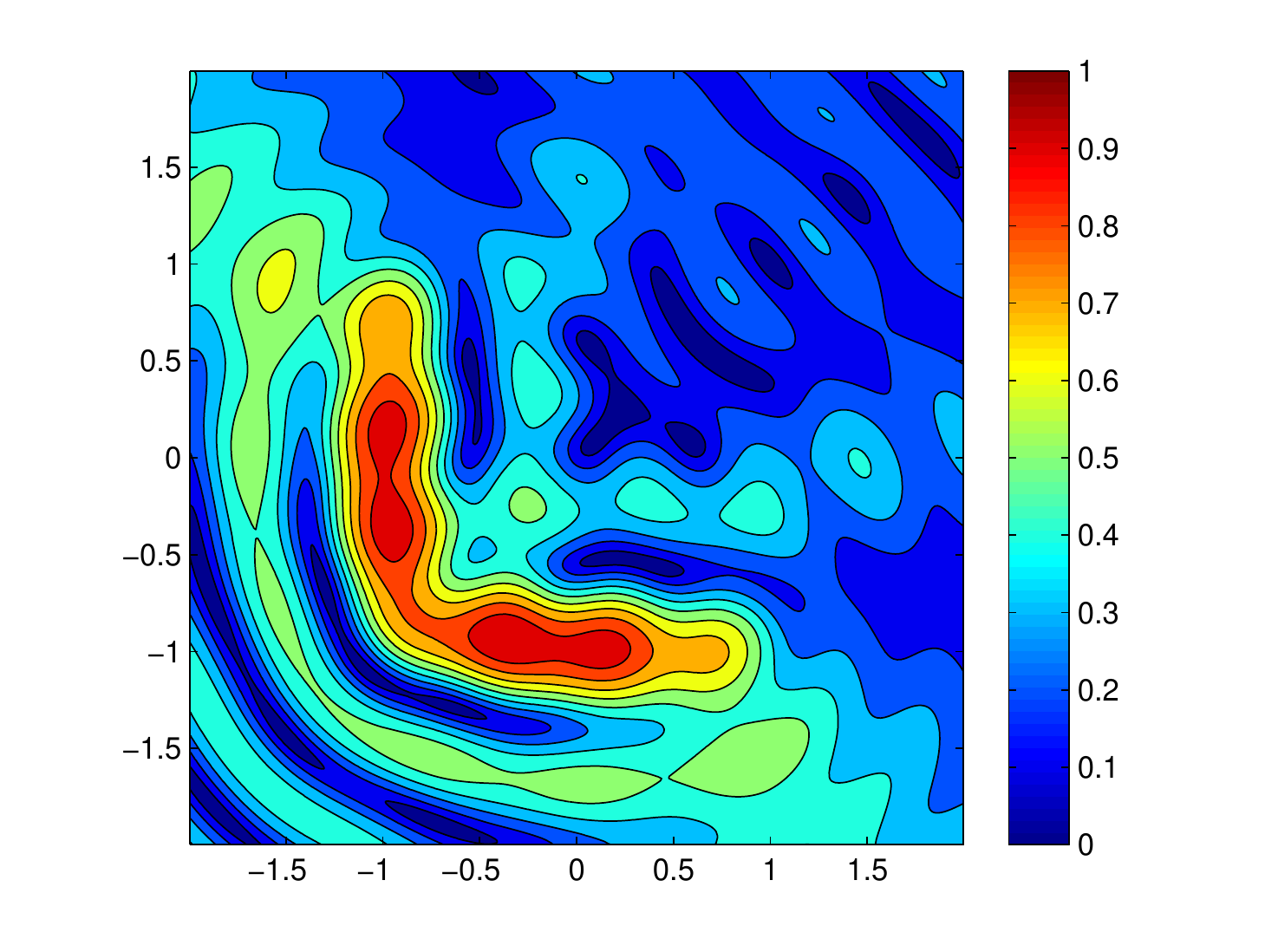} & \includegraphics[width=0.32\textwidth]{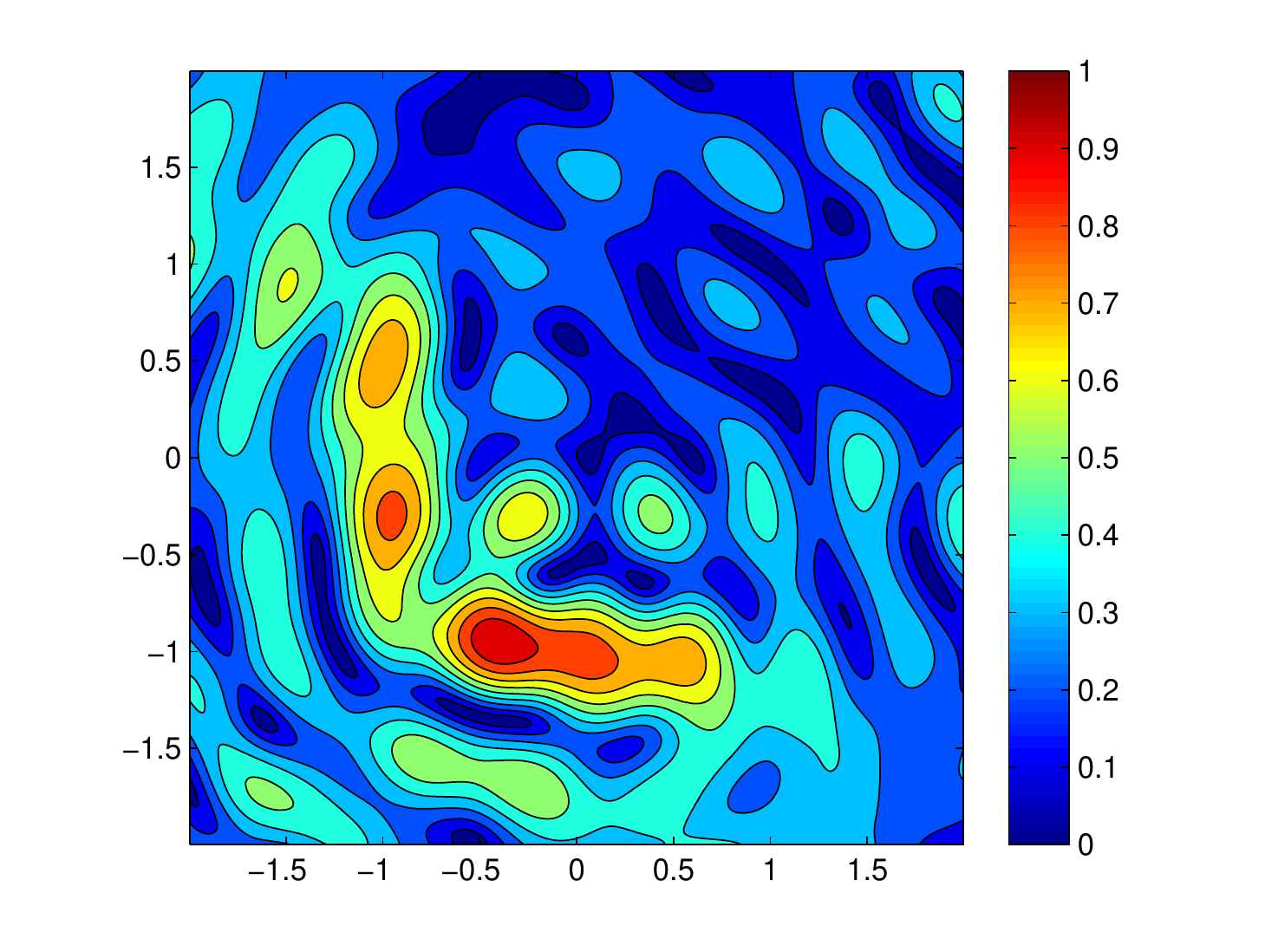}\tabularnewline
(b) & (c)\tabularnewline
\includegraphics[width=0.32\textwidth]{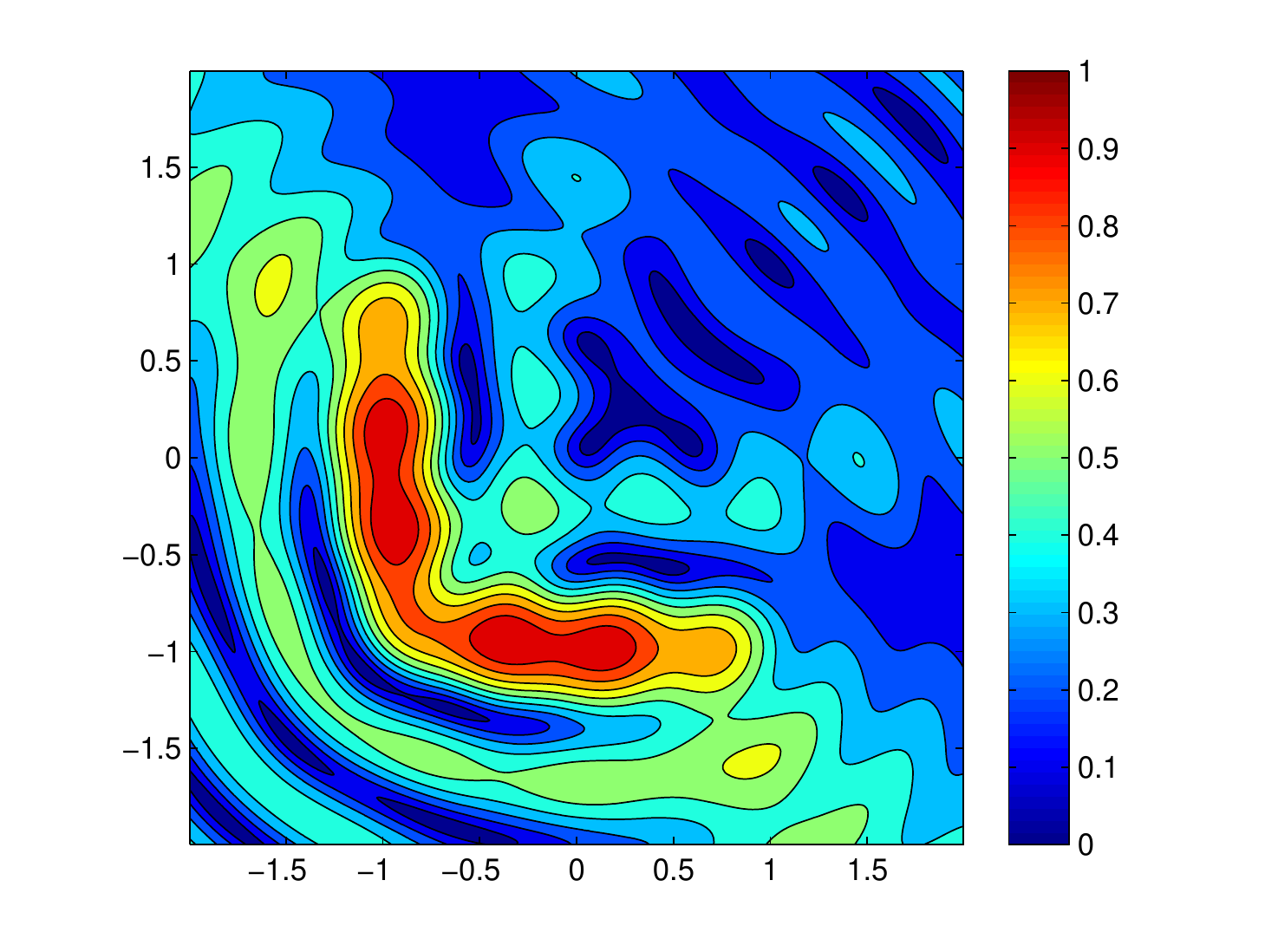} & \includegraphics[width=0.32\textwidth]{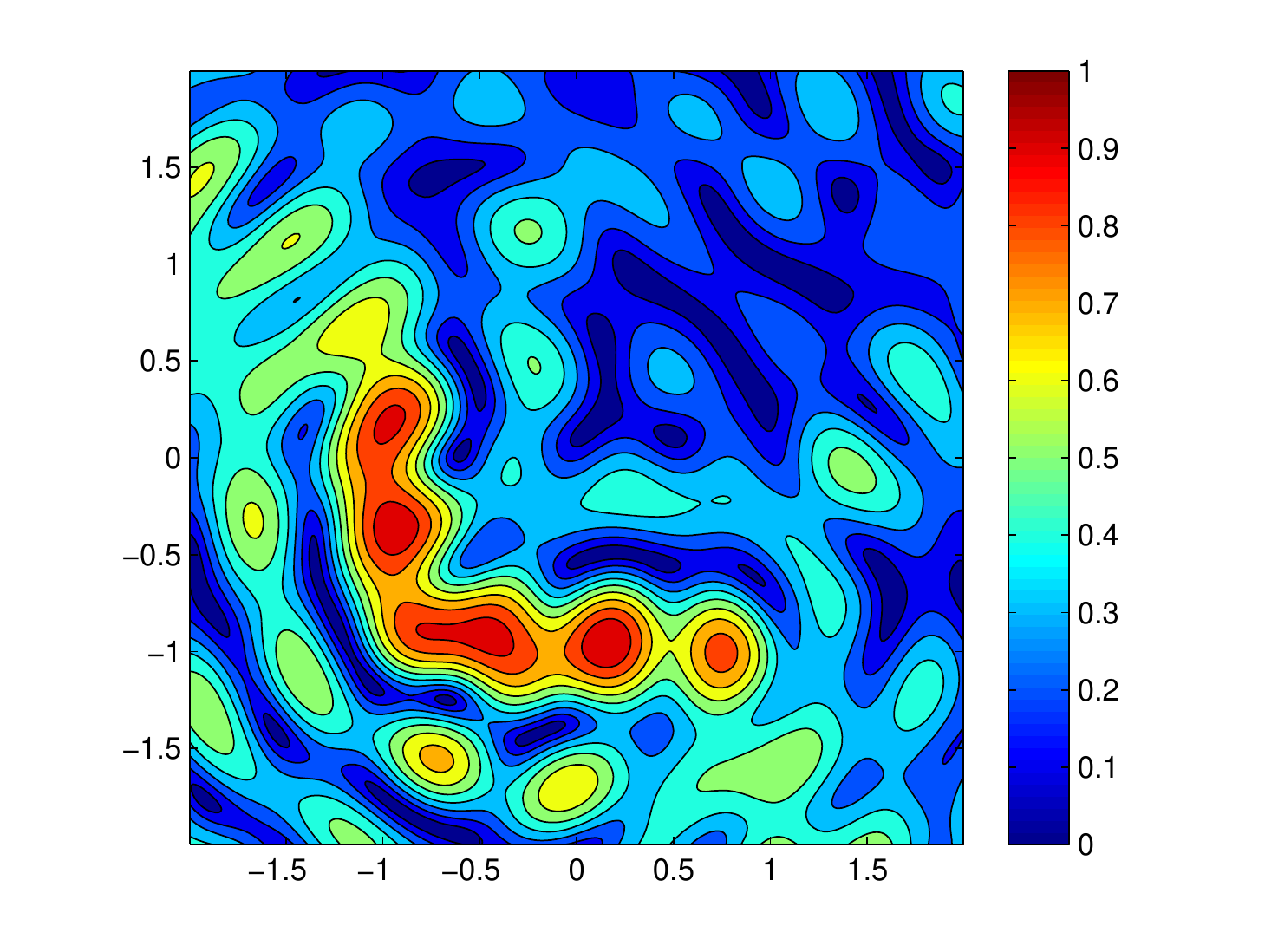}\tabularnewline
(d) & (e)\tabularnewline
\end{tabular}\hfill{}

\caption{\label{fig:ex7-2} Example 7 with incident direction $d_{2}=(1,\,-1)^{T}/\sqrt{2}$
: (a) true L-shape crack scatterer; reconstruction results using (b)
exact near-field data, (c) noisy near-field data with $\epsilon=20\%$,
(d) exact far-field data, (e) noisy far-field data with $\epsilon=20\%$.}
\end{figure}

\begin{figure}
\hfill{}\includegraphics[width=0.24\textwidth]{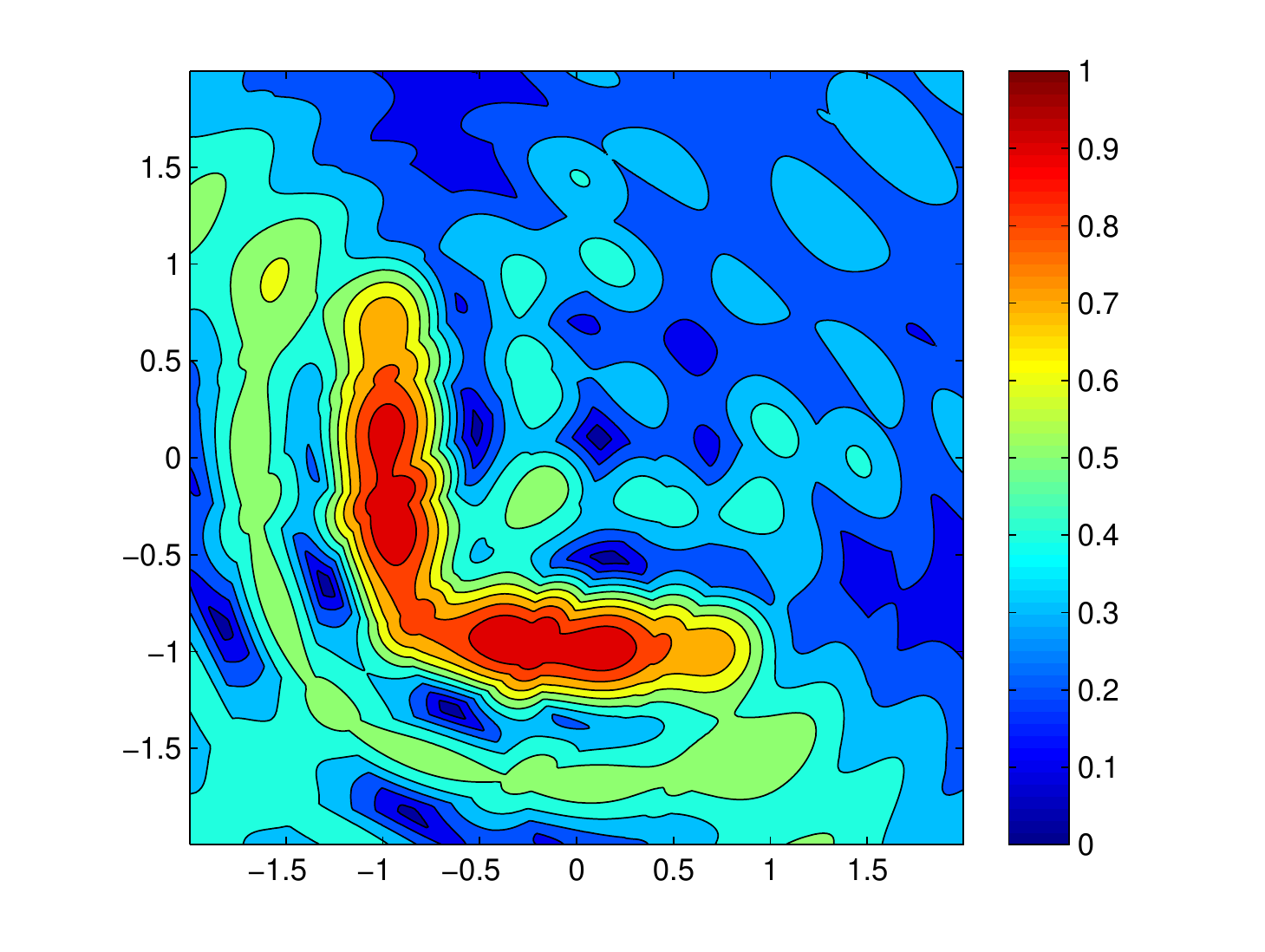}\hfill{}
\includegraphics[width=0.24\textwidth]{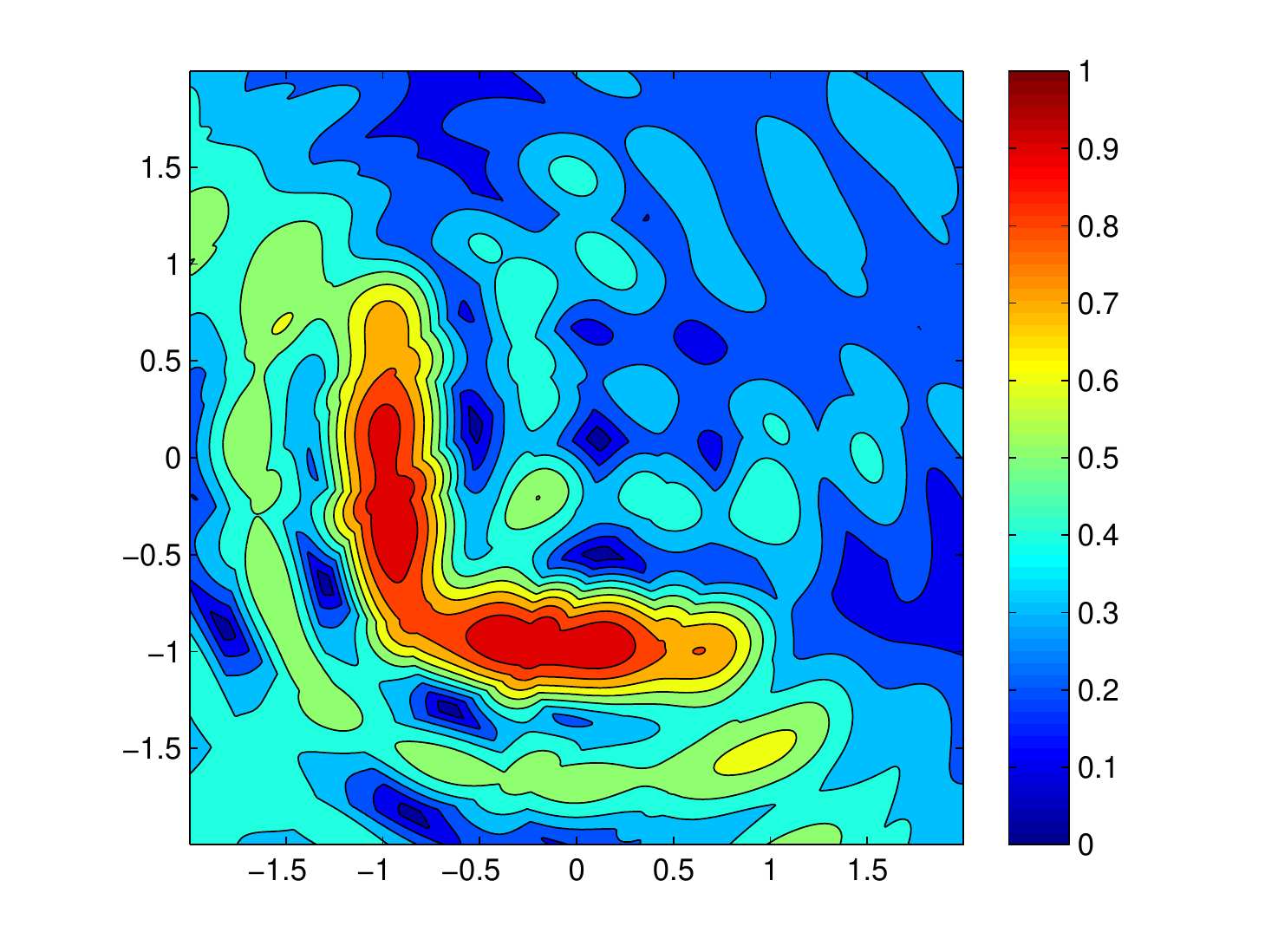}\hfill{}
\includegraphics[width=0.24\textwidth]{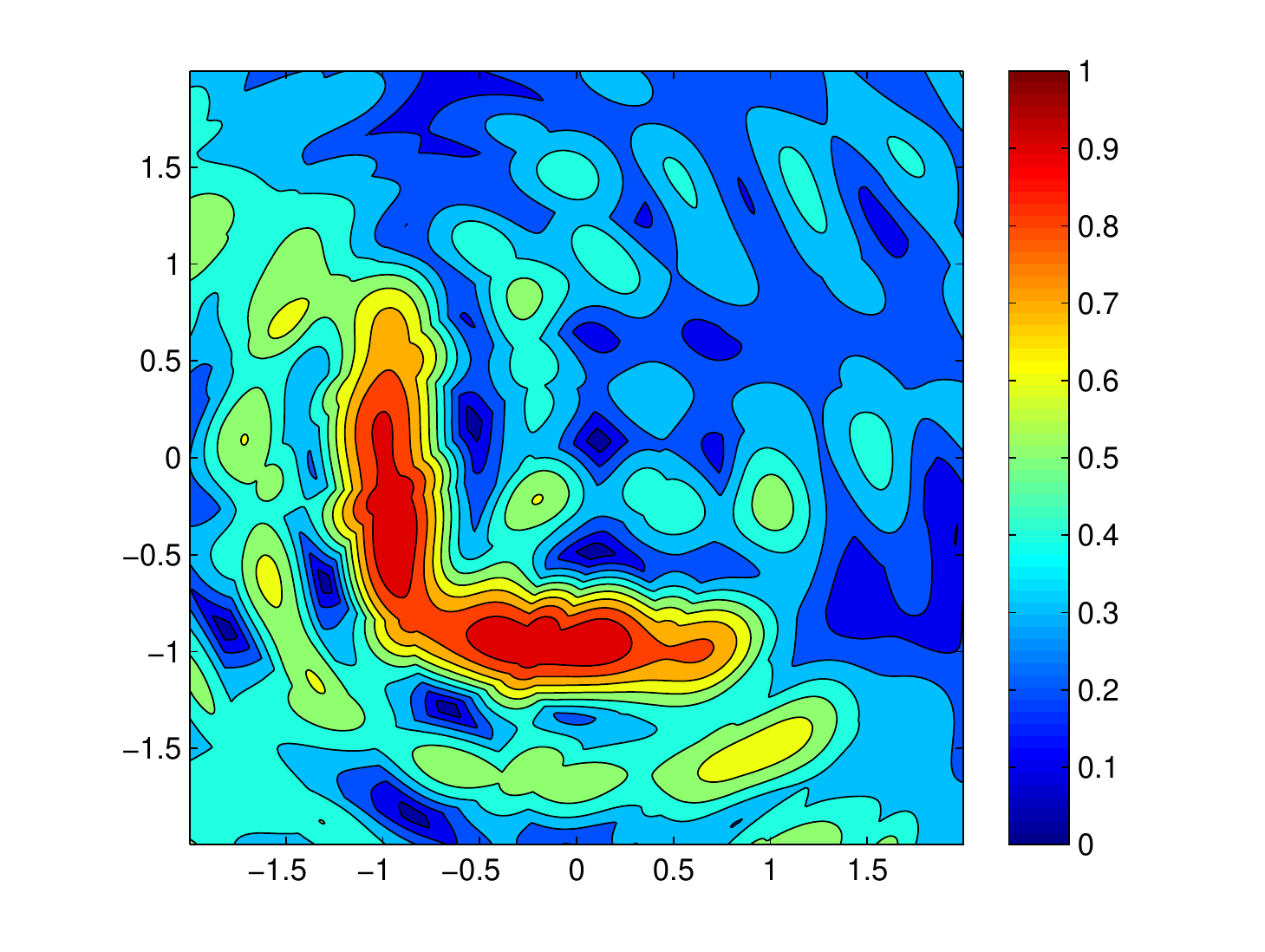}\hfill{}
\includegraphics[width=0.24\textwidth]{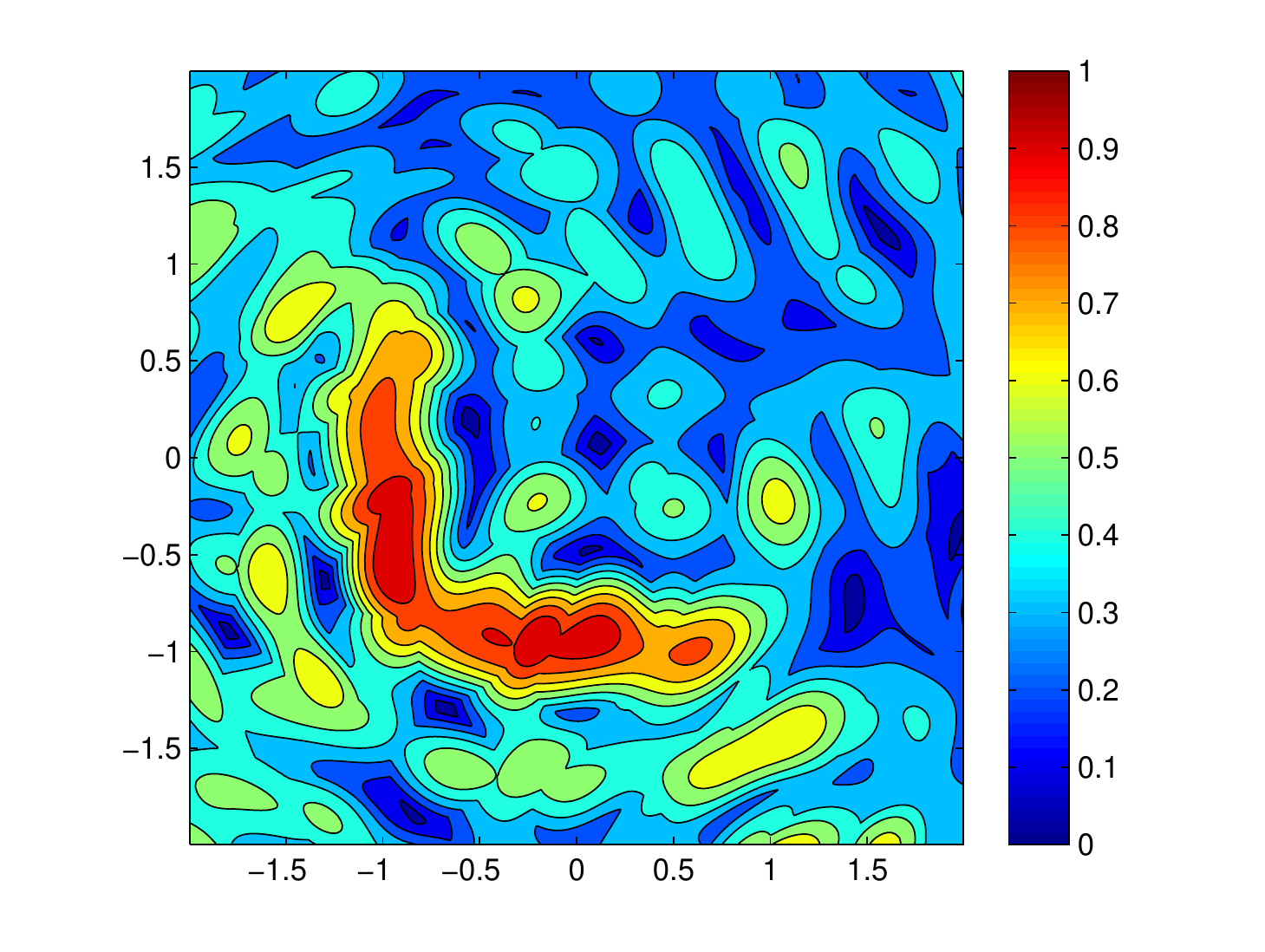}\hfill{}\hfill{}

\hfill{}DSM(n) using two incident waves: from left to right, noise
level is $0$, $5\%$, $10\%$, $20\%$.\hfill{}

\hfill{}\includegraphics[width=0.24\textwidth]{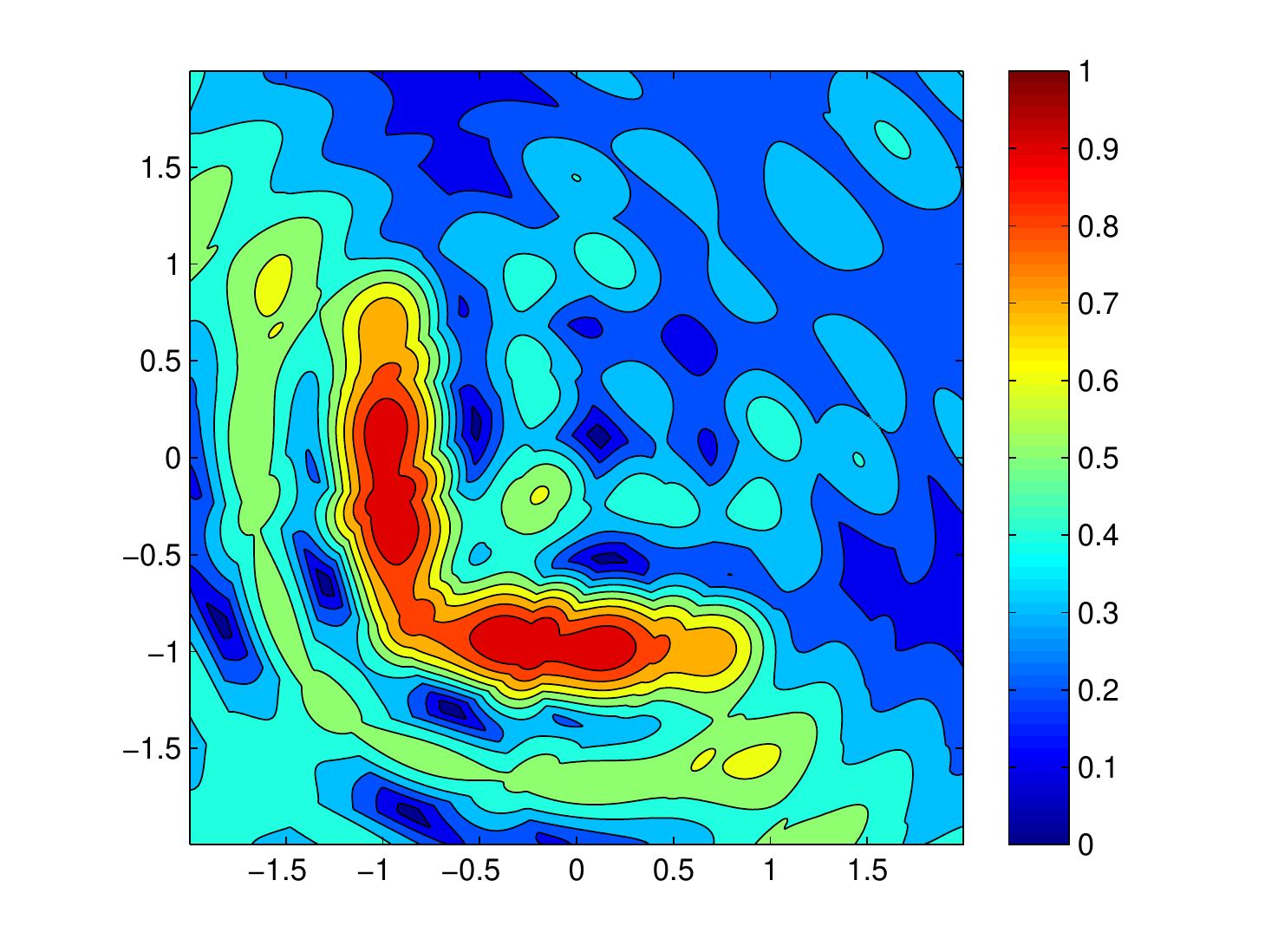}\hfill{}
\includegraphics[width=0.24\textwidth]{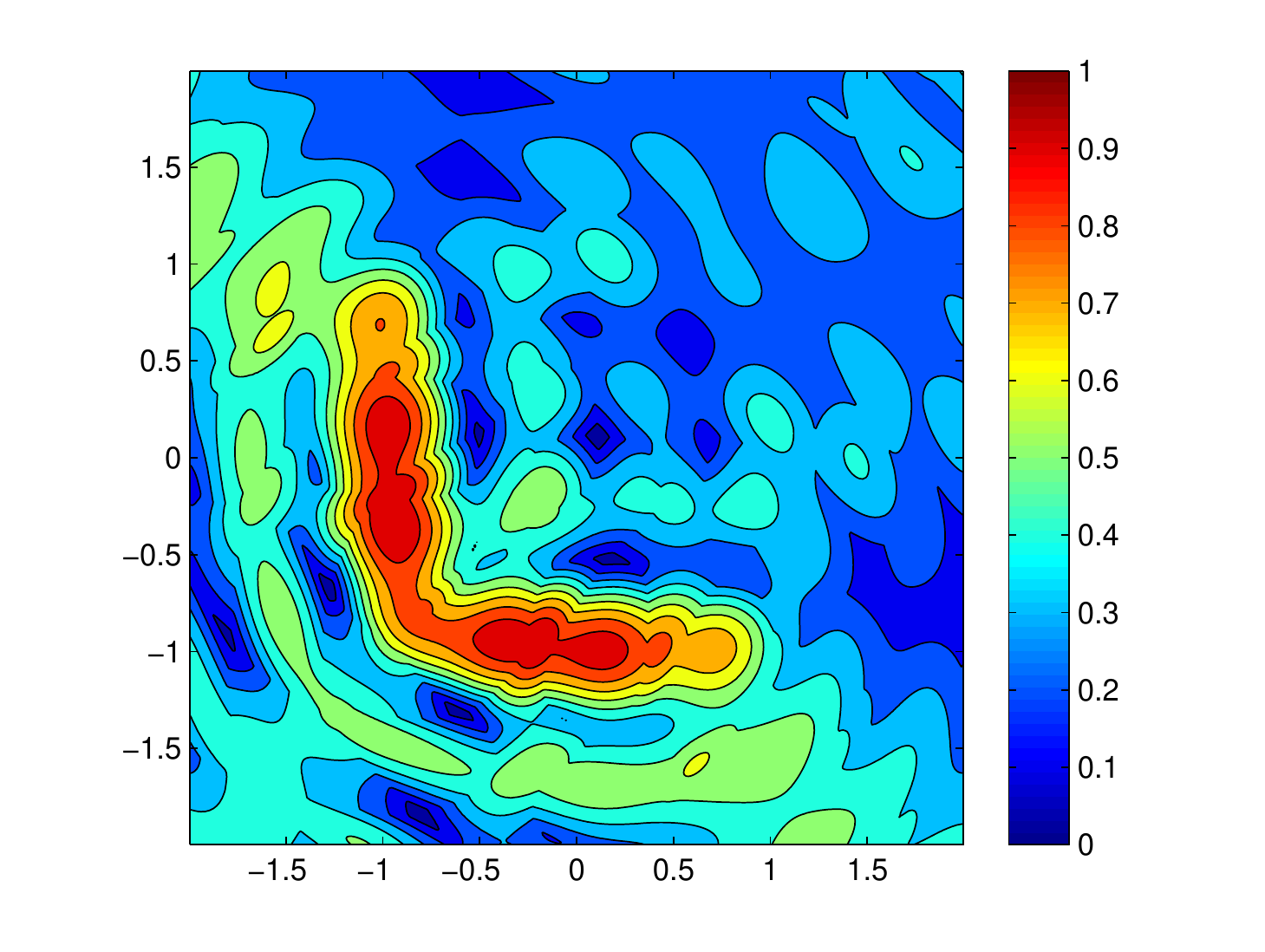}\hfill{}
\includegraphics[width=0.24\textwidth]{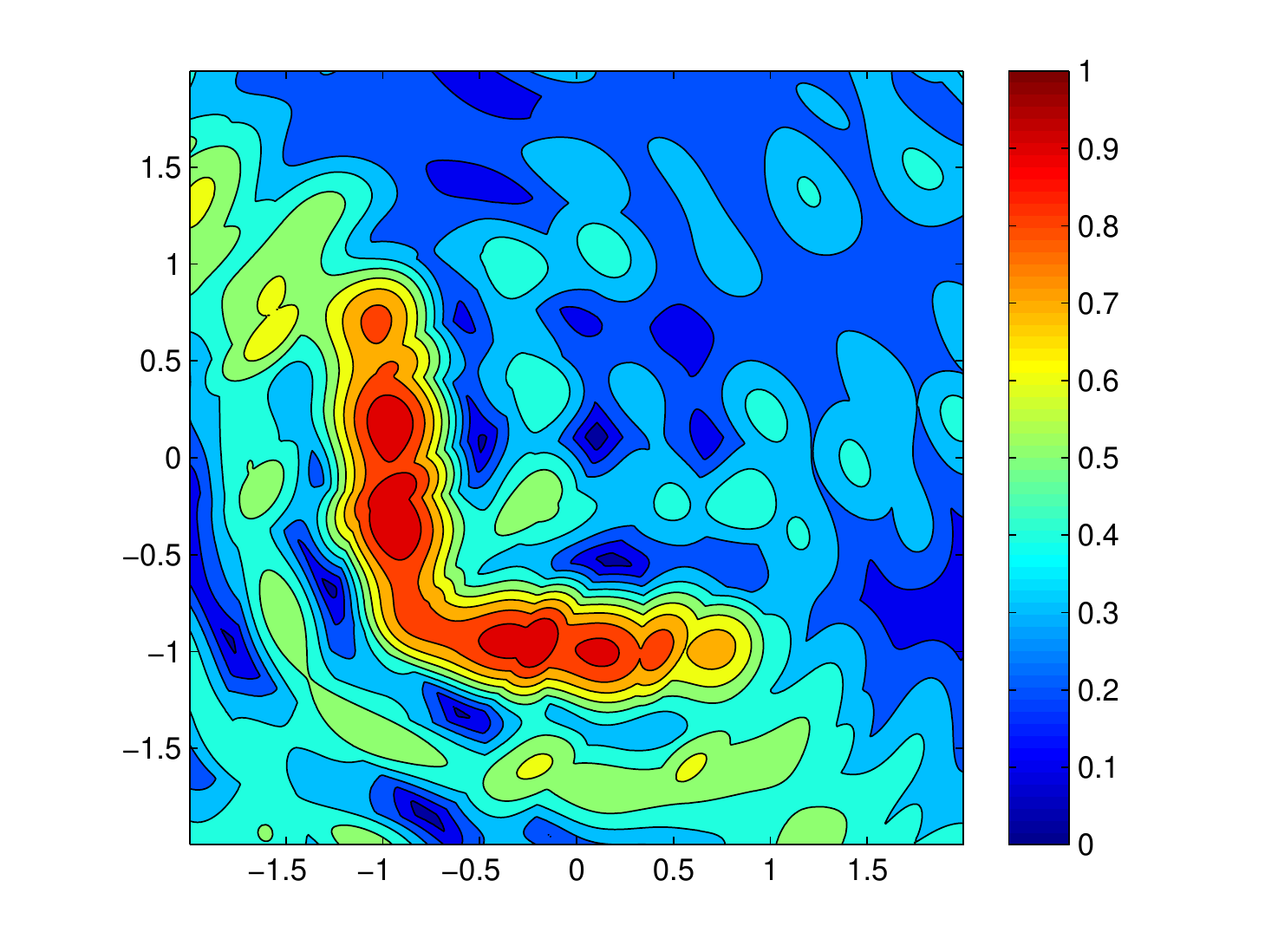}\hfill{}
\includegraphics[width=0.24\textwidth]{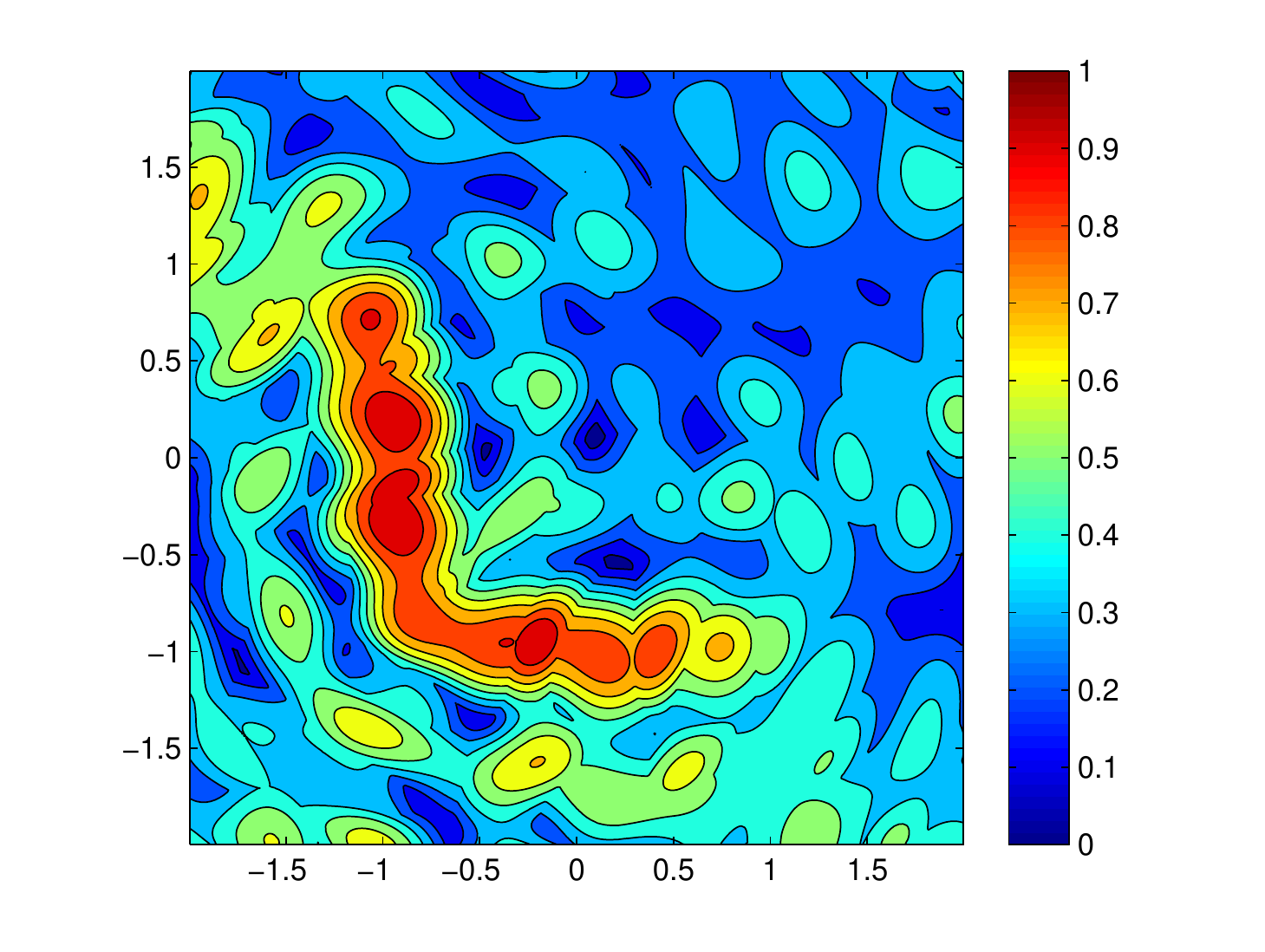}\hfill{}\hfill{}

\hfill{}DSM(f) using two incident waves: from left to right, noise
level is $0$, $5\%$, $10\%$, $20\%$.\hfill{}

\caption{\label{fig:ex7-3}Numerical results of Example 7 using two
incident directions $d_{1}$ and $d_{2}$ .}
\end{figure}

\section{Concluding remarks\label{sec:Conclusion}}

A systematic evaluation of the performance of the DSM using
both far-field data and near-field data has been carried out for some
inverse scattering benchmark problems, such as scatterers of obstacle,
medium and crack types. The numerical simulations have demonstrated
the robustness and effectiveness of the DSM to identify the locations
of obstacles, inhomogeneity media, as well as thin cracks, with just
one or two incident directions. An important advantage of
the DSM is its tolerance against large noise in the observation data.

The DSM method involves only simple operations like inner products,
without any inversion or solution process required, so it is computationally
very cheap, hence not intended for an accurate reconstruction of unknown
scatterers and their physical inhomogeneity.
As it is cheap and works with the data from even one incident
field, it can serve as a fast, simple and effective numerical tool
for locating reliable approximate positions of unknown scatterers.
The reconstructed shapes of scatterers and the final numerical indicator function
values by DSM may then serve as a good initial guess for any other more advanced
method to achieve a more accurate estimate of the scatterer shapes
and the inhomogeneity distribution. In addition, with the reliable location of
each individual scatterer component extracted by the DSM, one may
start with a much smaller sampling region in a more accurate but computationally
much more demanding method (see, e.g., \cite{BaL05,CoK98,Hoh01,IJZ12r}).
Considering the severe ill-posedness of the inverse scattering problems,
the reduction of the sizes of the initial sampling regions for unknown
scatterers may save us an essential fraction of the computational
efforts in the entire reconstruction process.

\section*{Acknowledgements}
The authors would like to thank two anonymous referees for very instructive and helpful comments
and suggestions, and to thank
Dr. Hongyu Liu (University of North Carolina at Charlotte)
and Dr. Hongpeng Sun (AMSS, Chinese Academy of Sciences) for many inspiring and valuable
discussions.

Received September 2012; revised November 2012.\\
 \indent {\it E-mail address: }li.jz@sustc.edu.cn\\
  \indent{\it E-mail address: }zou@math.cuhk.edu.hk \\
\end{document}